\documentclass[preprint]{elsarticle}

\pdfoutput=1

\newcommand{\onlyPaper}[1]{}
\newcommand{\onlyReport}[1]{#1}

\newcommand{\onlyDvi}[1]{}
\newcommand{\onlyPdf}[1]{#1}

\usepackage{xcolor}
\usepackage{amsmath}
\usepackage{bbm}
\usepackage{theorem}
\usepackage{enumerate}
\usepackage{mathabx}
\usepackage{prooftree}
\usepackage[all]{xy}
\usepackage{alltt}    %% For capf-neededness.tex
\usepackage[normalem]{ulem} %% strikeout 
\usepackage{url}
\usepackage{wrapfig}
\usepackage{float}
\usepackage{MnSymbol}
\usepackage{tikz}
\usetikzlibrary{matrix,arrows,positioning}

\usepackage{graphicx}

\newcommand{\arsTerms}{\mathcal{O}}
\newcommand{\arsRedexes}{\mathcal{R}}
\newcommand{\arsSource}{\textsf{src}}
\newcommand{\arsTarget}{\textsf{tgt}}
\newcommand{\arsResidual}[1]{\ldbrack #1 \rdbrack}
\newcommand{\arsResidualRel}{\arsResidual{\cdot}}

\newcommand{\arsRedseq}{\mathcal{RS}}

\newcommand{\rsrc}[1]{\arsSource(#1)}
\newcommand{\rtgt}[1]{\arsTarget(#1)}

\newcommand{\arsa}{\mathfrak{A}}

\newcommand{\residu}[3]{#1 \arsResidual{#2} #3}
\newcommand{\residus}[2]{\residu{#1}{#2}{}}

\newcommand{\tredex}{step}
\newcommand{\tredexes}{steps}
\newcommand{\tredexset}{multistep}
\newcommand{\tredexsets}{multisteps}
\newcommand{\redseq}{reduction sequence}
\newcommand{\redseqs}{reduction sequences}
\newcommand{\mredseq}{multireduction}
\newcommand{\mredseqs}{multireductions}

\newcommand{\setp}[1]{\ensuremath{\mathcal{#1}}}
\newcommand{\setpa}{\setp{A}}
\newcommand{\setpb}{\setp{B}}

\newcommand{\setreda}{\setpa}
\newcommand{\setredb}{\setpb}
\newcommand{\setredc}{\setp{C}}
\newcommand{\setredd}{\setp{D}}

\newcommand{\stepa}{\ensuremath{a}}
\newcommand{\stepb}{\ensuremath{b}}
\newcommand{\stepc}{\ensuremath{c}}

\newcommand{\ppcstepa}{\ensuremath{\ppcstep{a}}}
\newcommand{\ppcstepb}{\ensuremath{\ppcstep{b}}}
\newcommand{\ppcstepc}{\ensuremath{\ppcstep{c}}}

\newcommand{\red}[1]{\ensuremath{\underset{#1}{\rightarrow}}}

\newcommand{\redbf}[1]{\rightarrow}

\newcommand{\emptyred}[1]{\nil_{#1}}

\newcommand{\thefree}[1]{#1^F}
\newcommand{\thedomin}[1]{#1^E}
\newcommand{\thefreewrt}[2]{#1^F_{#2}}
\newcommand{\thedominwrt}[2]{#1^E_{#2}}
\newcommand{\freefrom}{\not\sqsupset}
\newcommand{\dominated}{>}
\newcommand{\Domination}{Embedding}
\newcommand{\domination}{embedding}
\newcommand{\domintext}{embedded}
\newcommand{\dominbytext}{embedded by}

\newcommand{\superfree}{independent}

\newcommand{\superfreeness}{independence}
\newcommand{\Superfreeness}{Independence}

\newcommand{\grip}{\ll}
\newcommand{\notgrip}{\not\ll}
\newcommand{\ngrip}{never\mbox{-}gripping}

\newcommand{\tRecatch}{recatching}
\newcommand{\tRecatchStr}{recatching strategy}
\newcommand{\tRecatchStrs}{recatching strategies}

\newcommand{\rstrat}[1]{\ensuremath{\mathcal{#1}}}
\newcommand{\strs}{\rstrat{S}}
\newcommand{\strspos}{\rstrat{S}_{\!\pi}}
\newcommand{\strsp}{\ensuremath{\rstrat{S}{\rstrat{M}}}}
\newcommand{\strspth}{\ensuremath{\strsp_\theta}}

\newcommand{\fnsmth}[2]{\strspth({#2},{#1})}
\newcommand{\strt}{\rstrat{T}}
\newcommand{\emptypair}{\pair{\emptyset}{\emptyset}}

\newcommand{\axFda}{Grip--Instantiation}
\newcommand{\axFdb}{Grip--Density}
\newcommand{\axFdc}{Grip--Convexity}
\newcommand{\axSelfReduction}{Self Reduction}
\newcommand{\axFiniteResiduals}{Finite Residuals}
\newcommand{\axAncestorUniqueness}{Ancestor Uniqueness}
\newcommand{\axFD}{FD}
\newcommand{\axSO}{SO}
\newcommand{\axLinearity}{Linearity}
\newcommand{\axCtxFreeness}{Context-Freeness}
\newcommand{\axEnclaveCreation}{Enclave--Creation}
\newcommand{\axEnclaveEmbedding}{Enclave--Embedding}
\newcommand{\axMisterious}{Pivot}
\newcommand{\axStability}{Stability}

\newcommand{\axiomuse}[1]{\textsf{#1}}

\newcommand{\wid}[1]{\ensuremath{\widehat{#1}}}
\newcommand{\lp}[1]{\l_{#1}}
\newcommand{\lpth}{\lp{\theta}}

\newcommand{\lpset}[1]{\l_{\set{#1}}}

\newcommand{\ppcstep}[1]{\mathfrak{#1}}
\newcommand{\ppcsetred}[1]{\mathfrak{#1}}

\newcommand{\rewto}{\rightarrow}

\newcommand{\kor}{{\tt or}}

\newcommand{\kt}{{\tt tt}}

\newcommand{\kpr}{{\tt p}}
\newcommand{\ka}{{\tt a}}
\newcommand{\kb}{{\tt b}}
\newcommand{\kc}{{\tt c}}
\newcommand{\kd}{{\tt d}}

\newcommand{\kid}{I}
\newcommand{\ca}{{\tt a}}
\newcommand{\cb}{{\tt b}}
\newcommand{\cc}{{\tt c}}
\newcommand{\cd}{{\tt d}}
\newcommand{\ce}{{\tt e}}

\newcommand{\pair}[2]{\ensuremath{\langle {#1},{#2}\rangle}}
\renewcommand{\l}{\lambda}
\newcommand{\fv}[1]{{\tt fv}(#1)}
\newcommand{\fm}[1]{{\tt fm}(#1)}

\newcommand{\Pos}[1]{{\tt Pos}(#1)}

\newcommand{\ROccur}[1]{\ensuremath{{\mathcal Red}(#1)}}

\newcommand{\repl}[3]{\ensuremath{ #1 [#2]_{#3} }}

\newcommand{\tsize}[1]{|#1|}
\newcommand{\rlength}[1]{| #1 |}
\newcommand{\isbm}[2]{bm(#1,#2)}

\newcommand{\subtat}[2]{{#1} \hspace{-1mm}\mid_{#2}\hspace{0.3mm}}
\newcommand{\subtatwide}[2]{{#1} \mid_{#2}}
\newcommand{\subtatnarrow}[2]{{#1} \hspace{-1mm}\mid_{#2}\hspace{-0.3mm}}
\newcommand{\disj}{\parallel}

\newcommand{\subst}[2]{\ensuremath{ \{ #1 \rightarrow #2 \} }}

\newcommand{\emptySubst}{\ensuremath{ \{ \} }} 

\newcommand{\dom}{{\tt dom}}
\newcommand{\set}[1]{\{ #1 \}}

\newcommand{\matchOpt}[3]
  {\ensuremath{\{  #3 /_{#1}\   #2  \}}}
\newcommand{\matchOpth}[2]{\matchOpt{\theta}{#1}{#2}}
\newcommand{\cmatchOp}[3]
   {\ensuremath{\{ \! \! \{  #3 \triangleright_{#1}  #2  \} \! \! \}}}
\newcommand{\cmatchOpth}[2]{\cmatchOp{\theta}{#1}{#2}}

\newcommand{\fail}{\textnormal{\texttt{fail}}}
\newcommand{\wait}{\textnormal{\texttt{wait}}}

\definecolor{OliveGreen}{rgb}{0.439608 0.656863 0.437255}
\definecolor{DarkGreen}{rgb}{0 0.5 0}
\definecolor{OtroGreen}{rgb}{0.30 0.76 0.50}
\definecolor{violetin}{rgb}{0.40 0.15 1.00}
\definecolor{musgo}{rgb}{0.4 0.8 0.76}
\definecolor{requeteBlue}{rgb}{0,0,0.5}

\definecolor{webred}{rgb}{0.75,0,0}

\newcommand{\delia}[1]{#1}
\newcommand{\cdelia}[2]{#2}

\newcommand{\carlos}[1]{{\color{red} #1}}
\newcommand{\edu}[1]{{\color{DarkGreen} #1}}

\newcommand{\ale}[1]{{\color{magenta} #1}}

\newcommand{\theppc}{\texttt{\textbf{PPC}}}

\newcommand{\ih}{\textit{i.h.}}

\newcommand{\Nat}{\mathbbmss{N}}
\newcommand{\ems}{\emptyset}
\newcommand{\sm}{\setminus}
\newcommand{\sym}[1]{{\tt sym}(#1)}
\newcommand{\partsOf}[1]{\mathcal{P}(#1)}

\newcommand{\setName}[1]{\ensuremath{\textit{\textbf{#1}}}}

\newcommand{\DStructs}{\setName{DS}}
\newcommand{\Abstract}{\setName{ABS}}
\newcommand{\MForms}{\setName{MF}}

\newcommand{\capfterms}{\setName{T}}

\newcommand{\NForms}{\setName{NF}}  % normal forms
\newcommand{\eqdef}{\ensuremath{:=}}

\newcommand{\textif}{\textnormal{ if }}

\newcommand{\textand}{\textnormal{ and }}

\newcommand{\metaexists}{\ensuremath{\pmb{\mathsf{\exists}}}}
\newcommand{\metaforall}{\ensuremath{\pmb{\forall}}}

\newcommand{\sthat}{\textrm{ s.t. }}

\newcommand{\wrt}{w.r.t.}

\newcommand{\minicenter}[1]{\hspace*{\stretch{1}} {#1} \hspace*{\stretch{1}}}

\newcommand{\minitem}{\vspace{-4pt}\item}

\newcommand{\conceptIntro}[1]{\textbf{#1}}

\newtheorem{theorem}{Theorem}[section]
\newtheorem{lemma}[theorem]{Lemma}
\newtheorem{definition}[theorem]{Definition}
\newtheorem{proposition}[theorem]{Proposition}

\newtheorem{notation}[theorem]{Notation}
\newenvironment{proof}[1][Proof]{\begin{trivlist}
\item[\hskip \labelsep {\bfseries #1}]}{\hspace{\stretch{1}}$\blacksquare$\end{trivlist}}
\newcommand{\thelem}{Lem.}
\newcommand{\theLem}{Lem.}
\newcommand{\thethm}{Thm.}
\newcommand{\theprop}{Prop.}
\newcommand{\thesec}{Sec.}

\newcommand{\thefig}{Fig.}

\newcommand{\ignore}[1]{}
\newcommand{\subterm}{\ensuremath{\subseteq}}
\newcommand{\sig}{\sigma}
\newcommand{\avoids}{\ensuremath{\#}}

\newcommand{\longtwoheadrightarrow}{\relbar\joinrel\twoheadrightarrow}

\newcommand{\mreda}{\ensuremath{\Delta}}

\newcommand{\mredb}{\ensuremath{\Gamma}}
\newcommand{\mredc}{\ensuremath{\Pi}}
\newcommand{\mredd}{\ensuremath{\Psi}}

\newcommand{\depth}[1]{\nu(#1)}
\newcommand{\reda}{\ensuremath{\delta}}
\newcommand{\redb}{\ensuremath{\gamma}}
\newcommand{\redc}{\ensuremath{\xi}}

\newcommand{\sstep}[1]{\overset{#1} \rightarrow}

\newcommand{\sstepdec}{\longrightarrow}
\newcommand{\sstepx}[1]{\overset{#1}{\sstepdec}} 

\newcommand{\sred}[1]{\overset{#1} \twoheadrightarrow}

\newcommand{\sredx}[1]{\overset{#1}{\longtwoheadrightarrow}}

\newcommand{\redel}[2]{#1 [#2]}

\newcommand{\mstep}[1]{ %
\mathbin{
{\ooalign{\hfil$\overset{#1}{\longrightarrow}$\hfil\cr\hfil$\circ$\hfil\crcr}}}}

\newcommand{\mred}[1]{ %
\mathbin{{\ooalign{\hfil$\overset{#1}{\longtwoheadrightarrow}
$\hfil\cr\hfil$\circ$\hfil\crcr } } } }

\newcommand{\mredel}[2]{\ensuremath{#1[#2]}} 
\newcommand{\mredsub}[3]{\ensuremath{#1[#2..#3]}} 
\newcommand{\mredunt}[2]{\mredsub{#1}{1}{#2}}

\newcommand{\mredmsf}{\chi}
\newcommand{\mredms}[1]{\mredmsf({#1})}

\newcommand{\ltms}{<_{\scriptsize{lex}}}
\newcommand{\leqms}{\leq_{\scriptsize{lex}}}

\newcommand{\proj}[2]{#1 \, |_{#2}}

\newcommand{\defi}[1]{\textbf{#1}}

\newcommand{\ie}{{\it i.e.}}

\newcommand{\nil}{{\tt nil}}

\definecolor{SuperLightGray}{gray}{.87}
\definecolor{LightGray}{gray}{.9}
\definecolor{MediumGray}{gray}{.68}
\definecolor{DarkGray}{gray}{.50}
\definecolor{Lightbleu}{cmyk}{.3,.1,0,0}

\renewcommand{\l}{\ensuremath{\lambda}}

\newcommand{\arMulti}[1]{\ar[#1]|-*=0@{o}}
\newcommand{\arMultiOp}[2]{\ar@{#2}[#1]|-*=0@{o}}

  {\begin{enumerate}%
    \setlength{\itemsep}{0pt}%
    \setlength{\parskip}{0pt}%
    \setlength{\topsep}{0pt}
    \setlength{\partopsep}{0pt}}
  {\end{enumerate}}

  {\begin{itemize}%
    \setlength{\itemsep}{0pt}%
    \setlength{\parskip}{0pt}%
    \setlength{\topsep}{0pt}
    \setlength{\partopsep}{0pt}}
  {\end{itemize}}

\newcommand{\Id}{I}
\newcommand{\Om}{\Omega}

\newcommand{\lc}{$\l$-calculus}
\newcommand{\setdisjoint}{pointwise disjoint}
\newcommand{\protoredex}{prestep}

\newcommand{\confer}{{\it cf.}}
\newcommand{\Confer}{{\it Cf.}}
\newcommand{\eg}{{\it e.g.}}
\newcommand{\Eg}{{\it E.g.}}

\newcommand{\thespc}{\texttt{\textbf{SPC}}}
\newcommand{\spcterms}{\setName{T}}
\newcommand{\spcMatchFail}{\ensuremath{\mathfrak{f}}}

\newcommand{\develops}{\ensuremath{\Vdash}}
\newcommand{\rng}{{\tt rng}}
\newcommand{\cng}{{\tt cng}}

\newenvironment{deliae}{\par}{\par}

\newcommand{\CompleteInReport}{See~\cite{BKLR:Long:2014} for further details}

\journal{Theoretical Computer Science}

\begin{document}

\begin{frontmatter}



\title{On abstract normalisation beyond neededness}

\author[1,2]{Eduardo Bonelli}
\ead{eabonelli@gmail.com}
\author[3]{Delia Kesner}
\ead{Delia.Kesner@pps.univ-paris-diderot.fr}
\author[1]{Carlos Lombardi}%
%\fnref{hola}\onlyPaper{\corref{cor1}}}
\ead{clombardi@unq.edu.ar}
\author[4]{Alejandro R\'ios}
\ead{rios@dc.uba.ar}
%\fntext[hola]{Permanent address: Universidad Nacional de Quilmes.}
\onlyPaper{\cortext[cor1]{Corresponding author}}
\address[1]{Univ.  Nac. de Quilmes. \\ Roque S\'aenz Pe\~na 352 (1876), Bernal, Prov. de Buenos Aires, Argentina}
\address[2]{CONICET. \\ Av. Rivadavia 1917 (1033) C.A.Buenos Aires, Argentina}
\address[3]{Univ. Paris-Diderot, SPC, PPS, CNRS \\ Case 7014 
75205 PARIS Cedex 13, France}
\address[4]{Univ. de Buenos Aires \\ Pabell\'on I, Ciudad Universitaria (1428) C.A.Buenos Aires, Argentina}

\begin{abstract}
  We study normalisation of multistep strategies, strategies that
  reduce a set of redexes at a time, focussing on the notion of
  \emph{necessary sets}, those which contain at least one redex that
  cannot be avoided in order to reach a normal form. This is
  particularly appealing in the setting of non-sequential rewrite
  systems, in which terms that are not in normal form may not have any
  \emph{needed} redex. We first prove a normalisation theorem for
  abstract rewrite systems (ARS), a general rewriting framework
  encompassing many rewriting systems developed by
  P-A.Melli\`es~\cite{thesis-mellies}. The theorem states that
  multistep strategies reducing so called \emph{necessary} and
  \emph{\ngrip} sets of redexes at a time are normalising in any
  ARS. Gripping refers to an abstract property reflecting the behavior
  of higher-order substitution. We then apply this result to the
  particular case of \theppc, a calculus of patterns and to the
  lambda-calculus with parallel-or.
\end{abstract}


\end{frontmatter}

%%% Local Variables:
%%% mode: latex
%%% TeX-master: "article"
%%% End:

\tableofcontents

\section{Introduction}
\label{sec:intro}
  
This paper is about computing normal forms in rewrite systems. Consider the \l-calculus. Let
$K$ stand for the term $\l x.\l y.x$, $I$ for $\l x.x$ and $\Omega$ for $(\l x. x\,x)\,(\l x. x\, x)$. Then
$s\eqdef K\,I\,\Omega$ admits an infinite reduction sequence of $\beta$-steps, namely the one obtained by
repeatedly reducing $\Omega$, and hence $s$, to itself. However, it also reduces in two $\beta$-steps to the
normal form $I$ by repeatedly reducing the leftmost-outermost redex:
\begin{equation}
\vspace{-1mm}
K\,I\,\Omega \red{\beta} (\l y.I)\, \Omega\red{\beta} I
\label{eq:ex:LOnormalisation:lambda}
\end{equation}
\delia{The reason this strategy normalises is that the redexes it
  selects are unavoidable or \emph{needed} in any reduction to normal
  form.  Indeed, leftmost-outermost redexes are needed in
  \lc~\cite{barendregt}.  This paper studies normalisation for the
  broader case of rewriting systems where needed redexes may not
  exist.  It does so by adapting Melli\`es' abstract rewriting
  framework~\cite{thesis-mellies} to encompass Sekar and
  Ramakrishnan's notion of \emph{needed sets of
    redexes}~\cite{sekar-rama}. In doing so, the relatively unfamiliar
  notion of \emph{gripping}, used only marginally in the work of
  Melli\`es, is shown to play a crucial
  r\^ole, thus giving it an interest of its own.  }

\textbf{Normalisation in TRS.}
Although in the \l-calculus the leftmost-outermost strategy does indeed attain a normal
form (if it exists)~\cite{curry-feys}, the same cannot be said for term rewriting
systems (TRS). For example, consider the TRS:
\begin{center}
$\begin{array}{rcl}
a & \to & b\\
c & \to & c\\
f(x,b) & \to & d
\end{array}$
\end{center}
%ACHICAR
%and the term $t\eqdef f(c,a)$. The leftmost-outermost strategy will select redex $c$ in $t$ producing an infinite reduction sequence. Yet this term admits a normal form:
and the term $t\eqdef f(c,a)$. The leftmost-outermost strategy selects redex $c$ in $t$ producing an infinite reduction sequence. Yet this term admits a normal form:
\begin{equation}
f(c,a) \red{} f(c,b) \red{} d
\label{eq:ex:LOnormalisation:TRS}
\vspace{-2mm}
\end{equation}
For \emph{left-normal} TRS (those in which variables do not precede
function symbols in the left-hand side of rewrite rules), the
leftmost-outermost strategy does indeed
normalise~\cite{ODonnell:1977}; the same applies to left-normal
higher-order rewrite systems~\cite{thesis-klop}. 
Alternatively, one might decide to 
reduce all \emph{outermost} redexes at once:
parallel-outermost reduction is normalising for (almost) orthogonal
TRS~\cite{ODonnell:1977}, \delia{where an \emph{orthogonal} TRS 
is one whose rewrite rules are left-linear and non-overlapping, and an \emph{almost orthogonal} TRS has trivial critical pairs at the root, if at all. 
Parallel-outermost reduction is also normalising for almost orthogonal systems in} higher-order
rewriting~\cite{thesis-vanRaamsdonk}. 

\textbf{Needed Redexes.} However, there is a deeper connection between
redexes reduced in (\ref{eq:ex:LOnormalisation:lambda}) and
(\ref{eq:ex:LOnormalisation:TRS}). They are unavoidable in the sense
that in any reduction sequence from $s$ and $t$ to normal form, they
(or their residuals) must be reduced at some point. Such redexes are
called \emph{needed} and a theory of needed redexes was developed by
Huet and L\'evy~\cite{HL:1991ab} for orthogonal TRS.
% \ale{where an \emph{orthogonal} TRS 
% is one whose
% rewrite rules are left-linear and non-overlapping;
% an \emph{almost orthogonal} TRS 
% has trivial critical pairs at the root, if at all.}
In~\cite{HL:1991ab} it is shown that in these TRS, terms that are not in
normal form always have at least one needed redex and that reduction
of needed redexes is normalising. They also showed that determining
whether a redex is needed or not is undecidable in general; this led
them to study restrictions of the class of orthogonal TRS (the
\emph{strongly sequential} TRS) in which needed redexes could be
identified effectively.

\ignore{
Further directions
that have been pursued in an effort to enhance the theory developed by Huet and L\'evy fall roughly
into one of the following categories:
\begin{enumerate}
\item enlarging the
class of orthogonal TRS for which decidability of needed redexes can be obtained~\cite{DBLP:conf/lics/Toyama92,DBLP:journals/siamcomp/Oyamaguchi93};

\item  enlarging the class of systems for which a theory of needed reduction makes
sense (with or without considerations on effective
neededness)~\cite{BKKS:1987,sekar-rama,thesis-mellies,GKZ00}; and \edu{ver referencias m'as recientes}

\item extending the concept of neededness w.r.t. a normal form to other more general
  notions~\cite{DBLP:conf/popl/Middeldorp97,GKZ00}.
\end{enumerate}
 }

A fundamental limitation of the above mentioned theory of neededness
is the requirement of orthogonality. This requirement does have its
reasons though: in non-orthogonal TRS, terms that are not in normal
form may not have needed redexes. A paradigmatic example is the
  ``parallel-or'' TRS:
\begin{deliae}
  \begin{center}
$\begin{array}{rcl}
	\kor(x,\kt) & \to & \kt \\
	\kor(\kt,x) & \to & \kt
\end{array}$
\end{center} 
The term $u\eqdef \kor(\kor(\kt,\kt),\kor(\kt,\kt))$ has four redexes:
the occurrence of $\kor(\kt,\kt)$ on the left is an instance of the
first and second rules of the parallel-or TRS, and the one on the
right is also an instance of both of these rules. None of these is
needed since one can always obtain a normal form without reducing
it. For example, the reduction sequence:
\begin{equation}
\kor(\kor(\kt,\kt),\kor(\kt,\kt)) \red{} \kor(\kor(\kt,\kt),\kt) \red{} \kt
\label{eq:ex:POr: reductionSequence}
\end{equation}
never reduces any of the two redexes on the left. A similar argument
applies to the two redexes on the right in $u$. In fact $u$ seems to
suggest that no sensible normalising \emph{strategy}, picking one
redex at a time by looking solely at the term, can be constructed.
A similar phenomenon occurs even in orthogonal TRSs, a
  paradigmatic example being Gustave's TRS~\cite{gustave}.
\end{deliae}
Any almost orthogonal TRS\footnote{In fact, any almost orthogonal Combinatory Reduction
  Systems~\cite{thesis-klop}.}, such as the parallel-\kor\ example above, does
admit a normalising one-step reduction
strategy~\cite{DBLP:journals/apal/Kennaway89,AntoyMiddeldorp:1996}. There
is a price to pay though, namely that such a strategy has to perform
lookahead (in the form of cycle detection within terms of a given
size).

Another example of the absence of needed redexes in non-orthogonal rewrite systems are
\emph{pattern calculi}. Let \kpr\ be a data constructor representing a person including her/his
name, gender and marital status. For example, $\kpr\,
\mathtt{j}\,{\tt m}\, {\tt s}$ represents the person named \texttt{j} (for ``Jack'') who is male and single. A function
such as $\l \kpr\, x\,{\tt m}\, {\tt s}. x$ returns the name of any person that is male and
single. It computes by matching the \emph{pattern} $\kpr\, x\,{\tt m}\, {\tt s} $ against its
argument: reporting an appropriate substitution, if it is successful, or a distinguished constant
\spcMatchFail,  if it fails (\confer\ Sec.~\ref{sec:spc}).  Consider the following term which results from applying the 
abovementioned function to a person
called \texttt{a} (for ``Alice'') that is female and divorced (recall from above that $I$ is the identity function $\l x.x$):
\begin{equation}
t_0 := (\l \kpr\, x\,{\tt m}\, {\tt s}. x) (\kpr\, \mathtt{a}\,(I {\tt f}) (I {\tt d}))
\label{eq:ex:termWithNoNeededRedexes}
\end{equation}
This term has two redexes, namely $I {\tt f}$ and $I {\tt d}$. Note that the term itself is not a
redex since the success or failure of the match between pattern and argument cannot be
determined. We have two possible reduction sequences to normal form:
\begin{center}
$\begin{array}{cccccc}
(\l \kpr\, x\,{\tt m}\, {\tt s}. x) (\kpr\, \texttt{a}\,(I {\tt f}) (I {\tt d})) 
& \red{} 
& (\l \kpr\,
x\,{\tt m}\, {\tt s}. x) (\kpr\, \texttt{a}\,(I {\tt f})\, {\tt d}) 
& \red{}
& \spcMatchFail
\\
(\l \kpr\, x\,{\tt m}\, {\tt s}. x) (\kpr\, \texttt{a}\,(I {\tt f}) (I {\tt d})) 
& \red{} 
& (\l \kpr\,
x\,{\tt m}\, {\tt s}. x) (\kpr\, \texttt{a}\,{\tt f} (I\, {\tt d})) 
& \red{}
& \spcMatchFail
\end{array}$
\end{center}
The first reduction sequence does not reduce $I {\tt f}$; the second does not reduce $I {\tt
  d}$. Therefore the term $t_0$ does not contain any needed redexes.

\textbf{Beyond Neededness.} This prompts one to consider whether it is possible to obtain
normalisation results for possibly overlapping, and more generally \emph{non sequential} (\cite{HL:1991ab}) rewrite systems. The following avenues
have been pursued in this direction:
\begin{enumerate}
\item Boudol~\cite{Boudol:1985} studies the reduction space of possibly non-orthogonal TRS and
  defines needed reduction for these systems. % \edu{No puedo localizar copia online de esto para comentarlo}

\item Melli\`es~\cite{thesis-mellies} extends the notion of needed redex to that of a needed
  derivation (actually \emph{external} derivation, a generalization of the notion of neededness).

\item van Oostrom~\cite{vanOostrom:1999} proves that outermost-fair reduction is normalising for
  weakly orthogonal fully-extended higher-order pattern rewrite systems (PRS). An outermost fair
  strategy is one in which no outermost redex is ignored (\ie\ not contracted) indefinitely.

\item Sekar and Ramakrishnan~\cite{sekar-rama} extend the notion of a needed redex to that of a \emph{set} of redexes, called a \emph{necessary set}, in the context of first-order rewriting.

\end{enumerate}

The results of the first item above are restricted to first-order rewriting and hence are not
applicable to our pattern calculus example. The second item above suffers from two problems. The
first is that it requires the calculus to verify a property (among others) called \emph{stability}
which fails for some pattern calculi such as the one of our example
(cf. Sec.~\ref{sec:embedding}). Also, it does not seem obvious how to implement the proposed
strategies. For example, in the case of $\kor(\kor(\kt,\kt),\kor(\kt,\kt))$, although there are no
needed redexes, \cite{thesis-mellies} declares the reduction sequence (\ref{eq:ex:POr:
  reductionSequence}) \emph{itself} to be external. It then goes on to show that composition of
these external reduction sequences are normalising. So in order to normalise a term one would have
to identify such reduction sequences. In~\cite{vanOostrom:1999} a number of normalisation
results are proved for PRS, the most relevant being that outermost-fair strategies are
normalising for weakly orthogonal PRS.  There are a number of notable differences with our work however. The
fundamental aspect that sets our paper apart from~\cite{vanOostrom:1999} is the axiomatic
development that we pursue. In~\cite{vanOostrom:1999}, the crucial notions of \emph{contribution}
and \emph{copying} rely heavily on \emph{positions} since it is terms that are rewritten. In
contrast, we propose a number of axioms that are assumed to hold over ``objects'' and ``steps''
whose compliance guarantees normalisation. The nature of the objects that are rewritten is irrelevant.  

%% \eduX{In particular, graph rewriting  or any other notion of rewriting that fits the axioms
%% are covered}{Justo los grafos no se llevan bien con los axiomas;
%%   propongo eliminar la oraci\'on.}. \eduso{Second, \cite{vanOostrom:1999} does not address needness (see below) for
%%   non-orthogonal rewrite systems. Finally, the notions of contribution and copying mentioned above
%%   are global in nature in the sense that the analysis performed in~\cite{vanOostrom:1999} relies on
%%   determining positions in terms that contribute to a result in a given derivation. Our
%%   corresponding notion of \emph{gripping} (\confer\ Sec.~\ref{sec:gripping}) is local in nature
%%   since it relies on an analysis of the steps in a given object.} \eduso{\textbf{REVER: Copying y
%%     contribution son locales, los que es global es la noci'on de essential development. Nuestro
%%     concpeto sustituto para eso es non-gripping, que tambi'en es global.......}}

We now focus our attention on~\cite{sekar-rama} mentioned above, the starting point of this paper.
As mentioned, terms such as $(\l \kpr\, x\,{\tt m}\, {\tt s}. x) (\kpr\, \texttt{a}\,(I {\tt f}) (I {\tt d}))$ do not contain needed redexes. However, at least one of the two redexes in each of those terms will need to be reduced in order to obtain a normal form. 
We thus declare the set$\{ I {\tt f}, I {\tt d}\}$ to be \emph{necessary} for this term. The intuition is that at
least one redex in a necessary set must be reduced in order to obtain a normal form, assuming that a
normal form exists. Of course, selecting all redexes in a term will indeed yield a necessary set;
the point is whether some given subset of the set of all redexes is a necessary one. These ideas
have been developed in~\cite{sekar-rama} for almost orthogonal TRS where it is shown that repeated contraction
of necessary sets of redexes is normalising. In this paper we extend the normalisation results for
necessary sets to the setting of abstract rewriting described by means of \emph{abstract
  rewrite systems} (ARS)~\cite{thesis-mellies}. This generalization encompasses the first-order
case, the higher-order case (in particular, pattern calculi such as the Pure Pattern Calculus --
\theppc~\cite{jk-ppc,jk-jfp}) or any other system that complies with the appropriate axioms.

\textbf{Towards an Abstract Proof of Normalisation.} In order to convey a more precise idea of
  the abstract nature of the setting in which we develop our proof, we provide a glimpse of
  \conceptIntro{abstract rewriting systems (ARS)}.  An ARS consists of a set $\arsTerms$ of 
  \emph{objects} that are rewritten, a set $\arsRedexes$ of rewriting \emph{steps} each having a
  corresponding source and target object, and the following three relations over rewriting steps:

\begin{center}
\begin{tabular}{|l|r@{\hspace{3pt}}c@{\hspace{3pt}}l|}
\hline
\emph{residual} relation & $\arsResidualRel$ & $\subseteq$ & $\arsRedexes\times \arsRedexes\times
\arsRedexes$\\
\hline
\emph{embedding} relation & $<$ & $\subseteq$ & $\arsRedexes\times \arsRedexes$ \\
\hline
\emph{gripping} relation & $\grip$ & $\subseteq$ & $\arsRedexes\times \arsRedexes$\\
\hline 
\end{tabular}
\end{center}

For instance, $\arsTerms$ could be the set of terms of our pattern
calculus example. A step would then be a pair consisting of a term and
a position such that the subterm at that position may be reduced. For
example, $\underline{(\l \kpr\, x\,{\tt m}\, {\tt s}. x) (\kpr\,
  \texttt{a}\, {\tt f}\, (I {\tt d}))}$, where we have used
underlining for denoting the position (the root position in this
case). The source object of this step is $(\l \kpr\, x\,{\tt m}\, {\tt
  s}. x) (\kpr\, \texttt{a}\, {\tt f}\, (I {\tt d}))$ and the target
$\spcMatchFail$. The \emph{residual relation} $\arsResidualRel$
relates to the tracing of steps. A triple $(a,b,a')\in
\arsResidualRel$, often written $\residu{a}{b}{a'}$, indicates that
after contracting step $b$, step $a$ becomes $a'$ (or, equivalently,
$a'$ is what is left of $a$). Here $a$ and $b$ are assumed to have the
same source.  For example, consider steps $c\eqdef \underline{ (\l
  \kpr\, x\,{\tt m}\, {\tt b}. I\,x) (\kpr\, (I\, \texttt{j})\,{\tt
    m}\,{\tt b})}$ and $d\eqdef (\l \kpr\, x\,{\tt m}\, {\tt
  b}. \underline{I\,x}) (\kpr\,(I\, \mathtt{j})\,{\tt m}\, {\tt b})$
and $d'\eqdef \underline{I\,(I\,\texttt{j})}$. Then
$\residu{d}{c}{d'}$. The \emph{embedding relation} allows steps with
the same source to be partially ordered. It is sometimes referred to
as ``nesting''.  For example, $\underline{(\l \kpr\, x\,{\tt m}\, {\tt
    s}. x) (\kpr\, \texttt{a}\,(I {\tt f}) (I {\tt d}))}$ embeds $(\l
\kpr\, x\,{\tt m}\, {\tt s}. x) (\kpr\, \texttt{a}\,(\underline{I {\tt
    f}}) (I {\tt d})) $ in the tree ordering given that the position
of the former is a prefix of the position of the latter. The
\emph{gripping relation} is an additional partial order on steps that
seeks to capture a typical property of higher-order rewrite systems in
which a reduction step $a$ may cause a step $b$ to be embedded inside
another one $c$. For this to happen, $a$ must embed $c$ and $b$.  In
addition, $c$ must have occurrences of variables that are to be
replaced by the substitution generated from a successful match arising
from the reduction of $a$. In this case we say $c$ grips $a$. For
example, $c\eqdef (\l x. \underline{ I x}) (I y)$ grips $a\eqdef
\underline{(\l x.I x) (I y)}$ since the former is embedded by the
latter and the former has a free occurrence of the bound variable $x$.
Note how reduction of $a$ would embed $b\eqdef (\l x. I x)\,
(\underline{I y})$ in the residual $c'\eqdef (\underline{I (I y)})$ of
$c$.

A number of \emph{axioms} on ARS shall be used to formulate a proof of normalisation of necessary
and \ngrip\  sets. These axioms verse over the above mentioned elements of an ARS and are drawn
from~\cite{thesis-mellies}, except for one of them which is new. They are developed in detail in
Sec.~\ref{sec:axiomsForARS}. Our abstract proof is then applied to concrete cases,
showing how one may obtain normalisation for \theppc\ and the $\l$-calculus with parallel-or.

\textbf{Contributions.}
The primary contributions may be summarized as follows:
 \begin{itemize}

 \item A gentle introduction to ARS and, in particular, to its
  axioms.
   
\item An abstract proof of normalisation that applies to possibly non-orthogonal systems.

\item A concrete normalisation strategy for a \emph{non-sequential} higher-order rewrite system, namely
  \theppc, and also for the $\l$-calculus with parallel-or.

\end{itemize}
Verification of compliance of a system with the axioms of an ARS, although in some cases tedious, provides valuable insight into its computation dynamics.  

This document supersedes~\cite{bklr-rta-2012} by reformulating the 
  normalisation technique, previously  specific
  to \theppc, into an axiomatic one (encompassed in Melli\`es' ARS),  introducing a new axiom along
  the way. It then shows how it may be applied not only to \theppc, but also to any other system
  satisfying the relevant axioms.

\textbf{Structure of the Paper.} We begin by introducing, in
Sec.~\ref{sec:spc}, a simple pattern calculus that shall serve as
study companion for the axiomatic development that
follows. Sec.~\ref{sec:ars} defines the axiomatic framework in which
we develop our results. The axioms themselves are presented in
Sec.~\ref{sec:axiomsForARS}. The concept of multireduction and
necessary multisteps are defined in
Sec.~\ref{sec:multireductionsInARS}. The axiomatic proof of
normalisation is elaborated in Sec.~\ref{sec:reduction-strategies}. 
We instantiate our axiomatic proof in Sec.~\ref{sec:twoCaseStudies}, 
 to obtain normalisation strategies for the Pure Pattern Calculus and for the
\l-calculus with parallel-or. Finally, we conclude and
suggest further avenues to pursue.

\onlyPaper{
\delia{For some results in Sec.~\ref{sec:twoCaseStudies}, we include
  only sketches of the corresponding proofs. An extended version of
  this work, including the full details of all \delia{the} proofs, can
  be found in~\cite{BKLR:Long:2014}.}
}

%%% Local Variables: 
%%% mode: latex
%%% TeX-master: "paper"
%%% End: 

\section{A Study Companion: the Simple Pattern Calculus}
\label{sec:spc}

The simple pattern calculus (\thespc), an extension of the lambda calculus, is presented for the
sole purpose of serving as our running example in order to illustrate the various notions we shall
be introducing. It is simple enough that we may be informal in our description below. Full
definitions are later supplied in Sec.~\ref{sec:ppc}, where the more general Pure Pattern Calculus
(\theppc), of which \thespc\ is just a fragment, is developed.

Terms (\spcterms) in \thespc\ are given by the following grammar:
\begin{center}
$\begin{array}{llccl}
          &  & t & :: = & x \mid \cc \mid tt \mid \l p.t  \\
\end{array}$
\end{center} 
where $x$ ranges over some set of term variables, \cc\ over some set
of constants, and $p$ ranges over a set of algebraic patterns. We
write $t_1\ldots t_n$ as an abbreviation for $((\ldots(t_1\,
t_2)\,\ldots) t_n)$. An \defi{algebraic pattern} is either a variable
$x$ or an expression of the form $\cc\, p_1\ldots p_n$, \eg\ $\kpr\,
x\,{\tt m}\, {\tt b}$.  The term $tu$ is called an \defi{application}
($t$ is the \defi{function} and $u$ the \defi{argument}) and $\l p.t$
an \defi{abstraction} ($p$ is the \defi{pattern} and $t$ is the
\defi{body}).  All variables in the body that also occur in the
pattern of an abstraction are said to be \emph{bound}. Application
(resp. abstraction) is left (resp. right) associative.  We consider
terms up to \defi{alpha-conversion}, \ie\ up to renaming of bound
variables. 
Positions in terms are extended to terms with patterns (\confer\ Sec.~\ref{sec:ppc-basic-elements}).
$\Pos{t}$ is the
set of positions of $t$; $<$ is the strict prefix relation over
positions; $\epsilon$ denotes the root position.  A term of the form
$\cc\,t_1\ldots t_n$ is called a \defi{data-structure}, \eg\ $\kpr\,
(I\, \texttt{j})\,{\tt m}\,{\tt b}$.

 The reduction semantics is given by the following rewrite rule:
\begin{center}
$
(\l p.s)\,t \red{} \cmatchOp{}{t}{p}(s)
$
\end{center}
\cmatchOp{}{t}{p} is the result of matching $t$ against $p$ and is called a \defi{match}. The
meaning of the expression $\cmatchOp{}{t}{p}(s)$ depends on this match.  The match can be
successful, in which case it denotes a substitution $\sigma$ and $\sigma(s)$ is thus the
application of the substitution to $s$. It can also be the special symbol $\fail$. The question here
is what does 
$\fail(s)$ denote? Following our introduction, it would be  the distinguished constant
$\spcMatchFail$. However, if $\spcMatchFail$ is produced it could block
subsequent computation unless some additional considerations on the behavior of terms such as
$\spcMatchFail\ t$ are taken. In order to encourage other patterns to be tested and avoid
overcomplicating the metatheory, it is natural to return
the identity function $I$ rather than \spcMatchFail. So we set $\fail(s)$ to denote 
the identity function $I$.
In any of these two cases, success or failure, we say that the match is
\defi{decided}. 
If it is not decided, in which case the match is the special symbol $\wait$, then the expression $(\l p.s)\,t$ is not a redex; \eg\ $(\l \cc.s)\, x$ or $(\l \cc.s)\, (\Id \cc)$. A match $\cmatchOp{}{t}{p}$, denoted $\mu$, is computed by applying the following equations in the order of appearance:
\vspace{-7mm}
\begin{center}
$\begin{array}{rcll}
\cmatchOp{}{t}{x}        & := & \subst{x}{t} & \\
\cmatchOp{}{\cc}{\cc}   & := & \emptySubst  & \\
\cmatchOp{}{\cc\,t_1\ldots t_n}{\cc\ p_1\ldots p_n}         & := & \cmatchOp{}{t_1}{p_1} \uplus
\ldots \uplus \cmatchOp{}{t_n}{p_n} 
			                       & n\geq 1 \\
\cmatchOp{}{\l q.t}{p}           & := & \fail
			                       & \\
\cmatchOp{}{t}{p}           & := & \fail
			                       & t \textrm{ data-structure} \\
\cmatchOp{}{t}{p}           & := & \wait
			                       & \textnormal{otherwise} 
\end{array}$
\end{center}

The use of disjoint union in the third clause of this definition restricts successful matching of
compound patterns to the linear ones~\footnote{A pattern $p$ is linear if it has at most one
  occurrence of any variable.}, which is necessary to guarantee
confluence~\cite{thesis-klop}. Indeed, disjoint union of two substitutions fails whenever their
domains are not disjoint.  Thus $\cmatchOp{}{\cc\,v\,w}{\cc\,x\, x}$ gives $\fail$. Other examples
are: $\cmatchOp{}{\cc\,(I \cd)}{\cc\,\cd}$ gives \wait, however $\cmatchOp{}{\cc\,(I\,
  \cd)}{\cd\,\cd}$ gives \fail.  Disjoint union of matches $\mu_1$ and $\mu_2$ is defined as: their
union if both $\mu_i$ are substitutions and $\dom(\mu_1) \cap \dom(\mu_2) = \ems$; \wait\ if either
of the $\mu_i$ is \wait\ and none is \fail; \fail\ otherwise.  Note that this definition of disjoint
union of matches validates the following equations:
\begin{center}
$\fail  \uplus \wait  = \wait
  \uplus\fail =  \fail$
\end{center}

These equations reflect the non-sequential
nature of reduction in \thespc. For example, in $\cmatchOp{}{\cc\,s\,t}{\cc\,\cd\,\ce}$ it is
  unclear whether we should pick $s$ or $t$ in order to obtain a decided match since either may not
  normalise while the other may help decide the match (towards \fail).

%%% Local Variables: 
%%% mode: latex
%%% TeX-master: "article"
%%% End: 

\section{Abstract Rewriting Systems}
\label{sec:ars}
This section revisits the definition of \emph{abstract rewriting systems} given in the introduction supplying further details and introduces the axioms that such systems must enjoy in order for the abstract proof of normalisation to be applicable.

\subsection{Basic components}
\label{sec:ars-intro}
Recall from the introduction that an ARS consists of a set of \emph{objects} $\arsTerms$ that are rewritten and a set of \emph{\tredexes}\footnote{Called \emph{redexes} (``radicaux'') in~\cite{thesis-mellies}, hence the reason why we use the letter $\arsRedexes$.}
$\arsRedexes$. 
Each step have a source and target object given by functions $\arsSource, \arsTarget
: \arsRedexes \to \arsTerms$. If $t\in\arsTerms$, then we write $\ROccur{t}$ for the set $\set{a \in
  \arsRedexes \sthat \rsrc{a} = t}$. Two \tredexes\ with the same source are said to be
\defi{coinitial}. We often write $t \sstep{a} u$ for a step $a$ s.t. $\rsrc{a} = t$ and $\rtgt{a} =
u$.

The following relations are given over steps:
\begin{itemize}
\item The \defi{residual} relation $\arsResidualRel \subseteq \arsRedexes\times \arsRedexes\times \arsRedexes$.

  This relation reflects how a step may be traced after some other \emph{coinitial} step is reduced. 
  Whenever $\residu{b}{a}{b'}$ we require $a$ and $b$ to be coinitial, and  $\rsrc{b'} =  \rtgt{a}$.
  When $\residu{b}{a}{b'}$ we say that $b'$ is a residual of $b$ after $a$. By
  $\residus{b}{a}$ we denote the set $\set{b' \sthat \residu{b}{a}{b'}}$ and similarly for
    $\residus{}{a}{b}$. 
Accordingly, we define $\residus{}{a}$ as the relation $\set{(b,b') \sthat \residu{b}{a}{b'}}$.
A step $b$ is said to be \defi{created} by a step $a$, with $\rsrc{b}=\rtgt{a}$, if $\residu{}{a}{b}=\emptyset$.

\item The \defi{embedding} relation $< \; \subseteq \arsRedexes\times \arsRedexes$.

  This relation allows coinitial steps to be strictly ordered\footnote{The embedding relation $<$ is assumed to be irreflexive and transitive.} by a well-founded relation%
		\footnote{Notice that if $\ROccur{t}$ is finite, then any relation on coinitial steps is necessarily well-founded.}%
	. 
	For each pair $a < b$, the steps $a$ and $b$ must be coinitial.  A \tredex\ $a$ is said to be \defi{outermost} iff there is no $b$ 
	such that
	$b < a$.  A step $a$ is \defi{disjoint} from $b$, written $a \disj b$, when $a$ and $b$ are coinitial,  $a \not< b$ and $b \not< a$.

\item The \defi{gripping} relation $\grip \; \subseteq \arsRedexes\times \arsRedexes$.

As mentioned, this additional strict order on steps
seeks to capture a typical property of higher-order rewrite systems in which a step $a$ may
affect two coinitial and disjoint steps by embedding one inside the other in the
target object of $a$.  Just like for embedding, for each pair $a\grip b$,
the steps $a$ and $b$ must be coinitial. 
\end{itemize}

An example of an ARS is the \thespc. Its objects $\arsTerms$ are just the terms \spcterms. 
A step is a pair consisting of a term and a position in the term s.t. the subterm at this position is
of the form $(\l p.s)u$, and $\cmatchOp{}{u}{p}$ is decided. For example, $a\eqdef \underline{(\l
  \kpr\, x\,{\tt m}\, {\tt b}. x) (\kpr\, \mathtt{a}\,{\tt f} (I {\tt d}))}$ is a step, where we
have underlined the relevant position. Then $\rsrc{a}$ is the term $(\l \kpr\, x\,{\tt m}\, {\tt
  b}. x) (\kpr\, \texttt{a}\,{\tt f} (I {\tt d}))$ and $\rtgt{a}$ is $I$ (since matching fails and
hence the identity function is produced). We could also write $(\l \kpr\, x\,{\tt m}\, {\tt b}. x)
(\kpr\, \texttt{a}\,{\tt f} (I {\tt d})) \sstep{a} I$. Also, $b\eqdef (\l \kpr\, x\,{\tt m}\, {\tt
  b}. x) (\kpr\, \mathtt{a}\,{\tt f} (\underline{I {\tt d}}))$ is a step. It has the same source as
$a$ but the target is $(\l \kpr\, x\,{\tt m}\, {\tt b}. x) (\kpr\, \mathtt{a}\,{\tt f}\, {\tt d})$.

For an example of steps related by the residual relation, consider the step $c\eqdef \underline{ (\l \kpr\,
  x\,{\tt m}\, {\tt b}. x) (\kpr\, (I\, \texttt{j})\,{\tt m}\,{\tt b})}$ and $d\eqdef (\l \kpr\,
x\,{\tt m}\, {\tt b}. x) (\kpr\,(\underline{I\, \mathtt{j}})\,{\tt m}\, {\tt b})$ and $d'\eqdef
\underline{I\,\texttt{j}}$. Then $\residu{d}{c}{d'}$. Steps may be erased by other
  steps. For example, $\underline{(\l \kpr\, x\, y\, z.{\tt c})\,(\kpr\, (I\, u)\,{\tt m}\,{\tt
      b})}$ erases the coinitial step $(\l \kpr\, x\,y\, z.{\tt c})\,(\kpr\, \underline{(I\, u)}\,
  {\tt m}\,{\tt b})$. It may also duplicate a coinitial step. For example, $\underline{(\l \kpr\,
    x\, y\, z.x\,x)\,(\kpr\, (I\, u)\,{\tt m}\,{\tt b})}$ duplicates $(\l \kpr\, x\, y\,
  z.x\,x)\,(\kpr\, \underline{(I\, u)}\,{\tt m}\, {\tt b})$ yielding two residuals 
 $\underline{(I\, u)}\, (I\, u)$ and $(I\, u)\, \underline{(I\, u)}$.  

In \thespc\ a step $a$ embeds another step $b$ iff the position of $a$ is a prefix of the position
of $b$. For example, $c$ described above embeds $d$. However, the two steps $(\l \kpr\, x\,{\tt m}\, {\tt b}. x)
(\kpr\, \mathtt{a}\,(\underline{I {\tt f}})\, (\underline{I {\tt d}}))$ are not related by
embedding and are hence disjoint.

An example of gripping was given in the introduction. We revisit gripping in Sec.~\ref{sec:gripping}.

\subsection{Reduction sequences, multisteps and developments}
\label{subsec:ReductionSequencesAndDevelopments}

A \conceptIntro{reduction sequence} (or \conceptIntro{derivation}) is either $\emptyred{t}$, \ie\  an
\emph{empty sequence} indexed by the object $t$, or a (possibly
infinite) sequence $a_1; a_2, \ldots; a_n; \dots$ of
\tredexes\ verifying $\rtgt{a_k} = \rsrc{a_{k+1}}$ for all $k\geq 1$.
In the former case, we define the \conceptIntro{source} as $t$ and in
the latter case as the source of the first step in the sequence.  We
define the \conceptIntro{target} of a finite reduction sequence as
follows: $\rtgt{\emptyred{t}} \eqdef t$, $\rtgt{a_1; \ldots; a_n}
\eqdef \rtgt{a_n}$.  The \conceptIntro{length} of a reduction
sequence, denoted by $\rlength{\cdot}$, is defined as follows:
$\rlength{\emptyred{t}} \eqdef 0$, $\rlength{a_1; \ldots; a_n} \eqdef
n$. The target and length of an infinite sequence are undefined.
We
write $\arsRedseq$ for the set of reduction sequences.  In the
following, reduction sequences are given the names $\reda$, $\reda'$,
$\reda_1$, $\redb$, $\redc$, etc.  We write $t \sred{\reda} u$ to
indicate that $\rsrc{\reda} = t$ and $\rtgt{\reda} = u$.  Also, if
$\reda = a_1; \ldots; a_n$, we denote with $\mredel{\reda}{k}$ the
\tredex\ $a_k$, and write $\mredsub{\reda}{i}{j}$ for the subsequence
$a_i; \ldots; a_j$, if $i\leq j$, and $\emptyred{\rsrc{a_i}}$, if
$i>j$.  We use the symbol $;$ to denote the concatenation of reduction
sequences, allowing to concatenate \tredexes\ and sequences freely,
\eg\ $a;\reda$ or $a;b$ or $\reda;a$ or $\reda ; \redb$, as long as
the concatenation yields a valid reduction sequence. If $\ROccur{t} =
\emptyset$ then we say that $t$ is a \conceptIntro{normal form}. An
object $t$ is \conceptIntro{normalising} iff there exists a reduction
sequence $\reda$ such that $t \sred{\reda} u$ and $u$ is a normal
form.

A \defi{multistep} is a set of coinitial \tredexes, \ie\ a subset of $\ROccur{t}$ for a certain
object $t$.  We denote such sets by the letters $\setreda$, $\setreda'$, $\setredb$, $\setredc$,
$\setredd$, etc. Two multisteps  are \defi{coinitial} if their union is a multistep. 
\defi{Residuals of coinitial \tredexes\ $\setredb$ after $a$} are defined 
by $\residu{\setredb}{a}{b'}$ iff $\residu{b}{a}{b'}$ for some $b \in \setredb$. We also
use the notation $\residus{\setredb}{a}$, defined analogously to $\residus{b}{a}$.  Notice that for
any $a$ and $b$, $\residus{b}{a}$ is a set of coinitial \tredexes; the same happens with
$\residus{\setredb}{a}$ for any $\setredb$.

\defi{Residuals after reduction sequences} $\arsResidualRel\subseteq\arsRedexes \times \arsRedseq
\times \arsRedexes$ are defined as follows: $\residu{b}{\emptyred{t}}{b}$ for all $b \in
\ROccur{t}$, and $\residu{b}{a;\reda}{b'}$ whenever $\residu{b}{a}{b''}$ and
$\residu{b''}{\reda}{b'}$ hold for some $b''$.  
%We denote this relation using $\arsResidualRel$, so that we write \eg\ $\residu{b}{\reda}{b'}$. 
We sometimes use the notation $\residus{b}{\reda}$ for the set of residuals of $b$ after $\reda$, 
and $\residus{}{\reda}$ to denote the relation $\set{(b,b') \sthat \residu{b}{\reda}{b'}}$. 
We also write $\residu{\setredb}{\reda}{b'}$ and $\residus{\setredb}{\reda}$ for the obvious extension of residuals of steps after a reduction sequence to \emph{multisteps}.
%after a reduction sequence, resp. 
Observe that $\residus{\setredb}{a;\reda} = \residus{\residus{\setredb}{a}}{\reda}$.

Next we consider contraction of multisteps. Since, in principle, the order in which the \tredexes\
comprising a multistep $\setreda$ are contracted could affect the target object of $\setreda$ and/or
its residual relation, it becomes necessary to lay out precise definitions on the meaning of
contraction.  This is achieved through the concept of \emph{development}.  Let $\setreda \subseteq
\ROccur{t}$ for some object $t$. The reduction sequence $\reda$ is a \conceptIntro{development} of
$\setreda$ iff $\mredel{\reda}{i} \in
\residus{\setreda}{\mredunt{\reda}{i-1}}$ for all $i \leq \rlength{\reda}$. \Eg\ a development of the multistep $\setreda\eqdef \{a,b\}$ where $a$ is $\underline{(\l x.I x) (I y)}$ and $b$ is $(\l x.I x) (\underline{I y})$ is the reduction sequence $\underline{(\l x.I x) (I y)} \sstep{a} I (\underline{I y}) \sstep{b'} I y$, since $a\in\residu{\setreda}{\emptyred{(\l x.I x) (I y)}}=\setreda$ and $b'\in\residu{\setreda}{a}$ given that $\residu{b}{a}{b'}$. 
The reduction sequence $(\l x.I x) (\underline{I y})\sstep{b} \underline{(\l x.I x) y} \sstep{a'} I y$, where $\residu{a}{b}{a'}$, is also a development of $\setreda$. Note also that the reduction sequence consisting solely of the step $a$ (or the step $b$) is a development of $\setreda$ too.  
A development $\reda$ of $\setreda$ is \conceptIntro{complete} (written $\reda\develops\setreda$) iff $\reda$ is finite and $\residus{\setreda}{\reda} = \emptyset$.  

The \defi{depth} of a multistep $\setreda$, written $\depth{\setreda}$, is the length of its longest complete development.  
If $a \in \setreda$ and $\reda \develops \residus{\setreda}{a}$, then $a ; \reda \develops \setreda$. Consequently, $\depth{\setreda} > \depth{\residus{\setreda}{a}}$, yielding a convenient induction principle for \tredexsets.
\delia{Both the notion of depth and the derived induction principle
  are important tools in 
several proofs of this work.}

Note that it is not a priori clear that a development terminates, nor that the residual relation is finitely branching. Moreover, since there may be
  more than one development of a multistep, it is natural to wonder whether they all have the same
  target and induce the same residual relation. These topics are discussed in the next section
  (\confer\ finite residuals, finite developments and semantic orthogonality axioms). Suffice it to
  say, for now, that complete developments are a valid means of defining contraction of multisteps since the latter do not depend on the complete development chosen (Prop.~\ref{rsl:SOplus}).

Let $\setreda \subseteq \ROccur{t}$ be a multistep. Define $\rsrc{\setreda} \eqdef t$,
$\rtgt{\setreda} \eqdef \rtgt{\reda}$, and $\residu{b}{\setreda}{b'}$ iff $\residu{b}{\reda}{b'}$
where $\reda$ is an arbitrary complete development of $\setreda$.  In order for these definitions to
be coherent, the empty \tredexset\ 
should be indexed by an object, \ie\ $\emptyset_t$, so that $\rsrc{\emptyset_t} \eqdef
\rtgt{\emptyset_t} \eqdef t$.  We will use the notations $\residus{b}{\setreda}$,
$\residu{\setredb}{\setreda}{b'}$, $\residus{\setredb}{\setreda}$ with meanings analogous to those
described for \tredexes. 
Notice that for any $a \in \setreda$, the reduction sequence $a ; \reda'$ is a (complete) development of $\setreda$ iff $\reda'$ is a (complete) development of $\residus{\setreda}{a}$.
As a consequence, $\residus{\setredb}{\setreda} = \residus{\residus{\setredb}{a}}{\residus{\setreda}{a}}$.
\delia{\defi{Multistep contraction} of $\setreda\subseteq \ROccur{t}$, written $t
\mstep{\setreda} u$, where $\rsrc{\setreda} = t$ and $\rtgt{\setreda} = u$, 
denotes an arbitrary complete development $\reda\develops\setreda$, where $t \sred{\reda} u$.}

As a closing comment to this section, it should be mentioned that 
the analysis of contraction of multisteps for higher-order rewriting is non-trivial even for sets of
\emph{pairwise disjoint} \tredexes, since residuals of such sets are not necessarily pairwise
disjoint.
\delia
{Conversely, in first-order rewriting, residuals of pairwise disjoint sets of \tredexes\ are always pairwise disjoint again.
This difference yields the need of a subtle analysis of the behaviour of \tredexsets, which is not required for the first-order case (\confer\ \cite{sekar-rama}), in order to obtain normalisation results applicable to higher-order rewrite systems.
}

%%% Local Variables: 
%%% mode: latex
%%% TeX-master: "paper"
%%% End: 

\section{Axioms for ARS}
\label{sec:axiomsForARS}

We next introduce the axioms for ARS. These are presented in three groups
(Fig.~\ref{fig:groupsOfAxioms}), together with their associated concepts. The \emph{fundamental
  axioms} deal with basic properties of the residual relation; the \emph{embedding axioms} deal with
the interaction between residuals and embedding; and the \emph{gripping axioms} deal with the basic
properties of the gripping relation on redexes. In what follows, free
variables in the statement of an axiom are assumed implicitly universally quantified. For example,
``$\residus{a}{a} = \emptyset$'' should be read as ``For all $a\in \arsRedexes$, $\residus{a}{a} =
\emptyset$''. Finally, bear in mind that in an
expression such as ``$\residu{a}{b}{a'}$'', steps $a$ and $b$ are assumed coinitial.

\begin{figure}[t]
\begin{center}
\begin{tabular}{||l|p{5cm}|c||}
\hline
\textbf{Axiom group} & \textbf{Axioms} & \textbf{Reference}\\
\hline
Fundamental & \axSelfReduction, \axFiniteResiduals, \axAncestorUniqueness, \axFD, \axSO & Sec.~\ref{subsec:fundamentalAxioms}\\
\hline
Embedding & \axLinearity, \axCtxFreeness, \axEnclaveCreation, \axEnclaveEmbedding,
\axMisterious & Sec.~\ref{sec:embedding}\\
\hline
Gripping & \axFda, \axFdb, \axFdc & Sec.~\ref{sec:gripping} \\
\hline
\end{tabular}
\caption{Axioms for ARS presented in three groups}\label{fig:groupsOfAxioms}
\end{center}
\vspace{-5mm}
\end{figure}

\subsection{Fundamental axioms}
\label{subsec:fundamentalAxioms}

The fundamental axioms of an ARS have to do with the properties of the residual relation over
redexes and derivations. The embedding and gripping relations do not participate in these axioms.
The first is \axSelfReduction\ and states, quite reasonably, that nothing is left of a step $a$ if it is contracted. 
\vspace{-3mm}
\begin{center}
\begin{tabular}{@{}p{43mm}p{74mm}}
\textsf{\axSelfReduction} & 
$\residus{a}{a} = \emptyset$. 
\end{tabular}
\end{center}
 
The second is \axFiniteResiduals\ and states that the residuals of a step $b$ after contraction of
a coinitial (and possibly the same) one $a$ is a finite set. In other words, a step
may erase ($\residus{b}{a}=\emptyset$) or copy other coinitial steps, however only a finite number of copies
can be produced.
\vspace{-3mm}
\begin{center}
\begin{tabular}{@{}p{43mm}p{74mm}}
\textsf{\axFiniteResiduals} & 
$\residus{b}{a}$ is a finite set. 
\end{tabular}
\end{center}

The third one, namely \axAncestorUniqueness, states that a step $a$ cannot ``fuse'' two different
steps $b_1$ and $b_2$, coinitial with $a$, into one. In other words, if we use the term ``ancestor'' to
refer to the inverse of the residual relation, then any step can have at most one ancestor.
\vspace{-3mm}
\begin{center}
\begin{tabular}{@{}p{43mm}p{74mm}}
\textsf{\axAncestorUniqueness} & 
$\residu{b_1}{a}{b'}$ $\land$
$\residu{b_2}{a}{b'}$ $\Rightarrow$ $b_1 = b_2$.
\end{tabular}
\end{center}

As discussed in Sec.~\ref{subsec:ReductionSequencesAndDevelopments}, a
multistep $\setreda$ is contracted by performing any complete
development of $\setreda$. However, as already mentioned, developments
may a priori not terminate and, since there may be more than one
development of a multistep, they may not all have the same target or
induce the same residual relation. Any of these situations would
render the purported notion of \tredexset\  contraction
senseless. The following two axioms FD and SO ensure exactly that
these three properties are met. The first states that any development
of a multistep $\setreda$ necessarily terminates.
\vspace{-3mm}
\begin{center}
\begin{tabular}{@{}p{43mm}p{74mm}}
\textsf{Finite developments (FD)} & 
All developments of $\setreda$ are finite. 
\end{tabular}
\end{center}

Note that, together with \axFiniteResiduals, FD implies (by K\"onig's Lemma) that the notion of
\emph{depth} of a multistep $\setreda$ (\ie\ the length of the longest complete development of
$\setreda$) is well-defined. As already mentioned, this provides us with a convenient measure for proving properties
involving multisteps.

The second axiom states that complete
developments of a multistep $\{a,b\}$, consisting of two coinitial
steps, are joinable and induce the same residual
relation (Fig.~\ref{fig:axiom:semanticOrthogonality}).
It is called  PERM in~\cite{thesis-mellies}.

\begin{center}
\begin{tabular}{@{}p{43mm}p{74mm}}
%\begin{tabular}{@{}lp{60mm}}
\textsf{Semantic orthogonality (SO)} &
$\exists \reda,\redb$ s.t. $\reda\develops \residus{a}{b}$ $\land$ $\redb\develops \residus{b}{a}$ $\land$
$\rtgt{a;\redb} = \rtgt{b;\reda}$ $\land$ the relations $\residus{}{a;\redb}$ and $\residus{}{b;\reda}$ coincide.
\end{tabular}
\end{center}

\begin{figure}[t]
\begin{center}
$\begin{array}{c}
\xymatrix{
    s \ar[r]^{a} \ar[d]_{b} & t \ar@{->>}[d]^{\redb} \\
    u \ar@{->>}[r]_{\reda} & v   
}
\end{array}
$
\end{center}
\vspace{-3mm}
\caption{The semantic orthogonality axiom}\label{fig:axiom:semanticOrthogonality}
\vspace{-5mm}
\end{figure}
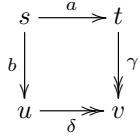

Developments of an \emph{arbitrary} multistep of an ARS are also joinable and induce the same
residual relation. This is reflected in the following result (Lem.~2.18 and Lem.~2.19
in~\cite{thesis-mellies}) which is proved by resorting to all of the above axioms (except for
\axAncestorUniqueness):

\begin{proposition}
\label{rsl:SOplus}
Consider an ARS enjoying the fundamental axioms.
Suppose $\reda\develops \setreda$ and $\redb\develops \setreda$.
Then $\rtgt{\reda} = \rtgt{\redb}$ and the relations $\residus{}{\reda}$ and $\residus{}{\redb}$ coincide.
\end{proposition}

\ignore{ 

\completar{Maybe we want to comment that the proof of our \theprop~\ref{rsl:SOplus} given by the cited lemmas in \cite{thesis-mellies} is based on the idea of permutation between steps in a reduction sequence, and the corresponding idea of \emph{permutation equivalence}.
Maybe there are other references to proposals of proving semantic orthogonality through permutation equivalence worth citing.
}

} % ignore

%%% Local Variables: 
%%% mode: latex
%%% TeX-master: "article"
%%% End: 

\subsection{Embedding axioms}
\label{sec:embedding}

The embedding axioms establish coherence conditions between the embedding relation $<$ and the residual
relation $\arsResidualRel$. In reading these axioms it helps to think of $a<b$ in the setting of
\thespc\ as indicating that the position of the step $a$ is a prefix of the position of $b$. Bear in
mind however, that an ARS does not assume the existence of terms nor positions; this reading is
solely for the purpose of aiding the interpretation of the axioms.

The first axiom, \axLinearity, states that the only way in which a step $a$ can either erase or
produce multiple (two or more) copies of a coinitial step, is if it embeds it.
\vspace{-3mm}
\begin{center}
\begin{tabular}{@{}p{42mm}p{78mm}}
\textsf{\axLinearity} & 
$a\not \leq b$ $\Rightarrow$ $\exists ! b' \sthat \residu{b}{a}{b'}$. 
\end{tabular} 
\end{center}

The second axiom pertains to the invariance of the embedding relation w.r.t. contraction of
steps. Consider three coinitial steps $a, b$ and $c$. Suppose that $\residu{b}{a}{b'}$ and
$\residu{c}{a}{c'}$, for some steps $b'$ and $c'$ (this implies $a\neq c$ and $a\neq b$). If $a$
does not embed $c$ ($a\not<c$), then $a$ cannot grant the ability to $b$ of embedding $c$ ($b \not<
c \ \Rightarrow b' < c'$ cannot happen) or revoke it from $b$ ($b < c \ \Rightarrow b' \not< c'$
cannot happen).
\vspace{-3mm}
\begin{center}
\begin{tabular}{@{}p{35mm}p{85mm}}
\textsf{\axCtxFreeness} & 
$\residu{b}{a}{b'}$ $\land$ $\residu{c}{a}{c'}$ $\Rightarrow$ $a < c \,\lor\, (b < c \ \Leftrightarrow b' < c')$. 
\end{tabular}
\end{center}

The next two axioms assume that two steps $a$ and $b$ are given such that $b<a$. It considers
under what conditions $b'$, the unique residual of $b$ after $a$ ($\residu{b}{a}{b'}$), embeds other
steps $c'$ in the target of $a$. Two cases are considered, first when $c'$ is created by $a$
(\axEnclaveCreation) and then when it is not (\axEnclaveEmbedding).
\vspace{-6mm}
\begin{center}
\begin{tabular}{@{}p{42mm}p{78mm}}
\textsf{\axEnclaveCreation} & 
$b < a$ $\land$ $\residu{b}{a}{b'}$ $\land$ $\residu{\emptyset}{a}{c'}$ $\Rightarrow$ $b' < c'$. \\
\textsf{\axEnclaveEmbedding} &
$\residu{b}{a}{b'}$ $\land$ $\residu{c}{a}{c'}$ $\land$ $b < a < c$ $\Rightarrow$ $b' < c'$. 
\end{tabular} 
 \end{center}

 Finally, the axiom \axMisterious\ is new, in the sense that it does not appear
   in~\cite{thesis-mellies} since it is not required for the results that are proved there. It
     reads as follows:
\vspace{-6mm}
\begin{center}
\begin{tabular}{@{}p{18mm}p{100mm}}
\textsf{\axMisterious} & $a < c$ $\land$ $b < c$ $\land$ $b \not\leq a$ $\land$ $\residu{c}{a}{c'}$ $\Rightarrow$ $\exists b' \sthat \residu{b}{a}{b'} \,\land\, b' < c'$.
\end{tabular} 
\end{center}
 To
 motivate this axiom we illustrate an important property that we shall need to prove for our
 normalisation result (Lem.~\ref{rsl:dominated-stable-contraction}).   We assume given $b<c$ and a
 step $a\neq b$ s.t. $\residu{c}{a}{c'}$ for some $c'$ (\confer\ shaded triangles in the figure). We would like to deduce the existence of $b'$ s.t. (i)
 $\residu{b}{a}{b'}$ and (ii) $b'<c'$. For that we proceed to consider all possible embedding
 relations between $a$, on the one hand, and $b$ and $c$, on the other (see Fig.~\ref{fig:arbolito}): 
% Fig.~\ref{fig:MotivatingAxMisterious}).
%\begin{wrapfigure}[12]{r}{0.5\textwidth}
%\begin{wrapfigure}[8]{r}{.4\textwidth}
\begin{figure}[h]
\begin{center}
\onlyPdf{
\vspace{-1mm}
\includegraphics[scale=0.17]{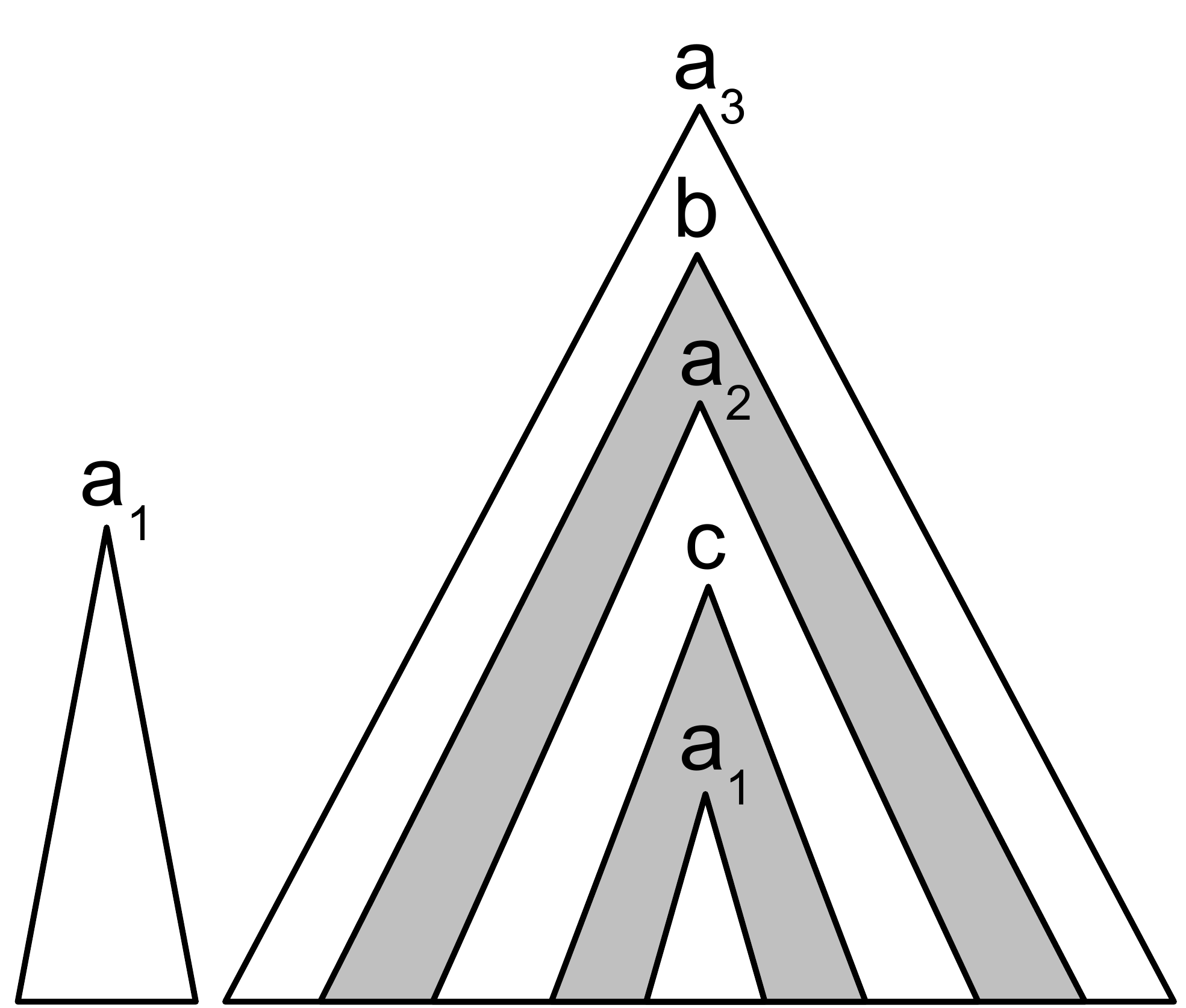}
\vspace{-4mm}
}
\onlyDvi{
\vspace{1cm}
}
\end{center}
\caption{Three redexes $a, b, c$ such that $b<c$ and $a\neq b$.}\label{fig:arbolito}
\vspace{-3mm}
\end{figure}

\begin{itemize}
\item $a\not<c$. This is represented with the two occurrences of $a$ subscripted with 1. We conclude
  (i) and (ii) using \axLinearity\ and \axCtxFreeness.

\minitem $a<c$. 
  \begin{itemize}
  \minitem $b < a$ (hence $b< a < c$). This is represented with the occurrence of $a$ subscripted with
    2. We conclude (i) and (ii) using \axLinearity\ and \axEnclaveEmbedding.
  \minitem $b\not< a$. This is represented with the occurrence of $a$ subscripted with 3. We conclude
    (i) and (ii) using \axMisterious.
  \end{itemize}
\end{itemize}

% ignore

%%% Local Variables: 
%%% mode: latex
%%% TeX-master: "article"
%%% End: 

\subsection{Gripping axioms}
\label{sec:gripping}

In order to motivate our need for the gripping relation~\cite{thesis-mellies} we provide a brief
glimpse of the approach we take for our abstract proof of normalisation. We shall show that,
starting from a normalising object $t$ and by repeatedly contracting multisteps enjoying certain
properties (let us call such multisteps \emph{judicious}), a normal form will be reached. That this
process does not continue indefinitely, shall be guaranteed by providing an appropriate measure that
decreases with each such judicious multistep. The elements that are measured are certain ``multireductions'',
sequences of multisteps, that originate from each of the sources and targets of judicious
multisteps. 
\begin{wrapfigure}[9]{r}{0.4\textwidth}
\begin{center}
$\xymatrix{
t_i \arMulti{d}_{\setpb} \arMulti{r}_{\setpa} & t_i' \\
t_{i+1} \arMulti{r}_{\residu{\setpa}{\setpb}{}}      & t_{i+1}' 
}
$
\end{center}
\vspace{-3mm}
\end{wrapfigure}
In the particular case that a multireduction consists of a sole multistep $\setreda$, our
measure computes its \emph{depth} (the length of the longest complete
development). This is depicted in the figure where $\setpb$ is the judicious multistep, $\setpa$
is the multireduction consisting of just one multistep that is measured and
$\residu{\setpa}{\setpb}$ is what remains of multistep $\setpa$ after $\setpb$ which will also be
measured and compared with the measure of $\setpa$.
We are interested in showing that the depth of $\setpa$ is greater than that of
$\residu{\setpa}{\setpb}{}$.  \delia{In general, this is not true as the
following example in \l-calculus (suggested by V. van Oostrom and also
applicable to higher-order rewrite systems in general) illustrates, where $\setreda := \set{a_1, a_2}$,
$\setredb := \set{b}$ and $D=\l z. z\, z$}. Note that indeed we have $\depth{\setreda} = 2 < 3 =
\depth{\residus{\setreda}{\setredb}}$.
\begin{center}
$\xymatrix{
\underline{(\l x.\underline{D\ x}_{b})\ (\underline{\Id\ y}_{a_2})}_{a_1} \arMulti{d}_{\setpb}
\ar[r]|-*=0@{o}_>(.8){\setpa} & D\, y \\
\underline{(\l x.x\ x) 
    (\underline{\Id\ y}_{a_2})}_{a_1} \ar[r]|-*=0@{o}_>(.8){\residu{\setpa}{\setpb}{}}& y\, y 
}
$
\end{center}

%\arMulti{r}_{\setpa}

The problem is that $a_1$ embeds the step $b$ that duplicates $a_1$'s bound variable; if this variable is
substituted by some other step (in this example, $a_2$) then, after contracting $b$ more copies of
$a_2$ have to be contracted in the development of $\residu{\setpa}{\setpb}{}$. When a step $a$
embeds another step $b$ that has a free occurrence of a variable bound by $a$ we say that $b$
\conceptIntro{grips} $a$ and write:
\begin{center}
$a\grip b$
\end{center}
We shall avoid the above situation 
by requiring our judicious multisteps to be \emph{\ngrip} (\confer\
Lem.~\ref{rsl:residual-same-depth}), in other words, that this situation never occurs. Since our ARS
are over abstract objects (hence there is no notion of term, nor variable, nor bound variable) we
must put forward appropriate axioms that capture gripping in an abstract way.
These
axioms were presented in~\cite{thesis-mellies} for the purpose of providing an abstract proof of
finite developments for ARS (see remark at the end of this section). We next present these axioms. 

\ignore{
Informally, a \tredex\ $b$ \textit{grips} another \tredex\ $a$ iff $a < b$ and there are occurrences
of variables inside $b$ which are bound by the abstraction of $a$. The gripping axioms portray de
properties that we assume valid for the gripping relation. There are three such axioms and they are
described in further detail below.
}

The first one, \axFda, states the role gripping plays in the creation of new embeddings. Consider
three coinitial steps $a,b,c$ and steps $b',c'$ s.t. $\residus{b}{a}{b'}$ and
$\residus{c}{a}{c'}$. Suppose $b'$ embeds $c'$ (\ie\ $b'<c'$). Two situations are possible. If
$a\not< c$, then by \axCtxFreeness, we already know that $b<c$. However, if $a<c$ (and $b\not<c$),
then this axiom may be seen to provide further information. It informs us that $b$ grips $a$: this
is the only way in which $a$ can place (the residual of) $c$ under (the residual of) $b$.

\begin{center}
\begin{tabular}{@{}p{35mm}p{78mm}}
\textsf{\axFda} & $\residus{b}{a}{b'}$ $\land$ $\residus{c}{a}{c'}$ $\land$ $b' < c'$ $\Rightarrow$ $b < c \,\lor\, (a \grip b \,\land\, a < c)$. 
\end{tabular}
\end{center}

The second axiom, \axFdb, states that at any moment a step grips some other step, then this can be
traced back to a ``chain'' of grippings over the ancestors of these steps. An example in \thespc\ of
how $b'\grip c'$ may follow from $b\grip a \grip c$ after contracting step $a$ is
$\underline{(\l y. (\underline{(\l x.\underline{I\, x}_{c})\, y}_{a})) u}_{b}$.

\begin{center}
\begin{tabular}{@{}p{35mm}p{78mm}}
\textsf{\axFdb} & $\residus{b}{a}{b'}$ $\land$ $\residus{c}{a}{c'}$ $\land$ $b' \grip c'$ $\Rightarrow$ $b \grip c$ $\lor$ $b \grip a \grip c$. 
\end{tabular}
\end{center}

The third axiom, \axFdc, states if a step $b$ grips another step $a$ (\ie\ $a\grip b$), then any
step that embeds $b$ either grips $a$ or embeds $a$.

\begin{center}
\begin{tabular}{@{}p{35mm}p{78mm}}
\textsf{\axFdc} & $a \grip b \,\land\, c < b$ $\Rightarrow$ $a \grip c \,\lor\, c \leq a$.
\end{tabular}
\end{center}

\delia{
Although in the abstract framework the embedding relation is clearly not included in the
gripping relation, it is worth noticing that the gripping axioms do
not enforce the opposite inclusion. That being said,  in our concrete $\theppc$ framework,
  the gripping relation  between $\theppc$-redexes (on page~\pageref{page-gripping-ppc})  is included in the
 embedding relation.
 
}

As remarked above, the gripping axioms entail FD, Thm. 3.2. in~\cite{thesis-mellies}.
% states that an ARS satisfying the
% following axioms and properties enjoys FD:
% \ale{queremos realmente poner todo lo que sigue? para que sirve?}
% \begin{enumerate}
% \item The axioms  \axSelfReduction\ and \axFiniteResiduals;
% 
% \item The gripping axioms together with the following one: $\residu{b}{a}{b_1'}$ and
%   $\residu{b}{a}{b_2'}$ and $b_1'\neq b_2'$ $\Rightarrow$ $a < b$; and
% 
% \item The property that the embedding and gripping relations are acyclic.
% \end{enumerate}

%%% Local Variables: 
%%% mode: latex
%%% TeX-master: "paper"
%%% End: 

\subsection{An additional axiom: Stability}

%We conclude this subsection with a comment on the axioms and standard derivations. 

One final axiom which, although not required for our abstract
normalisation proof,  is worthy of mention given the key r\^ole that it
plays in the axiomatic standardisation proof of~\cite{thesis-mellies},
is briefly discussed here.  An ARS satisfying the fundamental axioms
  and the embedding axioms (disregarding both enclave axioms and the
  axiom \axMisterious), enjoys the property of \emph{existence} of
  standard derivations~\cite{thesis-mellies}.  A standard derivation
  is one in which contraction takes place outside-in. The additional
  property of \emph{uniqueness} of such derivations can also be proved
  in this axiomatic framework. For that the ARS must satisfy the
  fundamental axioms, the embedding axioms (disregarding
  \axMisterious, which is not required) and an additional axiom called
  \axStability.  This last axiom states that steps can be created in a
  unique way.

\begin{center}
\begin{tabular}{@{}p{18mm}p{100mm}}
\textsf{Stability} &
Assume $a \disj b$, $\residu{a}{b}{a'}$, $\residu{b}{a}{b'}$, and there exists some $d'$ such that $\residu{d'_1}{b'}{d'}$ and $\residu{d'_2}{a'}{d'}$. 
Then there exists $d$ such that $\residu{d}{a}{d'_1}$, $\residu{d}{b}{d'_2}$, and either $a \not\leq d$ or $b \not\leq d$.
\end{tabular}
\end{center}

As mentioned, \axStability\ is not required for our abstract normalisation proof. This is quite fortunate since
neither the parallel-or TRS nor the \thespc\ of the introduction, enjoy stability. Let us look
at the case of \thespc. 

\vspace{-1mm}
\begin{center}
$\xymatrix@C=10pt@R-8pt{
    &  (\l \kpr\, x\,{\tt m}\, {\tt s}. x) (\kpr\, \mathtt{a}\,(I {\tt f}) (I {\tt d}))
    \ar[rd]^{b}\ar[ld]_{a} & \\
   (\l \kpr\, x\,{\tt m}\, {\tt s}. x) (\kpr\, \mathtt{a}\,{\tt f}\, (I {\tt d})) \ar[d]_{d_1'}
   \ar[dr]_{b'} & & (\l \kpr\,
   x\,{\tt m}\, {\tt s}. x) (\kpr\, \mathtt{a}\,(I {\tt f})\, {\tt d}) \ar[d]_{d_2'}\ar[dl]^{a'}\\
    I & (\l \kpr\, x\,{\tt m}\, {\tt s}. x) (\kpr\, \mathtt{a}\,{\tt f}\,{\tt d}) \ar[d]_{d'} & I \\
    & I & 
}$
\end{center}

\vspace{-2mm}
The steps depicted above meet the antecedent of the statement of the stability axiom. However, the
conclusion is not satisfied since both steps $d_1'$ and $d_2'$ are created (by $a$ and $b$, resp.). 

This concludes our presentation of the axioms of an ARS.  A summary of all three groups is given in Fig.~\ref{fig:summaryOfAxioms}. 

\begin{figure}
\begin{center}
\begin{tabular}{@{}p{42mm}p{78mm}}
\textbf{Fundamental axioms} 
\\
\hline
\\[-2pt]
\textsf{\axSelfReduction} & 
$\residus{a}{a} = \emptyset$. \\
\textsf{\axFiniteResiduals} & 
$\residus{b}{a}$ is a finite set. \\
\textsf{\axAncestorUniqueness} & 
$\residu{b_1}{a}{b'}$ $\land$ $\residu{b_2}{a}{b'}$ $\Rightarrow$ $b_1 = b_2$.\\
\textsf{Finite developments} & 
All developments of a multistep $\setreda$ are finite. \\
\textsf{Semantic orthogonality (SO)} &
$\exists \reda,\redb$ s.t. $\reda\develops \residus{a}{b}$ $\land$ $\redb\develops \residus{b}{a}$ $\land$
$\rtgt{a;\redb} = \rtgt{b;\reda}$ $\land$ the relations $\residus{}{a;\redb}$ and $\residus{}{b;\reda}$ coincide.
\\[4pt]
\textbf{Embedding axioms}  
\\
\hline
\\[-2pt]
\textsf{\axLinearity} & 
$a\not \leq b$ $\Rightarrow$ $\exists ! b' \sthat \residu{b}{a}{b'}$.  \\
\textsf{\axCtxFreeness} & 
$\residu{b}{a}{b'}$ $\land$ $\residu{c}{a}{c'}$ $\Rightarrow$ $a < c \,\lor\, (b < c \ \Leftrightarrow b' < c')$. 
\\
\textsf{\axEnclaveCreation} & 
$b < a$ $\land$ $\residu{b}{a}{b'}$ $\land$ $\residu{\emptyset}{a}{c'}$ $\Rightarrow$ $b' < c'$. \\
\textsf{\axEnclaveEmbedding} &
$\residu{b}{a}{b'}$ $\land$ $\residu{c}{a}{c'}$ $\land$ $b < a < c$ $\Rightarrow$ $b' < c'$. 
\\
\textsf{\axMisterious} & 
$a < c$ $\land$ $b < c$ $\land$ $b \not\leq a$ $\land$ $\residu{c}{a}{c'}$ $\Rightarrow$ $\exists b' \sthat \!\residu{b}{a}{b'} \land b' < c'$.
\\[4pt]
\textbf{Gripping axioms}
\\
\hline
\\[-2pt]
\textsf{\axFda} & $\residus{b}{a}{b'}$ $\land$ $\residus{c}{a}{c'}$ $\land$ $b' < c'$ $\Rightarrow$ $b < c \,\lor\, (a \grip b \,\land\, a < c)$. 
\\
\textsf{\axFdb} & $\residus{b}{a}{b'}$ $\land$ $\residus{c}{a}{c'}$ $\land$ $b' \grip c'$ $\Rightarrow$ $b \grip c$ $\lor$ $b \grip a \grip c$. \\
\textsf{\axFdc} & $a \grip b \,\land\, c < b$ $\Rightarrow$ $a \grip c \,\lor\, c \leq a$.
\end{tabular}
\end{center}
\vspace{-3mm}
\caption{The three groups of axioms of an ARS}\label{fig:summaryOfAxioms}
\vspace{-3mm}
\end{figure}

%%% Local Variables:
%%% mode: latex
%%% TeX-master: "article"
%%% End:

\section{Multireductions over an ARS}
\label{sec:multireductionsInARS}
Our normalisation result states conditions under which an object can
be reduced to a normal form by repeatedly contracting multisteps,
thus requiring a precise meaning
for sequences of such multisteps. Also,
we must introduce some qualifiers for multisteps that enjoy properties
that are useful for the development that shall follow.

\subsection{Multireductions}

The concept of reduction sequence introduced earlier for \tredexes,
makes sense for \tredexsets\ as well.  A \conceptIntro{multireduction
  sequence}, or just \conceptIntro{multireduction}, is either
$\nil_t$, an \emph{empty multireduction} indexed by the object $t$, or
a sequence of \tredexsets\ $\setreda_1; \ldots; \setreda_n; \ldots$ 
where $\rtgt{\setreda_{k+1}} = \rsrc{\setreda_{k}}$ for all $k \geq 1$.
We use $\mreda$, $\mredb$, $\mredc$, $\mredd$ to denote
multireductions and $\mredel{\mreda}{k}$ and $\mredsub{\mreda}{i}{j}$
with the same meanings given for reduction sequences.  Source, target
and length of multireductions are defined analogously as done before
for reduction sequences.  We write $t \mred{\mreda} u$ to denote that
$\rsrc{\mreda} = t$ and $\rtgt{\mreda} = u$.  Some comments:
\begin{itemize}
\item The \defi{length of a multireduction} is the number of \ale{its} \tredexsets\ , it is not related to the size of the sets.
\item An element of a multireduction can be an empty \tredexset, so that the only corresponding development is the empty reduction sequence indexed by its source. 
\item A multireduction consisting of one or more occurrences of $\emptyset_t$, and $\nil_t$, are
  different multireductions. In particular, $\rlength{\emptyset_t } = 1$ while $\rlength{\nil_t} = 0$.
We will say that a \mredseq\ is \defi{trivial} iff all its elements are empty \tredexsets. Empty \mredseqs\ are trivial.
\end{itemize}

The residual relation is extended from \tredexsets\ to multireductions, exactly as we have extended it from \tredexes\ to reduction sequences. 
We use the notations $\residu{b}{\mreda}{b'}$, $\residus{b}{\mreda}$, $\residu{\setredb}{\mreda}{b'}$, $\residus{\setredb}{\mreda}$ and $\residus{}{\mreda}$ 
as expected.

Let $\mathcal{MR}$ be the set of \tredexsets\ associated to an
ARS. Notice that, in contrast to the notion of residuals for steps, residuals
can be considered as a function on multisteps, \ie\ $\arsResidualRel :
\mathcal{MR} \times \mathcal{MR} \to \mathcal{MR}$, since
$\residus{\setreda}{\setredb}$ is again a \tredexset\ for any
$\setreda, \setredb$.  
This distinguishing feature of \tredexsets\ leads to the definition of the \conceptIntro{residual of a \mredseq\ after a \tredexset}, for which we will (ab)use the notation $\arsResidualRel$.  
If $\rsrc{\setredb} = t$ then $\residus{\nil_t}{\setredb} \eqdef \nil_{\rtgt{\setredb}}$; if
$\setreda$ and $\setredb$ are coinitial
then $\residus{(\setreda;\mreda)}{\setredb} \eqdef
\residus{\setreda}{\setredb} ; (\residus{\mreda}{\residus{\setredb}{\setreda}})$.  
Observe that, in
spite of the name ``residual'' and the notation $\arsResidualRel$, the
above definition corresponds to a partial function $\mathcal{MRS}
\times \mathcal{MR} \to \mathcal{MRS}$, where $\mathcal{MRS}$ stands
for the set of multireductions.  Notice also that
$\rlength{\residus{\mreda}{\setredb}} = \rlength{\mreda}$.

A \conceptIntro{(multistep) reduction strategy} 
for an ARS $\arsa$ is any function $\strs : (\arsTerms \setminus NF) \to \partsOf{\arsRedexes}$ such
that $\strs(t) \neq \emptyset$ and $\strs(t) \subseteq \ROccur{t}$ for all $t$; here $NF$ stands for the set of normal forms of $\arsa$.
A multistep reduction strategy determines, for each object, a \emph{\mredseq}:  if $t \in NF$,
then the associated \mredseq\ is $\emptyred{t}$, otherwise it is $\strs(t_0) ; \strs(t_1) ; \ldots ;
\strs(t_n) ; \ldots$ where $t_0 \eqdef t$ and $t_{n+1} \eqdef \rtgt{\strs(t_{n})}$.
A reduction strategy is \conceptIntro{normalising} iff for any object $t$, the determined \mredseq\
ends in a normal form for all normalising objects. A \defi{single-step reduction strategy} is a
multistep reduction strategy $\strs$ s.t. $\strs(t)$ is a singleton for every $t$ in the domain of
$\strs$. In this case, the multireduction sequence determined by $\strs$ is in fact a \emph{reduction sequence}.

The independence of order of contraction of steps, formalised in \theprop~\ref{rsl:SOplus}, extends
to \tredexsets~\cite[\thelem~2.24]{thesis-mellies} and to \mredseqs. The former is a
  consequence of \theprop~\ref{rsl:SOplus} and  the latter then follows by induction on
  $\mreda$.

\begin{proposition}
\label{rsl:SOredexset}
 Consider an ARS enjoying the group of fundamental axioms.
\begin{enumerate}
\item \label{it:SOredexset}
Let $\setreda, \setredb \subseteq \ROccur{t}$. The target and residual relation of $\setreda ;
\residus{\setredb}{\setreda}$ and $\setredb ; \residus{\setreda}{\setredb}$ coincide.

\item \label{it:SOmredseq}
Let $\mreda$ be a \mredseq, and $\setredb \subseteq \ROccur{t}$. The target and residual relation associated to $\mreda ; \residus{\setredb}{\mreda}$ and $\setredb ; \residus{\mreda}{\setredb}$ coincide.

\end{enumerate}

\end{proposition}

\ignore{
\begin{proof}
Observe that a development of $\setreda$ followed by one of $\residus{\setredb}{\setreda}$
corresponds to a development of $\setreda \cup \setredb$, and similarly for a development of
$\setredb$ followed by one of $\residus{\setreda}{\setredb}$. Then \theprop~\ref{rsl:SOplus} allows
us to conclude.
\end{proof}
} % ignore

%\begin{figure}[h]
\begin{center}
$\begin{array}{cc}
\xymatrix{
t \arMulti{d}_{\setpb} \arMulti{r}_{\setpa} & s \arMulti{d}^{\residu{\setpb}{\setpa}{}}\\
u \arMulti{r}_{\residu{\setpa}{\setpb}{}}      & v 
}
&
\xymatrix{
t \arMulti{d}_{\setpb} \arMultiOp{rr}{>>}_{\mreda} & & s \arMulti{d}^{\residu{\setpb}{\mreda}{}}\\
u \arMultiOp{rr}{>>}_{\residu{\mreda}{\setpb}{}}      & & v 
}
\end{array}
$
\end{center}
%\caption{Permutation of multisteps and permutation of multisteps and multireductions}\label{fig:permutationOfMultisteps} 
%\end{figure}

\ignore{

\begin{proposition}
\label{rsl:SOmredseq} 
Consider an ARS enjoying the group of fundamental axioms.
Let $\mreda$ be a \mredseq, and $\setredb \subseteq \ROccur{t}$. Then the target and residual relation of $\mreda ; \residus{\setredb}{\mreda}$ and $\setredb ; \residus{\mreda}{\setredb}$ coincide.
\end{proposition}

\begin{proof}
A simple induction on $\rlength{\mreda}$, resorting to Prop.~\ref{rsl:SOredexset}, suffices.
\Confer\ \thelem~2.33 and \thelem~2.29 in \cite{thesis-mellies}. 
\end{proof}
} % ignore

\ignore{

\completar{Se podr\'ia mostrar la prueba, y/o incluir alg\'un diagrama al estilo

$\xymatrix@R=20pt@C=40pt{
  t \arMulti{r}^{\setreda} \arMulti{d}_{\setredb} & 
  t_0 \arMultiOp{rr}{>>}^{\mreda} \arMulti{d}^{\residus{\setredb}{\setreda}} & & 
  t' \arMulti{d}^{\residus{\setredb}{\setreda ; \mreda}} \\
  s \arMulti{r}_{\residus{\setreda}{\setredb}} &
  s_0 \arMultiOp{rr}{>>}_{\residus{\mreda}{\residus{\setredb}{\setreda}}} & & s'
}$

que puede servir para introducir este tipo de diagramas, que se van a usar en la prueba abstracta. Incluso se pueden citar usos de este tipo de diagramas.
}

Prop.~\ref{rsl:SOplus} also allows us to introduce the \conceptIntro{depth} of $\setreda$, written $\depth{\setreda}$, a notion which will be an ingredient of the abstract normalisation proof.
It is defined as the length of the longest complete development of $\setreda$.
K\"onig's Lemma guarantees the well-definedness of this function, for ARS enjoying \axiomuse{\axFiniteResiduals} and \axiomuse{\axFD}.

} % ignore

\subsection{Key Concepts}
\label{sec:free-domin}

\delia{This section introduces several notions that are crucial in the 
development of our abstract normalization proof. More precisely, 
our normalization result holds for strategies choosing 
\ngrip\ and necessary multisteps, which are introduced as follows. 
First of all, starting from
the embedding relation on redexes, we  define two key relations on multisteps: 
{\it free-from}
and {\it embedded by}. Secondly, starting from the gripping notion on redexes, we 
define a gripping notion on multisteps, together with the associated
concept of {\it \ngrip}\  multisteps. Last but not least, we define the {\it uses} relation
which defines what it means for a multistep to be \emph{necessary}.
}

\begin{deliae}
\medskip\noindent
\textbf{Free-from and \domintext\ \tredexsets} \\ 

Two notions related with embedding and involving \tredexsets\ are
crucial to define the main elements of the abstract normalisation
proof.  In order to introduce these notions, we 
discuss different ways to extend the notion of embedding to \tredexsets.

Two different meanings could be assigned to the notation $a >
\setredb$: either that there exists some $b \in \setredb$ that verifies
$a > b$, or else that $a > b$ for all $b \in \setredb$. 
Since we are going to apply this general framework to {\it terms}, it
seems natural to adopt the
former interpretation.  Conversely, when considering $\setreda >
\setredb$, we take a ``forall'' meaning on the $\setreda$ side: to
consider that $\setredb$ embeds $\setreda$, it must embed each of its
elements.  In the sequel, we use the notations $a > \setredb$ and
$\setreda > \setredb$ with the just given meanings, and we say that
$a$ (or $\setreda$) is \conceptIntro{embedded by} $\setredb$.
%This is one of the two notions mentioned at the beginning of this section.

On the other hand, we say that a step is \conceptIntro{free from} a coinitial \tredexset, if it is neither equal to nor embedded by a step in $\setredb$. In turn, a \tredexset\ $\setreda$ is free from another, coinitial \tredexset\ $\setredb$, if $a$ is free from $\setredb$ \emph{for all}%
\footnote{\delia{Observe that given the just discussed meanings, ``$\setreda \not> \setredb$'' and ``$\setreda$ free from $\setredb$'' are different predicates.
%Thus the decision of giving separated name and notation for the free from relation.
}} %
$a \in \setreda$.
The notion also extends to \mredseqs, as defined below.
\end{deliae}
Formally, given $a,
\setreda, \mreda$ coinitial with $\setredb$:
\begin{itemize}
\item 
$a$ is \textbf{free from} \setredb, written $a \freefrom \setredb$, iff $a \not\geq b$ for all $b \in \setredb$.
\item 
$\setreda$ is \textbf{free from} \setredb, written $\setreda \freefrom \setredb$, iff $a \freefrom \setredb$ for all $a \in \setreda$.
\item $\mreda$ is \textbf{free from} \setredb, written 
$\mreda \freefrom \setredb$,  iff either $\mreda = \nil_{\rsrc{\setredb}}$ or $\mreda = \setreda; \mreda'$, $\setreda \freefrom \setredb$ and $\mreda' \freefrom \residus{\setredb}{\setreda}$.
%.
\item 
$a$ is \textbf{embedded by} \setredb, written $a \dominated \setredb$, iff $a \notin \setredb$ and
$\metaexists b \in \setredb \sthat a > b$.
\item $\setreda$ is \textbf{embedded by} \setredb, written
$\setreda \dominated \setredb$, iff $a \dominated \setredb$ for all $a \in \setreda$.
\end{itemize}

Notice that being free from and \dominbytext\ $\setredb$ are complementary for a single (coinitial) \tredex\ $a$, unless $a \in \setredb$, \ie\ exactly one of $a \in \setredb$, $a \freefrom \setredb$ and $a \dominated \setredb$ holds.  This need not be the case for a \tredexset\ $\setreda$: even if $\setreda \cap \setredb = \emptyset$, it could well be the case that neither $\setreda \freefrom \setredb$ nor $\setreda \dominated \setredb$ hold, if some elements of $\setreda$ are free from $\setredb$ while others are \dominbytext\ it.

\delia{On the other hand, any $\setreda$ verifying $\setreda \cap \setredb =
\emptyset$ can be \textbf{partitioned} into a free
subset $\thefree{\setreda}$ and an \domintext\ subset
$\thedomin{\setreda}$ \wrt\ $\setredb$, \ie\ $\setreda =
\thefree{\setreda} \uplus \thedomin{\setreda}$, $\thefree{\setreda}
\freefrom \setredb$, and $\thedomin{\setreda} \dominated \setredb$.
The partition of a \tredexset\ into a free and a\edu{n}
  \domintext\ part \wrt\ another, coinitial \tredexset, is a relevant
  notion for the development of the abstract proof we describe in
  Sec.~\ref{sec:reduction-strategies}.}

Consider the following \mredseq\ $\Delta$ in the \lc:
\begin{center}
$
\begin{array}{@{}l}
\overbrace{(\l x . x (\underbrace{I 5}_{a})) (\underbrace{I 3}_{b})}^{d} 
	\ 	(\overbrace{I (\underbrace{I 4}_{c})}^{e}  
\ 
	\ \mstep{\set{e}} \ 
\overbrace{(\l x . x (\underbrace{I 5}_{a'})) (\underbrace{I 3}_{b'})}^{d'} 
	\ 	(\underbrace{I 4}_{c'})
	\ \mstep{\set{d', c'}} \ 
(\underbrace{I 3}_{b''}) (\underbrace{I 5}_{a''}) \ 4
\end{array}
$
\end{center}
In this case, we have $\set{c,d,e} \freefrom \set{a,b}$, $\set{a,b} \freefrom \set{c,e}$,
$\set{a,b,c} \dominated \set{d,e}$.  Also, we have $\mreda \freefrom \set{a,b}$, because $\set{e} \freefrom \set{a,b}$ and $\set{d', c'} \freefrom \set{a', b'}$.
If we define $\setreda = \set{b,c,e}$ and $\setredb = \set{a,d}$, we observe that neither  $\setreda
\freefrom \setredb$ nor $\setreda \dominated \setredb$ hold.  The partition of $\setreda$ \wrt\
$\setredb$ gives $\thefree{\setreda} = \set{c,e}$ and $\thedomin{\setreda} = \set{b}$.

Observe also that being free from a \tredexset\ extends to parts of a \mredseq, namely%
\onlyReport{\footnote{Note that given the formal definition of the free from relation, it is not immediate that $\mreda_1 ; \mreda_2 \freefrom \setredb$ implies $\mreda_1 \freefrom \setredb$. In fact, a proof of this statement would follow the same lines of that we give for the more general Lemma~\ref{rsl:free-from-parts-mredseq}. This is the motivation for the statement of this Lemma. }}:
\begin{lemma}
\label{rsl:free-from-parts-mredseq}
Assume $\mreda_1 ; \mreda_2 ; \mreda_3 \freefrom \setredb$.
Then $\mreda_2 \freefrom \residus{\setredb}{\mreda_1}$.
\end{lemma}

\begin{proof}
%A straightforward induction on $\pair{\rlength{\mreda_1}}{\rlength{\mreda_2}}$ suffices.
We proceed by induction on $\pair{\rlength{\mreda_1}}{\rlength{\mreda_2}}$. Let $\mreda$ be $\mreda_1 ; \mreda_2 ; \mreda_3$.

The base case is when $\mreda_1 = \mreda_2 = \emptyred{\rsrc{\setredb}}$.
In this case $\residus{\setredb}{\mreda_1} = \setredb$. Then the definition of $\freefrom$ suffices to conclude.

Suppose that $\mreda_1 = \emptyred{\rsrc{\setredb}}$ and $\mreda_2 = \setreda ; \mreda'_2$.
In this case, $\mreda = \setreda ; \mreda'_2 ; \mreda_3$, so that $\mreda \freefrom \setredb$ implies $\setreda \freefrom \setredb$ and $\mreda'_2 ; \mreda_3 = \emptyred{\rtgt{\setreda}} ; \mreda'_2 ; \mreda_3 \freefrom \residus{\setredb}{\setreda}$.
We observe that $\pair{\rlength{\emptyred{\rtgt{\setreda}}}}{\rlength{\mreda'_2}} < \pair{\rlength{\mreda_1}}{\rlength{\mreda_2}}$, therefore we can apply \ih, obtaining that $\mreda'_2 \freefrom \residus{\residus{\setredb}{\setreda}}{\emptyred{\rtgt{\setreda}}} = \residus{\setredb}{\setreda}$. Recalling that $\setreda \freefrom \setredb$, we get $\mreda_2 \freefrom \setredb = \residus{\setredb}{\mreda_1}$.

If $\mreda_1 = \setreda ; \mreda'_1$, then $\mreda \freefrom \setredb$ implies $\setreda \freefrom \setredb$ and $\mreda'_1 ; \mreda_2 ; \mreda_3 \freefrom \residus{\setredb}{\setreda}$.
Observe $\pair{\rlength{\mreda'_1}}{\rlength{\mreda_2}} < \pair{\rlength{\mreda_1}}{\rlength{\mreda_2}}$, then \ih\ yields $\mreda_2 \freefrom \residus{\residus{\setredb}{\setreda}}{\mreda'_1} = \residus{\setredb}{\mreda_1}$.
\end{proof}

\delia{
The axiom \axiomuse{\axLinearity} can be extended to the
residuals of a step \emph{after a \tredexset} from which
  it is free from, as proved in the following Lemma. 
This fact is used in Sec.~\ref{sec:reduction-strategies}.
%The abstract normalisation proof described in Sec.~\ref{sec:reduction-strategies} resorts to this fact. 
} 
\begin{lemma}[\delia{Linearity after a multistep}]
\label{rsl:free-then-unique-residual}
Consider an ARS enjoying the fundamental axioms and \axiomuse{\axLinearity}; and $a, \setredb$ such that $a \freefrom \setredb$. Then $\residus{a}{\setredb}$ is a singleton.
\end{lemma}

\begin{proof}
  By induction on $\depth{\setredb}$.  If $\setredb = \emptyset$, then we conclude by observing that
  $\residus{a}{\emptyset} = \set{a}$.  Otherwise assume some $b \in \setredb$.  Then $a \freefrom
  \setredb$ implies $b \not\leq a$, thus \axiomuse{\axLinearity} yields $\residus{a}{b} = \set{a'}$.
  Let us show that $a' \freefrom \residus{\setredb}{b}$.  Take $b'_0$ such that
  $\residu{b_0}{b}{b'_0}$ for some $b_0 \in \setredb$.  Assume $b'_0 < a'$.  Then $b \not\leq a$ and
  \axiomuse{\axCtxFreeness} imply $b_0 < a$ thus contradicting $a \freefrom \setredb$. On the other
  hand, $b'_0 = a'$ would contradict \axiomuse{\axAncestorUniqueness}.  Consequently, $a' \freefrom
  \residus{\setredb}{b}$.  The \ih\ can then be applied to obtain that $\residus{a'}{\residus{\setredb}{b}{}}$
  is a singleton.  We conclude by observing that $\residus{a}{\setredb} =
  \residus{\residus{a}{b}}{\residus{\setredb}{b}} = \residus{a'}{\residus{\setredb}{b}{}}$.
\end{proof}

\medskip\noindent
\delia{\textbf{Gripping for Multisteps} \\}%
Now we discuss the extension of the gripping relation to multisteps.
We say that:
\begin{itemize}
\item $\setredb$ \defi{grips} $a$, written $a \grip \setredb$, iff $a \grip b$ for some $b \in \setredb$.

\item $\setredb$  \defi{grips} $\setreda$, written $\setreda \grip \setredb$, iff $a \grip \setredb$ for at least one $a \in \setreda$.
\end{itemize}
\begin{deliae}
We define  a multistep $\setredb$ to be \conceptIntro{\ngrip} iff for any finite
multireduction $\mredd$, if $\mredd$ is coinitial with $\setredb$, then 
$\ROccur{\rtgt{\mredd}} \notgrip \residus{\setredb}{\mredd}$. Notice
that $\setredb$ being \ngrip\ implies that every residual of 
$\setredb$ (after a coinitial step, multistep or multiderivation) is.

The extension of gripping to \tredexsets\ leads to a strengthened variant of the free from relation. Given two coinitial \tredexsets\ $\setreda$ and $\setredb$, we say that $\setreda$ is \conceptIntro{\superfree} from $\setredb$, iff $\setreda \freefrom \setredb$ and $\setreda \notgrip \setredb$. Analogously, if $\mreda$ and $\setredb$ are coinitial, we say that $\mreda$ is \conceptIntro{\superfree} from $\setredb$, iff $\mreda \freefrom \setredb$ and $\mredel{\mreda}{k} \notgrip \residus{\setredb}{\mredunt{\mreda}{k-1}}$ for all $k$.
\end{deliae}

\medskip
\begin{deliae}
It is worth noticing that an alternative definition of \ngrip\ can be
given {\it coinductively} as follows: a multistep $\setredb$ is
\conceptIntro{\ngrip} iff $\ROccur{\rsrc{\setredb}} \notgrip
\setredb$, and for any multistep $\setreda$ coinitial with $\setredb$, the set $\residus{\setredb}{\setreda}$ is \ngrip.

It is not difficult to show that both definitions are equivalent. 
For that, 
let us write $\rng$ for our first definition of  \ngrip\ while
$\cng$ is used for the coinductive definition. 

\begin{lemma}
\label{l:equivalencia-ng}
A multistep $\setredb$ is $\rng$ iff $\setredb$ is $\cng$.
\end{lemma}

\begin{proof} \mbox{}

\noindent $\Rightarrow$)
The proof is by coinduction.

Take $\mredd = \nil_{\rsrc{\setredb}}$. Then 
$\ROccur{\rsrc{\setredb}} = \ROccur{\rsrc{\nil_{\rsrc{\setredb}}}} =  $ \linebreak $
\ROccur{\rtgt{\nil_{\rsrc{\setredb}}}}  
\notgrip  \residus{\setredb}{\nil_{\rsrc{\setredb}}} = \setredb$. 

Take any multistep $\setreda$ which is coinitial with 
$\setredb$. By hypothesis  $\setredb$ is \rng\ so that
$\residus{\setredb}{\setreda}$ is also \rng. 
The coinductive hypothesis gives $\residus{\setredb}{\setreda}$ \cng\ and we can thus
conclude $\setredb$ \cng. 

\noindent $\Leftarrow$)
Let $\mredd$ be a finite multiderivation which is coinitial
with $\setredb$.  We want to show $\ROccur{\rtgt{\mredd}} \notgrip \residus{\setredb}{\mredd}$.
We proceed by induction on $\mredd$. 

If $\mredd = \nil_{\rsrc{\setredb}}$, then  
$\ROccur{\rtgt{\nil_{\rsrc{\setredb}}}} = \ROccur{\rsrc{\nil_{\rsrc{\setredb}}}}
= \ROccur{\rsrc{\setredb}} \notgrip  \setredb  = \residus{\setredb}{\nil_{\rsrc{\setredb}}}$. 

If $\mredd = \setreda; \mredd'$, then we want to show 
$\ROccur{\rtgt{\setreda;\mredd'}} \notgrip \residus{\setredb}{\setreda;\mredd'}$.  

Since $\setredb$ is \cng, then $\residus{\setredb}{\setreda}$ is \cng\
by definition, so that by \ih\ we have 
$\ROccur{\rtgt{\mredd'}} \notgrip  \residus{\residus{\setredb}{\setreda}}{\mredd'}$. We conclude since
$\ROccur{\rtgt{\mredd'}} =
\ROccur{\rtgt{\setreda;\mredd'}} $
and $\residus{\residus{\setredb}{\setreda}}{\mredd'} = \residus{\setredb}{\setreda;\mredd'} $. 
\end{proof}
\end{deliae}

\delia{The choice between the use of the predicate $\rng$ or $\cng$ remains a
matter of taste, as the (forthcoming) proofs relying on $\rng$
are not simplified in any notable way by adopting instead $\cng$.  Still, 
since $\ngrip$ plays  a central r\^ole  in this paper, 
explicitly  spelling out its coinductive nature 
provides one more way of understanding it.  
}

\medskip
\begin{deliae}
Another interesting remark concerns the relation between
our $\ngrip$ predicate and that of  {\it universally
  $<$-external}~\cite{DBLP:conf/popl/AccattoliBKL14}: 
a redex $a$ is said to be \defi{universally $<$-external}, \ie\
external with respect to any reduction step (and thus wrt any
derivation)
 if $a$ is $<$-minimal in $\ROccur{\rsrc{a}}$, and $\residus{a}{b}{a'}$ implies $a'$ universally $<$-external for all $b$ coinitial with $a$. 
A redex which is universally $<$-external is in particular $\ngrip$,
but the converse does not necessarily hold: a $\ngrip$ redex is not always
universally $<$-external. For example, in the case of the
$\lambda$-calculus, the redex $b$ in $\underline{I(\underline{I(\underline{II}_{c})}_{b})}_{a}$
is $\ngrip$, but is not
universally $<$-external as $b$ is not
$<$-minimal. Another example can be found in Section~\ref{sec:ppc-strategy},
where we define a
reduction strategy for $\theppc$, which selects
$\ngrip$ redexes which are not necessarily universally $<$-external. 

\end{deliae}

\delia{
\medskip\noindent
\textbf{Uses, Needed Step and Necessary Multisteps} \\ 
}%
Given a multireduction and some coinitial \tredexset, a further property the abstract normalisation
proof is interested in is whether the \tredexset\ is at least partially contracted along the
multireduction, or if it is otherwise completely ignored.  We will say that a \tredexset\ is
\conceptIntro{used} in a multireduction, iff at least one residual of the former is included (\ie\
contracted) in the latter.  Formally, let $b$ be a \tredex, $\setreda$ and $\setredb$ two
\tredexsets, and $\mreda$ a \mredseq, such that all of them are coinitial.
\begin{itemize}
\item \setreda\ \conceptIntro{uses} $b$ iff $b \in \setpa$;

\item  \mreda\ \conceptIntro{uses}  $b$ iff $\mredel{\mreda}{k} \cap
  (\residus{b}{\mredunt{\mreda}{k-1}}) \neq \emptyset$ for at least one $k$; and

\item \setreda\ (resp. \mreda) \conceptIntro{uses} \setredb\ iff it uses at least one $b \in
\setpb$. 

\end{itemize}

A \tredex\ $a$ is \conceptIntro{needed} iff for every multireduction $\rsrc{a}
\mred{\mreda} u$ such that $u$ is a normal form, $\mreda$ uses $a$.  A \tredexset\ $\setreda$ is
\conceptIntro{necessary}, iff for every multireduction $\rsrc{\setreda} \mred{\mreda} u$ such that
$u$ is a normal form, $\mreda$ uses $\setreda$.  The notion of necessary \tredexset\ generalises that of
needed redex (notice that any singleton whose only element is a needed redex is a necessary set). As
mentioned in the introduction, there is an important difference: while not all terms admit a needed
redex, any term admits at least one necessary set, \ie\ the set of \emph{all} its redexes.

%%% Local Variables: 
%%% mode: latex
%%% TeX-master: "paper"
%%% End: 

\section{Necessary normalisation for ARS}
\label{sec:reduction-strategies}
We prove in this section that, for any ARS verifying  the fundamental axioms, 
the embedding axioms, 
and the gripping axioms, the systematic contraction of \emph{necessary} and \emph{\ngrip} \tredexsets\ is normalising.

The overall structure of the proof is inspired by the work on
first-order term rewriting systems by Sekar and Ramakrishnan in~\cite{sekar-rama}.
Assume that $\strs$ is a reduction strategy selecting always necessary and \ngrip\ multisteps.
Consider an initial \mredseq\
%
%\begin{wrapfigure}[12]{r}{0.5\textwidth}
\begin{wrapfigure}[12]{r}{0.45\textwidth}
%\vspace{-2mm}
%\vspace{2mm}
\minicenter{
$\xymatrix@R=20pt@C=80pt{
  t_0 \arMulti{d}_{\strs(t_0)} \arMultiOp{r}{>>}^{\mreda_0} & u \ar@{=}[d] \\ 
%   t_1 \arMulti{d}_{\strs(t_1)} \arMultiOp{r}{>>}^{\mreda_1} & u \ar@{=}[d] \\ 
%   t_2 \ar@{.}[d]               \arMultiOp{r}{>>}^{\mreda_2} & u \ar@{.}[d] \\
  t_1 \ar@{.}[d]               \arMultiOp{r}{>>}^{\mreda_1} & u \ar@{.}[d] \\
  t_n \arMulti{d}_{\strs(t_n)} \arMultiOp{r}{>>}^{\mreda_n} & u \ar@{=}[d] \\
  u  \ar@{=}[r]^{\mreda_{n+1}} & u
}$
}
%\vspace{-1mm}
\caption{Proof idea}\label{fig:normalisationProofIdea}
\end{wrapfigure}
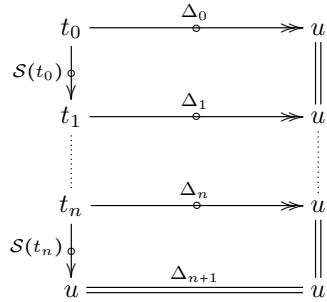
$t_0 \mred{\mreda_0} u \in \NForms$, and $t_1$ the target term of the multistep selected by $\strs$ for $t_0$, \ie\ $t_0 \mstep{\strs(t_0)} t_1$.
We construct a \mredseq\ $t_1 \mred{\mreda_1} u$, such that the \mredseq\ $\mreda_1$ is strictly smaller than the original one \wrt\ a convenient well-founded ordering $<$ on multireductions.
We have thus transformed the original $t_0 \mred{\mreda_0} u$ in $t_0 \mstep{\strs(t_0)} t_1 \mred{\mreda_1} u$.
Well-foundedness of $<$ entails
  that by iterating this procedure one may deduce that repeated contraction of the \tredexsets\ selected by the strategy $\strs$ yields the normal form $u$.
This is depicted in Fig.~\ref{fig:normalisationProofIdea} where $\mreda_{k+1}$ is strictly smaller than $\mreda_k$ for all $k$ and $\mreda_{n+1}$ is a trivial multireduction. The original multireduction $\mreda_0$ is first transformed into $\strs(t_0);\mreda_1$, then successively into $\strs(t_0);\ldots;\strs(t_k);\mreda_{k+1}$; and finally into $\strs(t_0);\ldots;\strs(t_n)$.

\medskip
Several notions contribute to this proof.  We define a
\textbf{measure} inspired from \cite{sekar-rama, vanOostrom:1999},
based on the {\it  depths} of the \tredexsets\ composing a multireduction. 
\begin{deliae}
More precisely, the measure of a \mredseq\ is
  defined as the sequence of the depths of its elements, \emph{taken
    in reversed order}, \ie\ given a \mredseq\ $\mreda = \mredunt{\mreda}{n}$, 
the measure of $\mreda$, written $\mredms{\mreda}$, is the $n$-tuple $\langle
\depth{\mredel{\mreda}{n}},\ldots,$ \linebreak $\depth{\mredel{\mreda}{1}} \rangle$.
Then, the \textbf{lexicographic order} $\ltms$ is used to compare (measures of)
multireductions, where $\ltms$ is  defined on $n$-tuples of natural numbers as follows:
\[ \langle x_1, \ldots, x_n \rangle \ltms \langle y_1, \ldots, x_n \rangle \mbox{ iff }
   \exists 1 \leq j \leq n\ x_j <  y_j \mbox{ and } 
   \forall 1 \leq  i  < j\ x_j = y_j \] 
Therefore, $\mredms{\mreda} < \mredms{\mredb}$ implies $\mredms{\mredc;\mreda} <\ \mredms{\mredd;\mredb}$ for any $\mredc$, $\mredb$ that verify $\rtgt{\mredc} = \rsrc{\mreda}$, $\rtgt{\mredd} = \rsrc{\mredb}$, and $\rlength{\mredc} = \rlength{\mredd}$.  
Notice that \mredseqs\ comparable by $\ltms$ may not be coinitial.
%Notice that $\mredc$ and $\mredd$ may not be coinitial. 
%   This allows us to compare
%   successive $\mreda_k$, as depicted in
%   Fig.~\ref{fig:normalisationProofIdea} in order to guarantee that the
%   iterative construction of these \mredseqs\ terminates.  
\end{deliae} 

This (well-founded) ordering allows only to compare \mredseqs\ having the
same length; the minimal elements are the $n$-tuples of the form
$\langle 0, \ldots, 0 \rangle$ which corresponds exactly to the
trivial \mredseqs.    As
remarked in~\cite{vanOostrom:1999}, the measure used
in~\cite{sekar-rama}, based on sizes of multisteps rather than depths,
is not well-suited for a higher-order setting.

\medskip 
To construct $\mreda_{k+1}$, we observe that the fact that
$\strs(t_k)$ is a \emph{necessary} set, implies that it is used along
$\mreda_k$ at least once.  Therefore, we can consider the last element
of $\mreda_k$ that includes (some residual of) an element of
$\strs(t_k)$. Let us call this element $\setreda$.  We build
the diagram shown in
\thefig~\ref{fig:normalisation-proof-iteration-1}, where $\mreda_k =
\mreda' ; \setreda ; \mreda''$, $\setreda \cap
\residus{\strs(t_k)}{\mreda'} \neq \emptyset$, and $\mreda''$ does not
use $\residus{\strs(t_k)}{\mreda' ; \setreda}$.
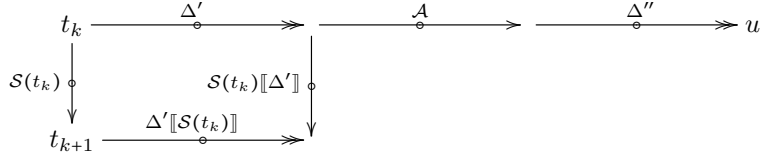
\begin{figure}[H]
\vspace{-1mm}
\begin{center}
$\xymatrix@C=35pt@R=30pt{ 
t_k \arMultiOp{rr}{>>}^{\mreda'} \arMulti{d}_{\strs(t_k)\,}
& & \arMulti{rr}^{\setreda} \arMulti{d}_{\,\residus{\strs(t_k)}{\mreda'}\, }
& & \arMultiOp{rr}{>>}^{\mreda''} & & u \\
t_{k+1} \arMultiOp{rr}{>>}^{\residus{\mreda'}{\strs(t_k)}} 
& & 
}$
\vspace{-3mm}
\end{center}
\caption{Construction of $\mreda_{k+1}$, starting point}\label{fig:normalisation-proof-iteration-1}
\end{figure}
By setting $\setreda_1 = \setreda \cap \residus{\strs(t_k)}{\mreda'} \neq
  \emptyset$, $\setreda_2 = \residus{(\setreda \setminus \setreda_1)}{\setreda_1}$, and $\setredb =
  \residus{\strs(t_k)}{\mreda' ; \setreda_1}$, we    can refine the previous diagram
as depicted in  \thefig~\ref{fig:normalisation-proof-iteration-3}. Now $\setreda_2 ; \mreda''$ does not use $\setredb$. Notice that $\setreda_1 \neq \emptyset$ implies $\depth{\setreda_2} < \depth{\setreda}$. Observe also that $\setreda_1 \subseteq \residus{\strs(t_k)}{\mreda'}$, implying $\residus{\setreda_1}{\residus{\strs(t_k)}{\mreda'}} = \emptyset$.
\begin{figure}[H]
\begin{center}
$\xymatrix@C=28pt@R=35pt{ 
t_k \arMultiOp{rr}{>>}^{\mreda'} \arMulti{d}_{\strs(t_k)\,}
& & \arMulti{rrr}^{\setreda_1} \arMulti{d}_{\,\residus{\strs(t_k)}{\mreda'}\, }
& & & s \arMulti{r}^{\setreda_2} \arMulti{d}^{\,\setredb}
& \arMultiOp{rr}{>>}^{\mreda''} & & \ar@{=}[d] u \\
t_{k+1} \arMultiOp{rr}{>>}_{\residus{\mreda'}{\strs(t_k)}} 
& & \ar@{=}[rrr]_{\residus{\setreda_1}{\residus{\strs(t_k)}{\mreda'}} = \emptyset} 
& & & s' \arMultiOp{rrr}{>>}_{\mredb'} 
& & & u
}$
\vspace{-3mm}
\end{center}
\caption{Construction of $\mreda_{k+1}$, finished}\label{fig:normalisation-proof-iteration-3}
\end{figure}
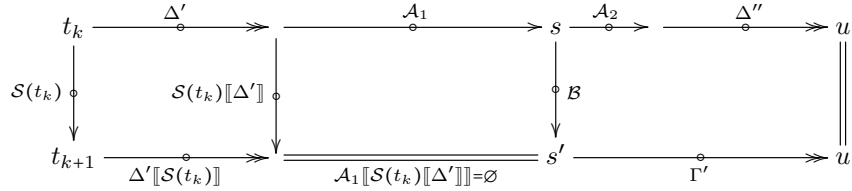
To conclude the construction of $\mreda_{k+1}$, it suffices to obtain a \mredseq\ $\mredb'$ such that $s' \mred{\mredb'} u$ and $\mredms{\mredb'} \leqms \mredms{\setreda_2 ; \mreda''} \ltms \mredms{\setreda ; \mreda''}$; 
taking the multisteps of a \mredseq\ in \emph{reversed} order in the measure allows
to assert $\mredms{\mreda_{k+1}} \ltms \mredms{\mreda_k}$. Building such a $\mredb'$ is the most demanding part of the proof. Following~\cite{sekar-rama}, this construction is based on the following observations:
\begin{deliae}
\begin{itemize}
\item 
\textbf{Partition of each multistep in free and \domintext\ parts.}
Each multistep comprising $\setreda_2 ; \mreda''$ can be partitioned into a free and an \domintext\ part \wrt\ $\setredb$, as remarked in Sec.~\ref{sec:free-domin} after the definition of free and \domintext\ multisteps.

\item 
\textbf{Postponement of \domintext\ parts.}  We prove that each
\domintext\ part can be postponed, \ie\ permuted with a subsequent
free part, preserving the free and \domintext\ nature of the permuted
multisteps, and the depth of the free part
(\confer\ Lemmas~\ref{rsl:dominated-stable-contraction},
\ref{rsl:free-setred-from-tgt-to-src} and \ref{rsl:postponement}). We
describe this phenomenon in more detail at the beginning of
Sec.~\ref{sec:main-proof}.  
\item 
\textbf{Irrelevance of postponed \domintext\ parts.}
Since $\setredb$ is not used and $u \in \NForms$ implies $\residus{\setredb}{\setreda_2 ; \mreda''} = \emptyset$, we prove that the (postponed) \domintext\ part  can be simply ignored when defining $\mredb'$ (\confer\ Lemmas~\ref{rsl:dominated-erasure} and \ref{rsl:free-equiv-mred}).
\item
\textbf{Measure of free multireduction does not increase in projection.}
Since  $\strs(t_k)$ is \ngrip, hence also $\setredb$ is \ngrip, the depth of the free part of each \tredexset\ can be proven greater or equal than that of its residual after (the corresponding residual of) $\setredb$ (\confer\ Lemmas~\ref{rsl:residual-same-depth} and \ref{rsl:residual-same-depth-mred}).
This is the reason for the introduction of gripping, which then allows to apply the general structure of the proof in \cite{sekar-rama} in the abstract setting of this work.
\end{itemize}
\end{deliae}

 We describe the details in the remainder of this section.

%%% Local Variables: 
%%% mode: latex
%%% TeX-master: "paper"
%%% End: 

\subsection{Relevance of gripping}
In this section we develop the abstract normalisation proof up to the result showing the relevance of the notion of \emph{gripping}, \ie\ the invariance of the depth of a \tredexset\ $\setreda$ after the contraction of a \tredexset\ $\setredb$, if
$\setreda$ is independent from $\setredb$, \ie\  $\setreda \freefrom \setredb$ and $\setreda \notgrip \setredb$, \confer\ \thelem~\ref{rsl:residual-same-depth}; 
and the extension of such invariance to the measure of \mredseqs\ after the contraction of a \ngrip\ set, \confer\ \thelem~\ref{rsl:residual-same-depth-mred}.

\begin{lemma}[\delia{\Superfreeness\ preservation}]
\label{rsl:free-grip-preservation-by-free-contraction}
Consider $\setreda, \setredb$ such that $\setreda \freefrom \setredb$, $\setreda \notgrip \setredb$, and $d \in \setreda$.
Then $\residus{\setreda}{d} \freefrom \residus{\setredb}{d}$ and $\residus{\setreda}{d} \notgrip \residus{\setredb}{d}$.
\end{lemma}

\begin{proof}
   If $\setredb=\emptyset$, then the result holds trivially since also
   $\residus{\setredb}{d}=\emptyset$. So assume $b \in
   \setredb$. Next, we may assume some $a \in \setreda$ s.t. $a\neq
   d$. Otherwise $\residus{\setreda}{d} = \emptyset$ and the result
   also holds trivially. For the same reason, we may assume $\residu{a}{d}{a'}$ for
   some $a'$. Similarly, we may assume there exists $b'$
   s.t. $\residu{b}{d}{b'}$.

  The hypotheses implies the following: $b \not\leq a$, $b \not\leq d$, $a \notgrip b$, and $d
  \notgrip b$.

Observe $b' = a'$ would contradict \axiomuse{\axAncestorUniqueness}.
On the other hand, $b' < a'$ would imply $b < a \,\lor\, (d \grip b \land d < a)$ by \axiomuse{\axFda}, while $a' \grip b'$ would imply $a \grip b \,\lor\, a \grip d \grip b$ by \axiomuse{\axFdb}.
Therefore, either case would contradict the hypotheses.
Thus we conclude.
\end{proof}

\begin{lemma}[\delia{Depth preservation}]
\label{rsl:residual-same-depth}
Let $\setreda, \setredb \subseteq \ROccur{t}$ such that $\setreda \freefrom \setredb$ and $\setreda \notgrip \setredb$. Then $\depth{\setreda} = \depth{\residus{\setreda}{\setredb}}$.
\end{lemma}

\begin{proof}
By induction on $\depth{\setreda}$.
If $\setreda = \emptyset$, then $\residus{\setreda}{\setredb} = \emptyset$ and we conclude.
Otherwise, let $\reda = d ; \reda'$ such that $\reda \develops \setreda$ and $\depth{\setreda} = \rlength{\reda}$.
Observe that $\reda' \develops \residus{\setreda}{d}$, implying $\depth{\setreda} = \depth{\residus{\setreda}{d}} + 1$.
\theLem~\ref{rsl:free-grip-preservation-by-free-contraction} allows to apply the \ih, obtaining $\depth{\residus{\setreda}{d}} = \depth{\residus{\residus{\setreda}{d}}{\residus{\setredb}{d}}}$, so that \theprop~\ref{rsl:SOredexset}:(\ref{it:SOredexset}) yields 
$\depth{\residus{\setreda}{d}} = \depth{\residus{\residus{\setreda}{\setredb}}{\residus{d}{\setredb}}}$.
In turn, \thelem~\ref{rsl:free-then-unique-residual} implies $\residus{d}{\setredb} = \set{d'}$ for some step $d'$.
\begin{center}
$\xymatrix@C=52pt@R-5pt{ 
t \ar[r]^{d} \arMulti{d}_{\setredb\,}
& \arMulti{rr}^{\residus{\setreda}{d}}  \arMulti{d}^{\,\residus{\setredb}{d}}
& &  \\
\ar[r]_{\residus{d}{\setredb} = \set{d'}} & \arMulti{rr}_{\residus{\residus{\setreda}{d}}{\residus{\setredb}{d}} = \residus{\residus{\setreda}{\setredb}}{\residus{d}{\setredb}}} & & 
}$
\end{center}
Therefore, for any $\redb$ such that $\redb \develops \residus{\residus{\setreda}{\setredb}}{\residus{d}{\setredb}}$, we have $d'; \redb \develops \residus{\setreda}{\setredb}$.
Consequently,
$\depth{\setreda}
		\leq \depth{\residus{\setreda}{d}} + 1
		= \depth{\residus{\residus{\setreda}{\setredb}}{\residus{d}{\setredb}}} + 1
		\leq \depth{\residus{\setreda}{\setredb}}$.

Conversely, consider $\redb = e' ; \redb'$ such that $\redb \develops \residus{\setreda}{\setredb}$ and $\depth{\residus{\setreda}{\setredb}} = \rlength{\redb}$; observe that $\redb' \develops \residus{\residus{\setreda}{\setredb}}{e'}$.
Let $e \in \setreda$ such that $\residu{e}{\setredb}{e'}$.
\theLem~\ref{rsl:free-then-unique-residual} implies $\residus{e}{\setredb} = \set{e'}$, implying that $\redb' \develops \residus{\residus{\setreda}{\setredb}}{\residus{e}{\setredb}}$, so that \theprop~\ref{rsl:SOredexset}:(\ref{it:SOredexset}) yields $\redb' \develops \residus{\residus{\setreda}{e}}{\residus{\setredb}{e}}$.
Therefore, $\depth{\residus{\setreda}{\setredb}} \leq \depth{\residus{\residus{\setreda}{e}}{\residus{\setredb}{e}}} + 1$.
Again, \thelem~\ref{rsl:free-grip-preservation-by-free-contraction} allows to apply the \ih, to obtain $\depth{\residus{\setreda}{e}} = \depth{\residus{\residus{\setreda}{e}}{\residus{\setredb}{e}}}$.
In turn, for any $\reda$ such that $\reda \develops \residus{\setreda}{e}$, we get $e ; \reda \develops \setreda$.
Consequently, 
$\depth{\residus{\setreda}{\setredb}}
	\leq \depth{\residus{\residus{\setreda}{e}}{\residus{\setredb}{e}}} + 1 
	= \depth{\residus{\setreda}{e}} + 1
	\leq \depth{\setreda}$.
Thus we conclude.
\end{proof}

\delia{
The following result lifts Lemma~\ref{rsl:residual-same-depth} to \mredseqs.
By replacing local absence of gripping to the hereditary \ngrip\ property, we enforce \superfreeness\ of the \mredseq\ from the given \tredexset. Hence, invariance of depth for a \tredexset\ can be lifted to invariance of measure for a \mredseq.
}

\begin{lemma}[\delia{Measure preservation}]
\label{rsl:residual-same-depth-mred}
Let $\mreda$ be a multireduction and $\setredb$ a \tredexset, such
that $\mreda$ and $\setredb$ are coinitial, $\setredb$ is
\ngrip\ and $\mreda \freefrom \setredb$.  Then $\mredms{\mreda} =
\mredms{\residus{\mreda}{\setredb}}$.
\end{lemma}

\begin{proof}
By induction on $\rlength{\mreda}$.
If $\mreda = \nil_{\rsrc{\setredb}}$, then $\residus{\mreda}{\setredb} = \nil_{\rtgt{\setredb}}$, so
we conclude immediately. Assume, therefore, $\mreda = \setreda; \mreda'$, so that $\residus{\mreda}{\setredb} = \residus{\setreda}{\setredb} ; \residus{\mreda'}{\residus{\setredb}{\setreda}}$.
Observe $\setreda \freefrom \setredb$, $\setreda \notgrip \setredb$, $\mreda' \freefrom \residus{\setredb}{\setreda}$ and $\residus{\setredb}{\setreda}$ is \ngrip.
Then Lem.~\ref{rsl:residual-same-depth} implies $\depth{\setreda} = \depth{\residus{\setreda}{\setredb}}$, and 
the \ih\ on $\mreda'$ yields $\mredms{\mreda'} = \mredms{\residus{\mreda'}{\residus{\setredb}{\setreda}}}$.
Thus we conclude.
\end{proof}

%%% Local Variables: 
%%% mode: latex
%%% TeX-master: "article"
%%% End: 

\subsection{Normalisation proof}
\label{sec:main-proof}

\begin{figure}[t]
\begin{center}
$\begin{array}{c@{\qquad}c}
\xymatrix@R=18pt@C=35pt{
     & & & \\
    & t \ar[d]|{\circ}_{\setredc \dominated \setredb \,}\arMulti{ul}_{\setredb} \ar@{.>}[r]|{\circ}_{\setreda \, \freefrom \, \setredb} 
		& \ar@{.>}[d]|{\circ}^{\ \residus{\setredc}{\setreda} \,\dominated\, \residus{\setredb}{\setreda}} \arMulti{ur}^{\residus{\setredb}{\setreda}} & \\
     & s \ar[r]|{\circ}_{\setreda' \,\freefrom \, \residus{\setredb}{\setredc}}\arMulti{dl}^{\residus{\setredb}{\setredc}}
     & u              &\\
    & \ar@{}[r]^{\setreda' = \residus{\setreda}{\setredc}} & &
}
&
\xymatrix@R=18pt@C=35pt{
     & & & & \\
    & t \ar[d]|{\circ}_{\setredc \dominated \setredb \,} \arMulti{ul}_{\setredb} \ar@{.>>}[rr]|{\circ}_{\mreda \, \freefrom \, \setredb} 
		& & \ar@{.>}[d]|{\circ}^{\ \residus{\setredc}{\mreda} \,\dominated\, \residus{\setredb}{\mreda}} \arMulti{ur}^{\residus{\setredb}{\mreda}} & \\
     & s \ar@{->>}[rr]|{\circ}_{\mreda' \,\freefrom \, \residus{\setredb}{\setredc}}\arMulti{dl}^{\residus{\setredb}{\setredc}}
     & & u       &\\
    & \ar@{}[rr]^{\mreda' = \residus{\mreda}{\setredc}} & & &
}
\end{array}$
\end{center}
\vspace{-3mm}
\caption{Postponement of \domintext\ multisteps: the one step and multiple step cases; \thelem~\ref{rsl:free-setred-from-tgt-to-src} and \thelem~\ref{rsl:postponement} respectively.}
\label{fig:postponementOfDominatedMultisteps}
\vspace{-3mm}
\end{figure}
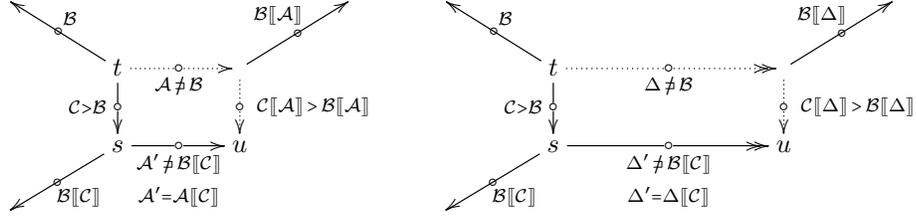

The next ingredient in the normalisation proof is the ability to
\emph{postpone}\index{postponement} an \domintext\ \tredexset\ after a
free \tredexset\ or \mredseq.  The situation is described in
Fig.~\ref{fig:postponementOfDominatedMultisteps}.  The diagram on the
left shows that an \domintext\ (by $\setredb$) multistep can be
\emph{postponed} after a free (from $\residus{\setredb}{\setredc}$)
one, yielding a \mredseq\ in which the free \tredexset\ precedes the
\domintext\ one (Lem.~\ref{rsl:free-setred-from-tgt-to-src}).
\delia{Moreover, the \emph{depth of the free \tredexset\ is preserved}
  by the postponement. To enforce this we show, resorting to
  \axiomuse{\axFdc}, that the \domintext\ \tredexset\ does not grip
  the (ancestor of the) free one, so that
  Lemma~\ref{rsl:residual-same-depth} can be applied. This is the only
  r\^ole of \axiomuse{\axFdc} in the normalisation proof.}

The diagram on the right shows that a \domintext\ \tredexset\ can
be postponed after a free \emph{\mredseq} as well (Lem.~\ref{rsl:postponement}).

We observe that the only r\^ole of the added \axiomuse{\axMisterious} axiom in the normalisation proof, is to verify that $\residus{\setredc}{\setreda} \dominated \residus{\setredb}{\setreda}$ in the left-hand side diagram.

\begin{lemma}[\delia{\Domination\ preservation}]
\label{rsl:dominated-stable-contraction}
Let $\setreda, \setredb, \setredc \subseteq \ROccur{t}$ such that $\setreda \cap\setredb = \emptyset$ and $\setredc \dominated \setredb$.
Then $\residus{\setredc}{\setreda} \dominated \residus{\setredb}{\setreda}$.
\end{lemma}

\begin{proof}
  By induction on $\depth{\setreda}$. If $\setreda = \emptyset$, then
  $\residus{\setredc}{\setreda} = \setredc$ and $\residus{\setredb}{\setreda} = \setredb$, so that
  we conclude immediately. Otherwise, consider $a \in \setreda$ and $c' \in \residus{\setredc}{a}$
  (if $\residus{\setredc}{a}=\emptyset$, then $\residus{\setredc}{\setreda}=\emptyset$ and
    $\residus{\setredc}{\setreda} \dominated \residus{\setredb}{\setreda}$ holds trivially).
  Let $c \in \setredc$ such that $c'
  \in \residus{c}{a}$. Note that $a\neq c$ for otherwise $\residus{c}{a}=\emptyset$. We will
  verify the existence of some $b' \in \residus{\setredb}{a}$ such that $b' < c'$, so that
  $\residus{\setredc}{a} \dominated \residus{\setredb}{a}$. Let $b \in \setredb$ be such that $b <
  c$, as follows from the hypothesis. Observe that $a = b$ or $a = c$ would contradict,
  respectively, the hypotheses of this lemma or our observation above on the existence of
  $c'$. Therefore $a
  \neq b$ and $a \neq c$. We consider two cases.

\begin{enumerate}
\item \textbf{Case $a \not < c$.} Then $b < c$ implies $a \not< b$, so that
\axiomuse{\axLinearity} implies $\residus{b}{a} = \set{b'}$, and then
\axiomuse{\axCtxFreeness} applies to obtain $b' < c'$. 

%In this case, $a < b$ would contradict either $a \not < c$ or $b < c$. Therefore $a \not < b$, so that 

\item \textbf{Case $a < c$. }If $b < a$, \ie\ $b < a < c$, then \axiomuse{\axLinearity} implies $\residus{b}{a} = \set{b'}$ (since $a \not < b$), and therefore \axiomuse{\axEnclaveEmbedding} applies to obtain $b' < c'$. Otherwise, we have $a < c$, $b < c$ and $b \not\leq a$, then \axiomuse{\axMisterious} applies to obtain $\residu{b}{a}{b'}$ and $b' < c'$ for some $b'$.
\end{enumerate}

Hence, we have verified $\residus{\setredc}{a} \dominated \residus{\setredb}{a}$. Moreover,
\axiomuse{\axAncestorUniqueness} yields $\residus{\setreda}{a} \cap \residus{\setredb}{a} =
\emptyset$. Consequently, we can apply the \ih\ on $\residus{\setreda}{a}$, obtaining $\residus{\residus{\setredc}{a}}{\residus{\setreda}{a}} \dominated \residus{\residus{\setredb}{a}}{\residus{\setreda}{a}}$. Thus we conclude.
\end{proof}

\begin{lemma}[\delia{Postponement after a multistep}]
\label{rsl:free-setred-from-tgt-to-src}
Let $\setredb \subseteq \ROccur{t}$, $t \mstep{\setredc} s \mstep{\setreda'} u$, such that $\setredc \dominated \setredb$, $\setreda' \freefrom \residus{\setredb}{\setredc}$ and $\setredb$ is \ngrip. 
Then there exists $\setreda \subseteq \ROccur{t}$ s.t. $\setreda' = \residus{\setreda}{\setredc}$,
$\setreda \freefrom \setredb$ and $\depth{\setreda} = \depth{\setreda'}$.
\end{lemma}

\begin{proof}
If $\setreda' = \emptyset_{s}$, then taking $\setreda = \emptyset_{t}$ suffices to conclude.
So we assume $\setreda' \neq \emptyset_{s}$ and proceed by induction on $\depth{\setredc}$.
If $\setredc = \emptyset$, \ie\  $s = t$, then we conclude by taking $\setreda' \eqdef \setreda$; observe that in this case $\residus{\setredb}{\setredc} = \setredb$.

Consider $c \in \setredc$ and $t \sstep{c} t_0 \mstep{\residus{\setredc}{c}} s$. Since $\setredc
\dominated \setredb$, $c \notin \setredb$ and hence $\set{c} \cap \setredb = \emptyset$; we can apply Lem.~\ref{rsl:dominated-stable-contraction} to obtain $\residus{\setredc}{c} \dominated \residus{\setredb}{c}$.
Moreover $\residus{\setredb}{\setredc} = \residus{\residus{\setredb}{c}}{\residus{\setredc}{c}}$, and $\setredb$ \ngrip\ implies $\residus{\setredb}{c}$ \ngrip.
Therefore, the \ih\ on $\residus{\setredc}{c}$ yields the existence of some $\setreda'' \subseteq \ROccur{t_0}$ such that $\setreda' = \residus{\setreda''}{\residus{\setredc}{c}}$, $\setreda'' \freefrom \residus{\setredb}{c}$ and $\depth{\setreda''} = \depth{\setreda'}$.
Hence, to conclude the proof, it suffices to verify the existence of some $\setreda \subseteq \ROccur{t}$ verifying 
(1) $\setreda'' = \residus{\setreda}{c}$, (2) $\setreda \freefrom \setredb$ and (3) $\depth{\setreda} = \depth{\setreda''}$.
Observe that $\setreda' \neq \emptyset_{s}$ and $\depth{\setreda''} = \depth{\setreda'}$ imply $\setreda'' \neq \emptyset_{t_0}$.

\begin{enumerate}
\item Let $b_0 \in \setredb$ such that $b_0 < c$, so that \axiomuse{\axLinearity} implies $\residus{b_0}{c} = \set{b_0''}$.
Let $a'' \in \setreda''$. Then $a''$ being created by $c$ would imply $b_0'' < a''$ by \axiomuse{\axEnclaveCreation}, contradicting $\setreda'' \freefrom \residus{\setredb}{c}$. Therefore, $\residu{a}{c}{a''}$ for some $a$.
Let $\setreda \eqdef \set{a \ \in \ROccur{t} \sthat \exists \,a'' \in \setreda'' \,.\ \residu{a}{c}{a''}}$. Observe that $\setreda'' \subseteq \residus{\setreda}{c}$ and let us show that also $\residus{\setreda}{c} \subseteq \setreda''$. 

Let $a_0 \in \residus{\setreda}{c}$, $a \in \setreda$ such that $\residu{a}{c}{a_0}$, $a''\in \setreda''$ such that $\residu{a}{c}{a''}$.
Observe that $c < a$ would imply $b_0 < c < a$, and then $b''_0 < a''$ by \axiomuse{\axEnclaveEmbedding}, contradicting $\setreda'' \freefrom \residus{\setredb}{c}$.
Moreover, $c = a$ would contradict $\residu{a}{c}{a''}$, \confer\ \axiomuse{\axSelfReduction}.
Therefore $c \not\leq a$, so that \axiomuse{\axLinearity} applies yielding that $\residus{a}{c}$ is a singleton, hence $a_0 = a'' \in \setreda''$.
Consequently, $\residus{\setreda}{c} \subseteq \setreda''$, and then $\residus{\setreda}{c} = \setreda''$.

\item 
Let $a \in \setreda$ and $b \in \setredb$. If $b$ is minimal in $\setredb$ \wrt\ $<$, then $\setredc \dominated \setredb$ implies 
$b_0 < c$ for some $b_0 \in \setredb$, therefore $c \not\leq b$ (since $c \leq b$ would contradict minimality of $b$), hence $\residus{b}{c} = \set{b''}$ by \axiomuse{\axLinearity}. 
Let $a'' \in \setreda''$ such that $\residu{a}{c}{a''}$. Observe that we have already verified that $c \not < a$.
Then $b < a$ would imply $b'' < a''$ by \axiomuse{\axCtxFreeness}, contradicting $\setreda'' \freefrom \residus{\setredb}{c}$; hence, $b \not< a$. 
In turn, if $b$ is not minimal in $\setredb$ \wrt\ $<$, then well-foundedness of $<$ implies the existence of some $b_0$ such that $b_0 < b$ and $b_0$ is minimal in $\setredb$ \wrt\ $<$, so that $b_0 \not < a$ as we have just shown, and therefore $b \not < a$.
Consequently, $\setreda \freefrom \setredb$.

\item 
Consider $b_0 \in \setredb$ such that $b_0 < c$ and $a \in \setreda$. 
Observe that $a \grip c$ would imply either $a \grip b_0$ or $b_0 \leq a$ by \axiomuse{\axFdc}, contradicting $\setredb$ being \ngrip\ and $\setreda \freefrom \setredb$ respectively. Therefore $\setreda \notgrip c$, and moreover $\setreda \freefrom c$ (recall $c \not\leq a$ for any $a \in \setreda$). Hence we can apply Lem.~\ref{rsl:residual-same-depth} to obtain $\depth{\setreda} = \depth{\setreda''}$.
Thus we conclude.
\end{enumerate}
\vspace{-6mm}
\end{proof}

The next result extends Lem.~\ref{rsl:free-setred-from-tgt-to-src} to multireductions: a
  multistep \dominbytext\ $\setredb$ may be \emph{postponed} after a multireduction
free from the same $\setredb$, without affecting neither the free-from and \domination\ relations
\wrt\ (the corresponding residual of) $\setredb$, nor the measure of the ``free'' multireduction.

\begin{lemma}[\delia{Postponement after a \mredseq}]
\label{rsl:postponement}
Let $t \mstep{\setredc} s \mred{\mreda'} u$ and $\setredb \subseteq \ROccur{t}$ such that $\setredb$
is \ngrip, $\setredc \dominated \setredb$, and $\mreda' \freefrom \residus{\setredb}{\setredc}$.
Then there exists some multireduction $\mreda$ verifying $\mreda' = \residus{\mreda}{\setredc}$, so
that \theprop~\ref{rsl:SOredexset}:(\ref{it:SOmredseq}) yields
$t \mred{\mreda} s' \mstep{\residus{\setredc}{\mreda}} u$ for some object $s'$; and moreover $\mreda \freefrom \setredb$,
$\residus{\setredc}{\mreda} \dominated \residus{\setredb}{\mreda}$, and $\mredms{\mreda} =
\mredms{\mreda'}$. 
\end{lemma}

\begin{proof}
By induction on $\rlength{\mreda'}$.
If $\mreda' = \nil_s$, \ie\  $u = s$, then it suffices to take $\mreda \eqdef \nil_t$, so that $s' = t$.

Assume $\mreda' = \mreda'_0; \setreda'$, so that $t \mred{\setredc} s \mred{\mreda'_0} u' \mstep{\setreda'} u$.
\thelem~\ref{rsl:free-from-parts-mredseq} \delia{on $\emptyred{s} ; \mreda'_0 ; \setreda'$} implies $\mreda'_0 \freefrom \residus{\setredb}{\setredc}$.
Then we can apply the \ih\ on $\mreda'_0$ obtaining $\mreda'_0 = \residus{\mreda_0}{\setredc}$ for
some multireduction $\mreda_0$, such that $t \mred{\mreda_0} s''
\mstep{\residus{\setredc}{\mreda_0}} u'$ for some object $s''$, and moreover $\mreda_0 \freefrom
\setredb$, $\residus{\setredc}{\mreda_0} \dominated \residus{\setredb}{\mreda_0}$, and
$\mredms{\mreda_0} = \mredms{\mreda'_0}$.
We can thus build the following diagram:
\begin{center}
$\xymatrix@C=45pt{ 
t \arMultiOp{r}{>>}^{\mreda_0} \arMulti{d}_{\setredc\,}
& s'' \arMulti{d}^{\,\residus{\setredc}{\mreda_0}}  \\
s \arMultiOp{r}{>>}_{\mreda'_0} & u' \arMulti{r}_{\setreda'} & u
}$
\end{center}

On the other hand, $\mreda' \freefrom \residus{\setredb}{\setredc}$ implies $\setreda' \freefrom \residus{\setredb}{\setredc ; \mreda'_0}$ (\confer\ again \thelem~\ref{rsl:free-from-parts-mredseq}\delia{, now on $\mreda'_0 ; \setreda' ; \nil_u$}), therefore \theprop~\ref{rsl:SOredexset}:(\ref{it:SOmredseq}) yields $\setreda' \freefrom \residus{\setredb}{\mreda_0; \residus{\setredc}{\mreda_0}} = \residus{\residus{\setredb}{\mreda_0}}{\residus{\setredc}{\mreda_0}}$.
Moreover, $\setredb$ \ngrip\ implies $\residus{\setredb}{\mreda_0}$ \ngrip.
Hence we can apply Lem.~\ref{rsl:free-setred-from-tgt-to-src} to $s'' \mred{\residus{\setredc}{\mreda_0}} u'$, obtaining that $\setreda' = \residus{\setreda}{\residus{\setredc}{\mreda_0}}$ for some $\setreda \subseteq \ROccur{s''}$ verifying $\setreda \freefrom \residus{\setredb}{\mreda_0}$ and $\depth{\setreda} = \depth{\setreda'}$. 
Consequently, we can complete the previous diagram as follows.

\smallskip\noindent
\begin{tabular}{@{}p{45mm}p{70mm}}
$\xymatrix@C=40pt{ 
t \arMultiOp{r}{>>}^{\mreda_0} \arMulti{d}_{\setredc\,}
& s'' \arMulti{r}^{\,\setreda} \arMulti{d}_{\,\residus{\setredc}{\mreda_0}} 
& s' \arMulti{d}_{\,\residus{\setredc}{\mreda_0; \setreda}} \\
s \arMultiOp{r}{>>}_{\mreda'_0} & u' \arMulti{r}_{\setreda'} & u
}$
&
We define $\mreda \eqdef \mreda_0 ; \setreda$. Given $\residus{\setredc}{\mreda_0} \dominated \residus{\setredb}{\mreda_0}$ and $\setreda \freefrom \residus{\setredb}{\mreda_0}$, so that $\setreda \,\cap\, \residus{\setredb}{\mreda_0} = \emptyset$, Lem.~\ref{rsl:dominated-stable-contraction} applied on $\setreda$ implies $\residus{\setredc}{\mreda} \dominated \residus{\setredb}{\mreda}$. 
\end{tabular}

\noindent
Moreover, given $\mreda_0 \freefrom \setredb$ and $\setreda \freefrom \residus {\setredb}{\mreda_0}$, a simple induction on $\rlength{\mreda_0}$ yields $\mreda \freefrom \setredb$. 
Finally, $\mredms{\mreda} = \mredms{\mreda'}$ is immediate. Thus we conclude.
\end{proof}

The postponement result is used to show that, whenever $t \mred{\mreda} u$ and $\setredb \subseteq
\ROccur{t}$ is \ngrip\ and not used in $\mreda$, and $\residus{\setredb}{\mreda} = \emptyset$, all
activity \dominbytext\ (the successive residuals of) $\setredb$ is irrelevant, \ie\ it can be omitted
without compromising the target object $u$, and moreover without increasing the measure.  Therefore,
the \domintext\ part of each \tredexset\ in $\setreda_2 ; \mreda''$ can just be discarded in the
construction of $\mreda_{k+1}$, \confer\ Fig.~\ref{fig:normalisation-proof-iteration-3} on
page~\pageref{fig:normalisation-proof-iteration-3}.

\begin{lemma}[\delia{Irrelevance of one multistep}]
\label{rsl:dominated-erasure}
Let $t \mstep{\setredc} s \mred{\mreda'} u$ and $\setredb \subseteq \ROccur{t}$, such that $\setredb$ is \ngrip, $\setredc \dominated \setredb$, $\mreda' \freefrom \residus{\setredb}{\setredc}$, and $\residus{\setredb}{\setredc; \mreda'} = \emptyset$.
Then there is a multireduction $\mreda$ such that $\mreda' = \residus{\mreda}{\setredc}$, $t \mred{\mreda} u$, $\mreda \freefrom \setredb$, $\residus{\setredb}{\mreda} = \emptyset$ and $\mredms{\mreda} = \mredms{\mreda'}$.
\end{lemma}

\begin{proof}
Lem.~\ref{rsl:postponement} implies the existence of $\mreda$ such that $\mreda' = \residus{\mreda}{\setredc}$, $t \mred{\mreda} s' \mstep{\residus{\setredc}{\mreda}} u$, $\mreda \freefrom \setredb$, $\residus{\setredc}{\mreda} \dominated \residus{\setredb}{\mreda}$, and $\mredms{\mreda} = \mredms{\mreda'}$.
Then $\residus{\residus{\setredb}{\mreda}}{\residus{\setredc}{\mreda}} = \residus{\setredb}{\mreda; \residus{\setredc}{\mreda}} = \residus{\setredb}{\setredc; \mreda'} = \emptyset$; \confer\ \theprop~\ref{rsl:SOredexset}:(\ref{it:SOmredseq}).
We will show that $\residus{\setredb}{\mreda} = \emptyset$, and also that $\residus{\setredc}{\mreda} = \emptyset$ implying $t \mred{\mreda} u$.

Assume for contradiction the existence of some $b \in \residus{\setredb}{\mreda}$, we assume wlog that $b$ is minimal in $\residus{\setredb}{\mreda}$ \wrt\ $<$ (recall that $<$ is assumed well-founded).
Then $\residus{\setredc}{\mreda} \dominated \residus{\setredb}{\mreda}$ implies $b \freefrom \residus{\setredc}{\mreda}$, so that Lem.~\ref{rsl:free-then-unique-residual} yields $\residus{b}{\residus{\setredc}{\mreda}} = \set{b'}$, contradicting $\residus{(\residus{\setredb}{\mreda})}{\residus{\setredc}{\mreda}} = \emptyset$.
Therefore $\residus{\setredb}{\mreda} = \emptyset$.
In turn, the existence of some $c \in \residus{\setredc}{\mreda}$ would imply the existence of some $b \in \residus{\setredb}{\mreda}$ such that $b < c$, contradicting \residus{$\setredb}{\mreda} = \emptyset$. Therefore $\residus{\setredc}{\mreda} = \emptyset$, implying $u = s'$ so that $t \mred{\mreda} u$. 
Thus we conclude.
\end{proof}

\begin{lemma}[\delia{Irrelevance of many multisteps}]
\label{rsl:free-equiv-mred}
Let $t \mred{\mreda} u$ and $\setredb \subseteq \ROccur{t}$, such that $\setredb$ is \ngrip, $\mreda$ does not use $\setredb$, and $\residus{\setredb}{\mreda} = \emptyset$.
Then there exists a multireduction $\mredb$ such that $t \mred{\mredb} u$, $\mredb \freefrom \setredb$, $\residus{\setredb}{\mredb} = \emptyset$ and $\mredms{\mredb} \leqms \mredms{\mreda}$.
\end{lemma}

\begin{proof}
By induction on $\rlength{\mreda}$.
If $\mreda = \nil_t$, then it suffices to take $\mredb \eqdef \mreda$.

Assume $\mreda = \setreda; \mreda_0$, so that $t \mstep{\setreda} s \mred{\mreda_0} u$ for some object $s$. 
Observe $\residus{\setredb}{\setreda}$ is \ngrip, $\mreda_0$ does not use $\residus{\setredb}{\setreda}$ and $\residus{\residus{\setredb}{\setreda}}{\mreda_0} = \residus{\setredb}{\mreda} = \emptyset$. 
Then we can apply the \ih\ on $s \mred{\mreda_0} u$, thus obtaining $s \mred{\mredb'_0} u$ for some $\mredb'_0$ verifying $\mredb'_0 \freefrom \residus{\setredb}{\setreda}$, $\residus{\residus{\setredb}{\setreda}}{\mredb'_0} = \emptyset$ and $\mredms{\mredb'_0} \leqms \mredms{\mreda_0}$.

\delia{We partition $\setreda$ in its free and \domintext\ parts \wrt\ $\setredb$, according to the idea described in Sec.~\ref{sec:free-domin} after the definition of free and \domintext\ \tredexsets. Formally, we} define $\thefree{\setreda} \eqdef \set{a \in \setreda \sthat a \freefrom \setredb}$ and $\thedomin{\setreda} \eqdef \residus{(\setreda \setminus \thefree{\setreda})}{\thefree{\setreda}}$, so that $t \mstep{\thefree{\setreda}} t' \mstep{\thedomin{\setreda}} s \mred{\mredb'_0} u$ for some object $t'$.
It is easy to check $\thefree{\setreda} \freefrom \setredb$ and $(\setreda \setminus \thefree{\setreda}) \dominated \setredb$; since $\setreda$ does not use $\setredb$ and then $\setreda \cap \setredb = \emptyset$.
As moreover $\thefree{\setreda} \cap \setredb = \emptyset$, then Lem.~\ref{rsl:dominated-stable-contraction} yields $\thedomin{\setreda} \dominated \residus{\setredb}{\thefree{\setreda}}$.
Observe that $\setredb$ \ngrip\ implies $\residus{\setredb}{\thefree{\setreda}}$ \ngrip,
$\mredb'_0 \freefrom \residus{\setredb}{\setreda} = \residus{\residus{\setredb}{\thefree{\setreda}}}{\thedomin{\setreda}}$,
and $\residus{\residus{\setredb}{\thefree{\setreda}}}{\thedomin{\setreda}; \mredb'_0} = \residus{\residus{\setredb}{\setreda}}{\mredb'_0} = \emptyset$; \confer\  \theprop~\ref{rsl:SOplus}.
Therefore Lem.~\ref{rsl:dominated-erasure} applies to $t' \mstep{\thedomin{\setreda}} s \mred{\mredb_0'} u$, implying the existence of some $\mredb_0$ verifying $t' \mred{\mredb_0} u$, $\mredb_0 \freefrom \residus{\setredb}{\thefree{\setreda}}$, $\residus{\residus{\setredb}{\thefree{\setreda}}}{\mredb_0} = \emptyset$ and $\mredms{\mredb_0} = \mredms{\mredb'_0} \leqms \mredms{\mreda_0}$.
Hence we conclude by taking $\mredb \eqdef \thefree{\setreda}; \mredb_0$
since  $\thefree{\setreda} \subseteq \setreda$ implies in particular  that  $\depth{\thefree{\setreda}} \leq \depth{\setreda}$.
\end{proof}

\delia{
\medskip
The following propositions describe the construction of the \mredseq\ $\mreda_{k+1}$
(\confer\ Fig.~\ref{fig:normalisation-proof-iteration-3} on
page~\pageref{fig:normalisation-proof-iteration-3}). 
In terms of the general proof structure described at the beginning of \thesec~\ref{sec:reduction-strategies}, we can consider $t_k$, $t_{k+1}$ and $\strs(t_k)$ as $t$, $s$ and $\setredb$ respectively in the statement of Prop.~\ref{rsl:sekar-rama-8} and Prop.~\ref{rsl:sekar-rama-10}. and $\mreda_k$ as $\mreda$ in the latter Proposition.
}

\begin{proposition}[\delia{Projection over non-used \tredexset}]
\label{rsl:sekar-rama-8}
Let $t \mred{\mreda} u$ and $\setredb \subseteq \ROccur{t}$ s.t.  \setredb\ is \ngrip, \mreda\ does not use \setredb, $\residus{\setredb}{\mreda} = \emptyset$ and $t \mstep{\setredb} s$.
Then there exists a multireduction $\mredb$ s.t. $s \mred{\mredb} u$ and $\mredms{\mredb} \leqms \mredms{\mreda}$.
\end{proposition}

\begin{proof}
Lem.~\ref{rsl:free-equiv-mred} implies the existence of some $\mredb_0$ such that $t \mred{\mredb_0} u$, $\mredb_0 \freefrom \setredb$, $\residus{\setredb}{\mredb_0} = \emptyset$ and $\mredms{\mredb_0} \leqms \mredms{\mreda}$.
We define $\mredb \eqdef \residus{\mredb_0}{\setredb}$. Then we can build the following diagram; \confer\ \theprop~\ref{rsl:SOredexset}(2).

$$ 
\xymatrix@C=48pt{
t \arMulti{r}^{\mredb_0} \arMulti{d}_{\setredb \ } &
u \ar@{=}[d] 
%u \ar@{=}[d] \arMulti{d}^{\ \residus{\setpb}{\mredb_0} = \emptyset}
\\
s \arMulti{r}_{\mredb} & 
u
}
$$ 

Lem.~\ref{rsl:residual-same-depth-mred} implies $\mredms{\mredb} = \mredms{\mredb_0} \leqms \mredms{\mreda}$. Thus we conclude.
\end{proof}

\begin{proposition}[\delia{Projection over used \tredexset}]
\label{rsl:sekar-rama-10}
Let $t \mred{\mreda} u$ and $\setredb \subseteq \ROccur{t}$, s.t. $\setredb$ is \ngrip, \mreda\ uses $\setredb$, $\residus{\setredb}{\mreda} = \emptyset$ and $t \mstep{\setredb} s$.
Then there exists a multireduction $\mredb$ such that $s \mred{\mredb} u$ and $\mredms{\mredb} \ltms \mredms{\mreda}$.
\end{proposition}

\begin{proof}
The hypotheses indicate $\mreda$ uses $\setredb$, therefore the ``last'' element of $\mreda$ which uses the corresponding residual of $\setredb$ can be determined, \ie\
$\mreda$ can be written as $\mreda_1; \setreda; \mreda_2$, such that $\setreda$ uses $\residus{\setredb}{\mreda_1}$ (\ie\ $\setreda \cap \residus{\setredb}{\mreda_1} \neq \emptyset$) and $\mreda_2$ does not use $\residus{\setredb}{\mreda_1; \setreda}$. Observe $\rlength{\mreda} = \rlength{\mreda_1} + \rlength{\mreda_2} + 1$.

Let $\setredb' \eqdef \residus{\setredb}{\mreda_1}$, $\setreda_1 \eqdef \setreda \cap \setredb'$, and $\setreda_2 \eqdef \residus{(\setreda \setminus \setreda_1)}{\setreda_1}$.
Observe that $\setreda_1 \neq \emptyset$. %
\delia{To verify that $\depth{\setreda_2} < \depth{\setreda}$, let $\reda$, $\redb$ such that $\reda \develops \setreda_1$, $\redb \develops \setreda_2$, and particularly $\rlength{\redb} = \depth{\setreda_2}$. Observe $\reda ; \redb \develops \setreda$. We obtain $\rlength{\reda} > 0$ since $\setreda_1 \neq \emptyset$. Then $\depth{\setreda} \geq \rlength{\reda ; \redb} > \depth{\setreda_2}$. 

}
Therefore, $\mredms{\setreda_2; \mreda_2} \ltms \mredms{\setreda; \mreda_2}$.
Moreover $\residus{\setreda_1}{\setredb'} = \emptyset$.
We can build the following diagram:
\begin{center}
$ \xymatrix@C=40pt{ t
    \arMultiOp{r}{>>}^{\mreda_1} \arMulti{d}_{\setredb} &
    t_0 \arMulti{r}^{\setreda_1} \arMulti{d}_{\setredb'} &
    t_1 \arMulti{r}^{\setreda_2} \arMulti{d}_{\residus{\setredb'}{\setreda_1}} & t_2
    \arMultiOp{r}{>>}^{\mreda_2} & u \\ s & s_0
    \ar@{=}[r] & s_0 
		} 
$
\end{center}
Suppose $\setreda_2
  $ uses $\residus{\setredb'}{\setreda_1}$. Notice that the existence of some $b' \in \setreda_2 \cap \residus{\setredb'}{\setreda_1}$
  would in turn imply the existence of some $b_1\in\setredb'$ s.t. $\residu{b_1}{\setreda_1}{b'}$
  and also the existence of some $b_2\in \setreda\setminus\setreda_1$
  s.t. $\residu{b_2}{\setreda_1}{b'}$.
	%\eccarlos{
	%By \axiomuse{\axAncestorUniqueness} }{}
  %\edu{(strictly speaking, we resort to the obvious extension of this axiom  to derivations: for any $b_1,b_2,b',\reda$,
    %$\residu{b_1}{\reda}{b'}$ and $\residu{b_2}{\reda}{b'}$ implies $b_1=b_2$)}  \eccarlos{$b_1=b_2$,}
		Consider an arbitrary $\reda \develops \setreda_1$. Then, by a simple induction on
                $\rlength{\reda}$ and resorting to \axiomuse{\axAncestorUniqueness}, one deduces $b_1=b_2$.
  Therefore
  $b_1=b_2\in \setredb'\cap(\setreda\setminus\setreda_1)$. But then, by definition of $\setreda_1$,
  $b_1=b_2\in\setreda_1$, which is absurd.  Therefore $\setreda_2
  $ does not use $\residus{\setredb'}{\setreda_1}$ and hence, since $\mreda_2$ does not use
  $\residus{\setredb}{\mreda_1; \setreda}$, 
  we obtain that $\setreda_2;\mreda_2$ does not use $\residus{\setredb'}{\setreda_1}$.
Moreover, $\setredb$ \ngrip\ implies $\residus{\setredb'}{\setreda_1}$ \ngrip.  Hence
Prop.~\ref{rsl:sekar-rama-8} yields the existence of some $\mredb_2$ verifying $s_0 \mred{\mredb_2}
u$ and $\mredms{\mredb_2} \leqms \mredms{\setreda_2 ; \mreda_2} \ltms \mredms{\setreda ;
  \mreda_2}$. 
Remark that, by definition of $\mredmsf$, $\rlength{\mredb_2} = \rlength{\setreda; \mreda_2} = \rlength{\mreda_2} + 1$.
\begin{center}
$ \xymatrix@C=40pt{ t
    \arMultiOp{r}{>>}^{\mreda_1} \arMulti{d}_{\setredb} &
    t_0 \arMultiOp{r}{>>}^{\setreda ; \mreda_2} \arMulti{d}_{\setredb'} & u \\
		s \arMultiOp{r}{>>}^{\residus{\mreda_1}{\setredb}} & 
		s_0 \arMultiOp{r}{>>}^{\mredb_2} & u
		} 
$
\end{center}
Thus if we define $\mredb \eqdef
\residus{\mreda_1}{\setredb} ; \mredb_2$, then $\rlength{\mredb} = \rlength{\mreda_1} + \rlength{\mreda_2} + 1 = \rlength{\mreda}$, and $\mredms{\mredb_2} \ltms \mredms{\setreda ; \mreda_2}$ implies $\mredms{\mredb} \ltms \mredms{\mreda}$ independently of the relative measures of $\residus{\mreda_1}{\setredb}$ and $\mreda_1$, since elements of multireductions are considered in \emph{reversed order} when building measures.
Thus we conclude.
\end{proof}

\medskip
\delia{As we already remarked, \theprop~\ref{rsl:sekar-rama-10} shows the existence of an adequate $\mreda_{k+1}$ following the general proof structure described at the beginning of \thesec~\ref{sec:reduction-strategies}. Therefore, we can prove the main result of this work.}

\begin{theorem}[\delia{The abstract normalisation result}]
\label{rsl:necessary-and-ngrip-then-normalises}
Let $\arsa = \langle \arsTerms, \arsRedexes, \arsSource, \arsTarget, \arsResidualRel, <, \grip
\rangle$ be an ARS enjoying all the axioms listed in Fig.~\ref{fig:summaryOfAxioms}. Repeated contraction of necessary and \ngrip\ \tredexsets\ on $\arsa$ normalises.
\end{theorem}

\begin{proof}
Let $t_0 \in \arsTerms$ a normalising object in $\arsa$.
Then there exists some multireduction $\mreda_0$ such that $t_0 \mred{\mreda_0} u$ where $u$ is a normal form.
We proceed by induction on $\mredms{\mreda_0}$, \ie\ using the well-founded ordering defined at the beginning of \thesec~\ref{sec:reduction-strategies}.
If $\mredms{\mreda_0}$ is minimal, \ie\ either $\mreda_0 = \nil_{t_0}$ or $\mreda_0 = \langle \emptyset_{t_0}, \ldots, \emptyset_{t_0} \rangle$, then $t_0$ is a normal form, and therefore there is nothing to prove.
Otherwise, let $\setredb$ be a necessary and \ngrip\ \tredexset\ such that $t_0 \mstep{\setredb} t_1$. Then $\mreda_0$ uses $\setredb$, and $u$ being a normal form implies $\residus{\setredb}{\mreda_0} = \emptyset$. 
Therefore Prop.~\ref{rsl:sekar-rama-10} implies the existence of a multireduction $\mreda_1$ such that $t_1 \mred{\mreda_1} u$ and $\mredms{\mreda_1} \ltms \mredms{\mreda_0}$.
The \ih\ on $\mreda_1$ suffices to conclude.
\end{proof}

%%% Local Variables: 
%%% mode: latex
%%% TeX-master: "paper"
%%% End: 

\section{Applications}
\label{sec:twoCaseStudies}
\subsection{The Pure Pattern Calculus (and the Simple Pattern Calculus)}
\label{sec:ppc}

\theppc\ is a pattern calculus which extends \thespc\ and stands out for the novel forms of
polymorphism it supports. Since arbitrary terms may be used as patterns and hence reduction inside
patterns is allowed, \theppc\ models \emph{pattern polymorphism} where functions over patterns that
are computed at runtime may be defined. Another language feature is \emph{path polymorphism},
which permits functions that are generic in the sense that they operate over arbitrary data
structures.
 
This section has four parts. We first present a brief overview of \theppc\
following~\cite{jk-jfp}. Then we show that \theppc\ fits the ARS framework, including all the
axioms.  The third part formulates a multistep strategy \strs. The final part shows that \strs\
computes necessary and \ngrip\ multisteps. In view of the results of the previous section,
\confer\ \thethm~\ref{rsl:necessary-and-ngrip-then-normalises}, these last three parts -- taken
together -- imply that \strs\ is normalising for this calculus.

\subsubsection{Overview of \theppc}
\label{sec:ppc-basic-elements}

%\ignore{
%\completar{The contents of \thesec~\ref{sec:ppc-basic-elements} has been taken textually from the RTA 2012 article, including some text which was indicated for the report version only, and adding the definition of (variable) occurrence (I indicate the added text by putting my color, \ie\ red). Maybe some changes should be applied. A possible addition is to give an explicit definition of $\Pos{t}$ and $\subtat{t}{p}$, explaining that for an abstraction $\lpth p.s$, positions inside $p$ and $s$ are given the prefixes 1 and 2 respectively.}
%
%} % ignore

%\subparagraph{Syntax:} 
Consider a countable set of \defi{symbols} $f, g, \ldots, x, y, z$.
Sets of symbols are denoted by meta-variables $\theta$, $\phi$, \ldots.
The  syntax  of  \theppc\ is summarised  by the following grammar: 
\[ \begin{array}{llccl}
\mbox{{\bf Terms}}           & (\capfterms) & t & :: = & x \mid \wid{x} \mid tt \mid \lpth\ t.t  \\
\mbox{{\bf Data-Structures}} & (\DStructs)  & D & :: = & \wid{x} \mid Dt \\
\mbox{{\bf Abstractions}}   & (\Abstract)   & A & :: = & \lpth\ t.t \\
\mbox{{\bf Matchable-forms}} & (\MForms)    & F & :: = & D \mid A 
\end{array} \]

The term $x$ is called a \defi{variable}, $\wid{x}$ a
\defi{matchable}, $tu$ an \defi{application} ($t$ is the
\defi{function} and $u$ the \defi{argument}) and $\lpth\ p.u$ an
\defi{abstraction} ($\theta$ is the set of \defi{binding symbols}, $p$
is the \defi{pattern} and $u$ is
the \defi{body}).  Application (resp. abstraction) is left
(resp. right) associative.  

 A
$\l$-abstraction $\l x.t$ can be defined by
$\lp{\set{x}}\ \wid{x}. t$. The \defi{identity function}
$\lp{\set{x}}\ \wid{x}. x$ is abbreviated $\Id$. \onlyReport{The notation
  $\tsize{t}$ is used for the \defi{size} of $t$, defined as
  expected. }

A binding symbol $x \in \theta$ of an abstraction $\lpth\ p.s$
\emph{binds} matchable occurrences of $x$ in $p$ 
and  variable occurrences of $x$ in $s$. The derived
notions of \defi{free variables} and \defi{free matchables} 
are respectively denoted by $\fv{\_}$ and $\fm{\_}$.
\onlyPdf{This is illustrated in Fig.~\ref{fig:free-bound}.}

\onlyPdf{\begin{wrapfigure}[5]{l}{0.4\textwidth}
\vspace{-1mm}
\minicenter{
\newcommand{\testigo}{}
\newcommand{\alturasub}{-0.1}
\newcommand{\alturanorm}{0}
\newcommand{\anchomult}{0.85}
\newcommand{\lebend}{90}
\begin{tikzpicture}%[remember picture,overlay,baseline=(primerlambda.base)]
  \node (formula) {$\lp{\set{x}} \ x \ \wid{x} \ . \ x \ \wid{x}$};
  \node (subxuno) at (-0.73 * \anchomult,\alturasub) {\testigo};
  \node (subxdos) at (-0.7 * \anchomult,\alturasub) {\testigo};
  \node (matx) at (0.17 * \anchomult,\alturanorm) {\testigo};
  \node (varx) at (0.7 * \anchomult,\alturanorm) {\testigo};
  \draw[->, thick] (subxuno) to[bend right=80] (matx);
  \draw[->, thick] (subxdos) to[bend right=\lebend] (varx);
\end{tikzpicture}
}
\caption{Binding in \theppc}\label{fig:free-bound}
\end{wrapfigure}
}
\onlyDvi{
\smallskip
\fbox{aca va un dibujo}
\smallskip
}

Formally, \defi{free
    variables} and \defi{free matchables} of terms are defined by:
  $\fv{x} := \set{x}$, $\fv{\wid{x}} := \ems$, $\fv{tu} := \fv{t} \cup
  \fv{u}$, $\fv{\lpth\ p.u} := (\fv{u} \sm \theta) \cup \fv{p}$,
  $\fm{x} := \ems$, $\fm{\wid{x}} := \set{x}$, $\fm{tu} := \fm{t} \cup
  \fm{u}$, $\fm{\lpth\ p.u} := (\fm{p} \sm \theta) \cup
  \fm{u}$. 
As usual, we consider
  terms up to \defi{alpha-conversion}, \ie\ up to renaming of
  bound matchables and variables.     
\defi{Constructors} are  matchables which are not bound
  and, to ease the presentation,
  they are often denoted in typewriter fonts $\ca,\cb,\cc,\cd, \ldots$,  
  thus for example $\lp{\set{x,y}}\ \wid{x}\ y\ \ca.y$
  denotes $\lp{\set{x,y}}\ \wid{x}\ y\ \wid{z}.y$.
  The distinction between matchables and variables is 
unnecessary for standard  (static) patterns which do not contain
free variables.

A \defi{position} is either $\epsilon$ (the empty position), or $na$,
where $n \in \set{ 1, 2 }$ and $a$ is a position.  We use $a, b, \ldots$ 
%\eccarlos{(resp. $\setpa, \setpb, \ldots$ and $\delta, \rho, \pi, \ldots$)}{} 
to denote positions.
%\eccarlos{(resp. sets and sequences of positions) and $b \setpa$ to mean $\set{ b a \mid a \in \setpa }$}{}.  
The \defi{set $\Pos{t}$ of positions} of $t$ is defined as expected, provided that for abstractions
$\lpth\ p.s$ positions inside both $p$ and $s$ are considered.  Here is
an example $\Pos{\lp{\set{x}}\ \ca\ \wid{x}.\ca\ x\ x} =
\set{\epsilon, 1, 2, 11, 12, 21, 22, 211, 212}$. We write $a \leq b$ (resp. $a \disj b$) when the position $a$ is a
\defi{prefix} of (resp. \defi{disjoint} from) the position $b$.
Notice that $a \disj b$ and $a \leq c$ imply $c \disj b$.
All these notions are defined as expected~\cite{baader-nipkow}
and extended to sets of positions as well.  
In particular, given a position $a$ and a set of positions \setpb, we will say that $a \leq \setpb$ iff $a \leq b$ for all $b \in \setpb$, and analogously for $<$, $\disj$, etc..
\ignore{
Particularly, given two
sets of positions \setpa\ and \setpb, they are said to be
\defi{\setdisjoint}, written $\setpa \disj \setpb$,  iff $\metaforall a \in \setpa\ \metaforall b \in
\setpb\ a \disj b$. 
}

We write $\subtat{t}{a}$ for the \defi{subterm of $t$ at position $a$}
and $t[s]_a$ for the \defi{replacement} of the subterm at position $a$
in $t$ by $s$. Finally, we write $s\subterm t$ if $s$ is a
subterm of $t$ (note in particular $s\subterm s$).
 Notice that replacement may capture variables.  An
\emph{occurrence} of a term $s$ in a term $t$ is any position $p \in
\Pos{t}$ verifying $\subtat{t}{p} = s$. Particularly, \emph{variable
  occurrences} are defined this way.

%\onlyReport{We use the symbol $\setminus$ to denote set difference.}

\subparagraph{Substitution and Matching.}
A \defi{substitution} $\sigma$ is a mapping from variables to terms with
finite domain $\dom(\sigma)$. We write $\set{x_1 \rewto t_1, \ldots,
  x_n \rewto t_n}$ for a substitution with domain $\set{x_1, \ldots,
  x_n}$. 
A \defi{match} $\mu$ is either a substitution or a special
constant in the set $\set{\fail, \wait}$.  A \defi{match} is
\defi{positive} if it is a substitution; it is \defi{decided} if it is
either positive or \fail. 
The set of free variables of a match $\mu$ are defined as follows: $\fv{\sig} = \bigcup_{x \in \dom(\sig)} \fv{\sig x}$, $\fv{\fail} = \ems$ and $\fv{\wait}$ is undefined. Similarly for $\fm{\mu}$. 
We also define $\dom(\fail) = \ems$, whereas $\dom(\wait)$ is undefined.
The \defi{symbols} of $\mu$ are $\sym{\mu}: = \dom(\mu) \cup \fv{\mu} \cup \fm{\mu}$.  A set of symbols $\theta$ \defi{avoids} a match $\mu$, written $\theta \avoids \mu$, iff $\forall x \in \theta, x \notin
\sym{\sig}$.  
The \defi{application of a substitution} $\sig$ to a
term is written and defined as usual on alpha-equivalence classes; in particular $\sig (\lpth\ p.s) := \lpth\ \sigma(p).\sigma(s)$, if $
\theta \avoids \sigma$.
Notice that data structures and matchable forms are
stable by substitution.
The \defi{application of a match} $\mu$ to a
term $t$, written $\mu t$, is defined as follows: if $\mu$ is a
substitution, then it is applied as explained above; if $\mu = \wait$, then
$\mu t$ is undefined; if $\mu = \fail$, then $\mu t$ is the identity
function $\Id$. Other \emph{closed terms in normal form} could be taken to
define the last case, this one allows in particular to encode
pattern-matching definitions given by alternatives~\cite{jk-jfp}.

The \defi{restriction} of a substitution $\sig$ to a set of
  variables $\{x_1,\ldots,x_n\}\subseteq\dom(\sig)$ is written
  $\sig|_{\{x_1,\ldots,x_n\}}$. This notion is extended to matchings
  by defining $\wait|_{\{x_1,\ldots,x_n\}}=\wait$ and
  $\fail|_{\{x_1,\ldots,x_n\}}=\fail$, for any set of variables
  $\{x_1,\ldots,x_n\}$. 
	The \defi{composition} $\sig \circ \eta$ of two substitutions
  $\sig$ and $\eta$ is defined by $(\sig \circ \eta)x = \sig(\eta
  x)$. Furthermore, if $\mu_1$ and $\mu_2$ are
  matches of which at least one is $\fail$, then $\mu_2 \circ \mu_1$
  is defined to be $\fail$.  Otherwise, if $\mu_1$ and $\mu_2$ are
  matches of which at least one is $\wait$, then $\mu_2 \circ \mu_1$
  is defined to be $\wait$.  Thus, in particular, $\fail \circ \wait$
  is $\fail$.

The \defi{disjoint union} of two matches $\mu_1$ and $\mu_2$ is as in \thespc. In particular, the
equation from \thespc\ also holds
\begin{center}
$\fail  \uplus \wait  = \wait
  \uplus\fail =  \fail$
\end{center}
and is the culprit for the non-sequential nature of \theppc\ (just as in \thespc)\footnote{Sequentiality can be recovered (see \eg~\cite{Jaybook,Balhor10,Balppdp10}) by simplifying the equations of disjoint union, however, some meaningful terms will no longer be normalising. 
\Eg\ if in particular $\wait \uplus \fail = \wait$, then 
then $(\lp{\ems}\ \ca\ \cb\ \cb\ . \wid{y}) (\ca\ \Om\ \cc)$, where $\Om$ is a non-terminating term, would never fail as we expect.}

\ignore{  % ya introducido en el marco de spc
The \defi{disjoint union} of two matches $\mu_1$ and $\mu_2$ is a crucial operation
used to define the
operational semantics of \theppc.  Disjoint union is written $\mu_1
  \uplus  \mu_2$ and  is  defined  as:  their union  if  both $\mu_i$  are
  substitutions  and  $\dom(\mu_1)   \cap  \dom(\mu_2)  =  \ems$;
  \wait\  if  either of  the  $\mu_i$ is  \wait\  and  none is  \fail;
  \fail\  otherwise. 
This definition of disjoint union of matches validates the following equations which are responsible for the non-sequential nature of \theppc:
\begin{center}
$\fail  \uplus \wait  = \wait
  \uplus\fail =  \fail$
\end{center}
We return to these equations immediately after the definition of the
operational semantics of the calculus.
} % ignore

The \defi{compound matching operation} takes a term, a set of binding symbols 
and a pattern and returns a match, it is defined by applying  the following 
equations in order: 
\begin{center}
$\begin{array}{rcll}
\cmatchOpth{t}{\wid{x}}        & := & \subst{x}{t} & \textif x \in \theta \\
\cmatchOpth{\wid{x}}{\wid{x}}  & := & \emptySubst  & \textif x \notin \theta \\
\cmatchOpth{tu}{pq}         & := & \cmatchOpth{t}{p} \uplus \cmatchOpth{u}{q} 
			                       & \textif tu, pq \in \MForms  \\
\cmatchOpth{t}{p}           & := & \fail
			                       & \textif p,t \in \MForms \\
\cmatchOpth{t}{p}           & := & \wait
			                       & \textnormal{otherwise} \\
\end{array}$
\end{center}

The use of disjoint union in the third case
of the previous definition restricts compound matching to linear patterns, as in \thespc.
\ignore{ % ya dicho en spc
~\footnote{A pattern $p$ is
  linear w.r.t. $\theta$ if for every $x$ 
  in $\theta$, the matchable $\wid{x}$ appears at most once in $p$.}, 
which is  necessary to guarantee confluence. Indeed, disjoint union 
of two substitutions fails whenever their 
domains are not disjoint. 
} % ignore
The result of the \defi{matching operation}\footnote{Note that the notation for (compound) matching we have just given differs from \cite{jk-ppc} and \cite{jk-jfp}: the pattern and argument appear in reversed order there.}
$\matchOpth{t}{p}$ 
is defined to be the \emph{check} of $\cmatchOpth{t}{p}$ on $\theta$; 
where the \defi{check} of a match $\mu$  on $\theta$ 
is \fail\ if  $\mu$ is a substitution whose domain is not $\theta$, 
$\mu$ otherwise. Notice that $\matchOpth{t}{p}$ is never positive if
$p$ is not linear with respect to  $\theta$.
We now give some examples: 
$\matchOpt{\set{x}}{uv}{\wid{x}\wid{x}}$ gives $\fail$
because $\wid{x}\wid{x}$ is not linear; 
$\matchOpt{\set{x,y,z}}{uv}{\wid{x}\wid{y}}$ gives $\fail$
because $\set{x,y,z} \neq \set{x,y}$, 
$\matchOpt{\ems}{u}{\wid{x}}$ gives $\fail$
because $\ems \neq \set{x}$;
$\matchOpt{\set{x}}{\wid{y}}{\wid{y}}$ gives $\fail$
because $\set{x} \neq \ems$;
$\matchOpt{\set{x}}{u \wid{z}}{\wid{x}\wid{y}}$ gives $\fail$
because $\cmatchOp{\set{x}}{\wid{z}}{\wid{y}}$ is $\fail$; 
$\matchOpt{\ems}{u \wid{z}}{\wid{x}\wid{y}}$ gives $\fail$
since both $\matchOpt{\ems}{u}{\wid{x}}$ and $\cmatchOp{\ems}{\wid{z}}{\wid{y}}$ are $\fail$.

%%% Local Variables: 
%%% mode: latex
%%% TeX-master: "article"
%%% End: 

\subsubsection{\theppc\ as an ARS} 
\label{sec:ppc-ars}

\theppc\ can be described as an ARS. Its objects $\arsTerms$ are the \emph{terms} of \theppc. The
\emph{\tredexes} are the pairs $\pair{t}{a}$ where $t$ is a term, $a \in \Pos{t}$, $\subtat{t}{a} =
(\lpth p.s)u$, and $\matchOpth{u}{p}$ is decided. In this case $\rsrc{\pair{t}{a}} \eqdef t$ and $\rtgt{\pair{t}{a}} \eqdef t[\matchOpth{u}{p} s]_a$.
If $\matchOpth{u}{p} = \fail$, then we say that the \tredex\ is a \defi{matching failure}.
\delia{We will often
denote by $\ppcstepa$ a given step $\pair{t}{a}$; analogously, 
we will often denote by $\ppcsetred{D}$
the set $\set{\pair{t}{d} \mid d \in D}$ where $D \subseteq \Pos{t}$.
Conversely, whenever $\ppcstepa$ is a step, we often refer to its position as $a$, even without specifying explicitly that $\ppcstepa = \pair{t}{a}$ for some term $t$, and similarly, whenever $\ppcsetred{D}$ is a set of steps, we refer to the corresponding set of positions as $D$.}
This notation shall prove convenient when we address the compliance of \theppc\ w.r.t. the axioms of an ARS. Regarding the relations over objects and steps:

\begin{itemize}

\item \textbf{Residual relation}.
\label{dfn:ppc-residual}
If $\ppcstep{a} = \pair{t}{a}$, $\ppcstep{b} = \pair{t'}{b}$ and $\ppcstep{b'} = \pair{u}{b'}$ are \tredexes, then $\residu{\ppcstep{b}}{\ppcstep{a}}{\ppcstep{b'}}$ iff $t' = t$, $u = \rtgt{\ppcstep{a}}$, and one of the following cases apply, where $\subtat{t}{a} = (\lpth p.s) u$: \\
\begin{tabular}{@{$\ \ \bullet\quad$}p{.8\textwidth}}
$a \not\leq b$ and $b' = b$. \\
$b = a12n$, $b' = an$ and $\matchOpth{u}{p} \neq \fail$. \\
$b = a2mn$, $b' = akn$, $\matchOpth{u}{p} \neq \fail$, and there is a variable $x \in \theta$ such that $\subtat{t}{a11m} = \subtat{p}{m} = \wid{x}$ and $\subtat{t}{a12k} = \subtat{s}{k} = x$.
\end{tabular}

\item \textbf{Embedding relation}. 
\delia{We define the embedding relation between redexes as the \emph{tree order}~\cite{thesis-mellies}.} 
Namely,  $\ppcstep{a} < \ppcstep{b}$ iff $\ppcstep{a} = \pair{t}{a}$, $\ppcstep{b} = \pair{t}{b}$, and $a < b$.
Notice that 
whenever $\ppcstep{a} < \ppcstep{c}$ and $\ppcstep{b} < \ppcstep{c}$, then $\ppcstep{a}$ and $\ppcstep{b}$ are comparable \wrt\ the embedding,
\ie\  either $\ppcstep{a} = \ppcstep{b}$, $\ppcstep{a} < \ppcstep{b}$ or $\ppcstep{b} < \ppcstep{a}$.

\label{page-gripping-ppc}
\item \textbf{Gripping relation}.
\ignore{
The intent of the \emph{gripping} relation is exactly to prevent the phenomenon described in \thesec~\ref{sec:gripping}, taking into account that an abstraction $\lpth p.s$ in \theppc\ bounds the \emph{set} of variables $\theta$.
This observation leads to the following definition: 
} % ignore
Let $\ppcstep{a} = \pair{t}{a}$ and $\ppcstep{b} = \pair{t}{b}$ be \tredexes\ and let $\subtat{t}{a} = (\lpth p.s) u$. 
Then $\ppcstep{a} \grip \ppcstep{b}$ iff $\matchOpth{u}{p} \neq \fail$, $b = a12n$, and $\theta \cap \fv{\subtat{s}{n}} \neq \emptyset$.
%\completar{Add at least one example of gripping}.
\end{itemize}

We now address the axioms of Fig.~\ref{fig:summaryOfAxioms}. A word on
notation: if $t$ and $\theta$ are a term and a set of symbols
respectively, then we will write $\isbm{t}{\theta}$ when $t = \wid{x}$
for some $x \in \theta$.

% ignore

\ignore{
\subsection{Non-sequentiatlity of \theppc}
\completar{We can show the non-sequential behaviour of \theppc\ by means of an example. I think that the following sentence taken from the RTA 2012 article would be adequate here.}

It is worth noticing that sequentiality of  \theppc\ can be recovered (see \eg~\cite{Jaybook,Balhor10,Balppdp10}) by modifying the equations of disjoint union, however, some meaningful terms will no longer be normalising. 
Thus for example, if $\fail \uplus \mu$ is defined to be $\fail$, while $\wait \uplus \fail = \wait$ and $\sig \uplus \fail = \fail$, then $(\lp{\ems}\ \ca\ \cb\ \cb\ . \wid{y}) (\ca\ \Om\ \cc)$, where $\Om$ is a non-terminating term, would never fail as expected.  
}

%%% Local Variables: 
%%% mode: latex
%%% TeX-master: "article"
%%% End: 

%\input{ppc-axioms-fundamental}
%\input{ppc-axioms-embedding}
%\input{ppc-axioms-gripping}
%%%%%%%%%%%%%%%%%%%%%%%%%%%%%%%%%%%%%%%%%%%%%%%%
%%%%%      Fundamental axioms
\bigskip\noindent\textbf{Fundamental axioms.} \\
\axiomuse{\axSelfReduction} is immediate from the definition of residuals for \theppc: none of the cases there applies for $\residus{\ppcstep{a}}{\ppcstep{a}}$. 
\axiomuse{\axFiniteResiduals} follows from the fact that terms are finite. 
Axiom \axiomuse{\axAncestorUniqueness} is proved below.

\begin{lemma}[\axiomuse{\axAncestorUniqueness}]
Let $\ppcstep{b_1}, \ppcstep{b_2}, \ppcstep{a}, \ppcstep{b'}$ be \tredexes\ verifying $\residu{\ppcstep{b_1}}{\ppcstep{a}}{\ppcstep{b'}}$ and $\residu{\ppcstep{b_2}}{\ppcstep{a}}{\ppcstep{b'}}$. Then $\ppcstep{b_1} = \ppcstep{b_2}$.
\end{lemma}

\begin{proof}
Let $\ppcstep{b_1} = \pair{t}{b_1}$, $\ppcstep{b_2} = \pair{t}{b_2}$ and $\ppcstep{b'} = \pair{t'}{b'}$, where $t \sstep{\ppcstep{a}} t'$.  
We prove that $b_1 = b_2$. Let $\subtat{t}{a} = (\lpth p.s) u$. We consider three cases according to the definition of
$\residu{\ppcstep{b_1}}{\ppcstep{a}}{\ppcstep{b'}}$.
\begin{itemize}
\item If $a \not\leq b_1$, then $ b_1 = b'$ so that $a \not \leq b'$.  A straightforward case
  analysis on the definition of residuals yields $a \not\leq b_2$, therefore $b_1 = b_2 = b'$.

\item If $b_1 = a2mn$ and $b' = akn$, then $\subtat{s}{k} = x$ and $\subtat{p}{m} = \wid{x}$ for
  some $x \in \theta$.  Observe that $a < b'$ implies $a < b_2$. We consider two cases.  If $b_2 =
  a12n'$ and $b' = an'$, then $kn = n'$.  This would imply $\subtat{t}{b_2} = \subtat{s}{kn}$ has
  the form $(\lp{\theta'} p'.s') u'$, contradicting $\subtat{s}{k}$ being a variable.  Therefore,
  $akn = b' = ak'n'$ and $b_2 = a2m'n'$, where $\subtat{s}{k'} = y$ and $\subtat{p}{m'} = \wid{y}$
  for some $y \in \theta$.  Observe that $k < k'$, \ie\ $k' = kc$ where $c \neq \epsilon$, would
  imply $kc \in \Pos{s}$, contradicting the fact that $\subtat{s}{k}$ is a variable; so that $k
  \not< k'$. We obtain $k' \not< k$ analogously. On the other hand, $k \disj k'$ would contradict
  $kn = k'n'$. Hence $k = k'$, implying $n = n'$ and also $y = x$.  In turn, $\matchOpth{u}{p}$
  being positive implies that $p$ is linear, and then $m = m'$. Thus we conclude.

\item If $b_1 = a12n$ and $b' = an$, then we have again that $a < b'$ implies $a < b_2$.  On the
  other hand, assuming $b_2 = a2m'n'$, so that $an = b' = akn'$, would yield a contradiction as
  already stated. Therefore $b_2 = a12n'$ and $an = b' = an'$, implying $n = n'$ and consequently
  $b_1 = b_2$.
\end{itemize}
\vspace{-6mm}
\end{proof}

Finally, \axiomuse{\axFD} and \axiomuse{\axSO}  are left for the end of this section.

%%%%%%%%%%%%%%%%%%%%%%%%%%%%%%%%%%%%%%%%%%%%%%%%
%%%%%      Enclave-creation
\bigskip\noindent\textbf{The \axiomuse{\axEnclaveCreation} axiom.} \\
To verify \axiomuse{\axEnclaveCreation} involves a rather long technical development, including some preliminary lemmas, particularly a creation lemma indicating the creation cases for \theppc.
%One of these lemmas \edu{Cu'al?} is used in the proof of subsequent axioms as well. 

\begin{lemma}
\label{rsl:reduction-mantains-decided-match} 
Let $p \sred{} p'$ and $u \sred{} u'$. Then, 
\begin{enumerate}[(i)]
\minitem \label{it:reduction-mantains-positive-cmatch}
$\cmatchOpth{u}{p}$ positive implies $\cmatchOpth{u'}{p'}$ positive,
\minitem \label{it:reduction-mantains-fail-cmatch}
$\cmatchOpth{u}{p} = \fail$ implies $\cmatchOpth{u'}{p'} = \fail$.
\minitem \label{it:reduction-mantains-positive-match}
$\matchOpth{u}{p}$ positive implies $\matchOpth{u'}{p'}$ positive,
\minitem \label{it:reduction-mantains-fail-match}
$\matchOpth{u}{p} = \fail$ implies $\matchOpth{u'}{p'} = \fail$.
\end{enumerate}
\end{lemma}

\onlyPaper{
\begin{proof}
By induction and case analysis on $p$. \CompleteInReport .
\end{proof}
}

\onlyReport{
\begin{proof}
  We prove item~(\ref{it:reduction-mantains-positive-cmatch}).  Given $\cmatchOpth{u}{p}$ is
  positive, a straightforward induction on $p$ yields that $p$ is a normal form, implying $p' = p$.
  If $\isbm{p}{\theta}$, then $\cmatchOpth{u'}{p}$ is positive for any term $u'$.  If $p$ is a
  matchable and $\neg \isbm{p}{\theta}$, then $\cmatchOpth{u}{p}$ positive implies $u = p$, \ie\ $u$
  is a normal form, and therefore $u' = u$, which suffices to conclude.  Assume $p = p_1 p_2$. Then
  hypotheses imply $p \in \MForms$, $u = u_1 u_2 \in \MForms$, and $\cmatchOpth{u_i}{p_i}$ positive
  for $i = 1,2$. In turn, $u \in \MForms$ implies $u' = u'_1 u'_2$ and $u_i \sred{} u'_i$ for $i =
  1,2$. Hence, the \ih\ can be applied for each $u_i \sred{} u'_i$, which suffices to conclude.
  Finally, any other case would contradict $\cmatchOpth{u}{p}$ positive.

  We prove item~(\ref{it:reduction-mantains-fail-cmatch}).  Observe $\cmatchOpth{p}{u} = \fail$
  implies $p,u \in \MForms$, and therefore $p', u' \in \MForms$. Therefore, $p$ and $p'$ share their
  syntactic form (\ie\ they are either both matchables, both applications or both abstractions), and
  similarly for $u$ and $u'$.  If $p$ and $u$, and therefore $p'$ and $u'$, have different syntactic
  forms, or else if $p, p', u, u'$ are abstractions, then it suffices to observe that
  $\cmatchOpth{p'}{u'} = \fail$ for any such $p'$ and $u'$.  If $p, p', u, u'$ are matchables, then
  $p = p'$ and $u = u'$, thus we immediately conclude.  Assume $p = p_1 p_2$, $p' = p'_1 p'_2$, $u =
  u_1 u_2$ and $u' = u'_1 u'_2$. In this case, hypotheses imply $\cmatchOpth{u_i}{p_i} = \fail$ for
  some $i \in \set{1,2}$, and moreover $p, u \in \MForms$ imply $p_i \sred{} p'_i$ and $u_i \sred{}
  u'_i$. Therefore, we conclude by applying the \ih, and recalling that $\fail \uplus R = \fail$ for
  any possible $R$.

To prove items~(\ref{it:reduction-mantains-positive-match}) and
(\ref{it:reduction-mantains-fail-match}), we observe that a straightforward induction on $p$ yields that $\cmatchOpth{u}{p} = \sigma$ implies $dom(\sigma) = \fm{p}$, and therefore in this case $\matchOpth{u}{p}$ is positive iff $\theta = \fm{p}$, and $\matchOpth{u}{p} = \fail$ otherwise.
Recall also that $\cmatchOpth{u}{p}$ positive implies $p$ being a normal form, and then $p' = p$. 
For item~(\ref{it:reduction-mantains-positive-match}): $\matchOpth{u}{p}$ positive implies
$\cmatchOpth{u}{p} = \sigma$ where $\theta = \fm{p} = \fm{p'}$. On the other hand, item~(\ref{it:reduction-mantains-positive-cmatch}) just proved implies $\cmatchOpth{u'}{p'} = \sigma'$, which suffices to conclude.
For item~(\ref{it:reduction-mantains-fail-match}): assume $\matchOpth{u}{p} = \fail$. If
$\cmatchOpth{u}{p} = \fail$, then item~(\ref{it:reduction-mantains-fail-cmatch}) just proved implies
$\cmatchOpth{u'}{p'} = \fail$, thus we conclude. Otherwise, $\cmatchOpth{u}{p} = \sigma$ and $\sigma
\neq \fm{p} = \fm{p'}$, and item~(\ref{it:reduction-mantains-positive-cmatch}) just proved implies $\cmatchOpth{u'}{p'} = \sigma'$, which suffices to conclude.
\end{proof}
}

\begin{lemma}[Creation cases]
\label{rsl:ppc-creation}
Let $t \sstep{\ppcstep{a}} t'$, and $\residu{\emptyset}{\ppcstep{a}}{\ppcstep{b}}$, \ie\ $\ppcstep{b}$ is \emph{created} by (the contraction of) $\ppcstep{a}$.
Say $\subtat{t}{a} = (\lpth p.s) u$ and $\subtat{t'}{b} = (\lp{\theta'} p'.s') u'$.
Then one of the following holds: 
\begin{description}
\item [Case I.]
the contraction of $\ppcstep{a}$ 
%contributes to the creation of $b_r$ from below. In this case, 
%generates the abstraction $\lp{\theta'} p'.s'$, \ie, 
contributes to the creation of $\ppcstep{b}$ from below, \ie, 
$b \in \Pos{t}$, $a = b 1$ implying $\subtat{t}{b} = (\lpth p.s) u u'$, and either
\begin{enumerate}[(i)]
\item \label{it:create-below-var}
$s = x$ where $x \in \theta$ and $\wid{x}$ occurs in $p$, $\matchOpth{u}{p} = \sigma$, $\sigma x = (\lp{\theta'} p'.s')$.
\item \label{it:create-below-abstr}
$s = \lp{\theta'} p''.s''$, $\matchOpth{u}{p} = \sigma$, $p' = \sigma p''$, $s' = \sigma s''$.
\item \label{it:create-below-fail}
$\matchOpth{u}{p} = \fail$, $\lp{\theta'} p'.s' = I$.
\end{enumerate}
\item [Case II.]
the contraction of $\ppcstep{a}$ 
contributes to the creation of $\ppcstep{b}$ from above, \ie, 
$b = an$, $\subtat{s}{n} = xu''$, $\matchOpth{u}{p} = \sigma$, $\sigma x = (\lp{\theta'} p'.s')$, $u' = \sigma u''$. 
\item [Case III.]
The argument of a redex pattern becomes decided. We have three such situations:
\begin{enumerate}[(i)]
\item \label{it:create-decided-both}
$b = an$, $\subtat{s}{n} = (\lp{\theta'} p''.s'') u''$, $\matchOpt{\theta'}{u''}{p''} = \wait$, $\matchOpth{u}{p} = \sigma$, $p' = \sigma p''$, $s' = \sigma s''$ and $u' = \sigma u''$.
\item \label{it:create-decided-argument}
$a = b2n$, $\subtat{t}{b} = (\lp{\theta'} p'.s') u''$ and $\matchOpt{\theta'}{u''}{p'} = \wait$.
\item \label{it:create-decided-pattern}
$a = b11n$, $\subtat{t}{b} = (\lp{\theta'} p''.s') u'$ and $\matchOpt{\theta'}{u'}{p''} = \wait$.
\end{enumerate}
\end{description}
\end{lemma}

\begin{proof}
We proceed by comparing $a$ with $b$.
\begin{itemize}
\minitem
If $a \disj b$, then $\subtat{t}{b} = \subtat{t'}{b}$ so that $\residu{\pair{t}{b}}{\ppcstep{a}}{\ppcstep{b}}$, contradicting the hypotheses.

\minitem
Assume $a \leq b$, \ie\ $b = ac$. 

In this case, $\matchOpth{u}{p} = \fail$ would imply $\subtat{t'}{a} =
I$, contradicting $\subtat{t'}{b}$ being a redex. Then
$\matchOpth{u}{p} = \sigma$, implying $\subtat{t'}{b} = \subtat{\sigma
  s}{c}$.  Now the redex at position $c$ of $\sigma s$ is either entirely contained in $\sigma$ or otherwise it occurs at a non-variable position of $s$. Observe that $c = kn$, $\subtat{s}{k} = x$ and
$\subtat{t'}{b} = \subtat{\sigma x}{n}$ for some variable $x$ would
imply $\residu{\pair{t}{a2mn}}{\ppcstep{a}}{\ppcstep{b}}$ where $\subtat{p}{m} =
\wid{x}$. This is not possible since $\ppcstep{b}$ is created. Therefore $\subtat{s}{c} = t_1 u''$ and $\subtat{t'}{b} =
(\lp{\theta'} p'.s')u' = (\sigma t_1) \sigma u''$.  If $t_1$ is a
variable, so that $\sigma t_1 = \lp{\theta'} p'.s'$, then case II
applies, otherwise case III.(\ref{it:create-decided-both}) applies.

\minitem
Assume $b < a$. 

If $a = b1$, \ie\ $\subtat{t}{b} = (\lpth p.s) u u'$, then observe $\matchOpth{u}{p} s = \subtat{t'}{a} = \lp{\theta'} p'.s'$.
If $\matchOpth{u}{p} = \fail$, then case I.(\ref{it:create-below-fail}) applies.
If $s$ is a variable, then case I.(\ref{it:create-below-var}) applies.
Otherwise, $s$ is an abstraction, so that case I.(\ref{it:create-below-abstr}) applies.

If $b11 \leq a$, \ie\ $\subtat{t}{b} = (\lp{\theta'} p''.s') u'$, then observe $\residu{\emptyset}{\ppcstep{a}}{\ppcstep{b}}$ implies $\matchOpt{\theta'}{u'}{p''} = \wait$. Then case III.(\ref{it:create-decided-pattern}) applies. 
If $b2 \leq a$, a similar argument yields that case III.(\ref{it:create-decided-argument}) applies.

Finally, $b12 \leq a$ implies $\subtat{t}{b} = (\lp{\theta'} p'.s'') u'$, and $\subtat{t'}{b}$ being a redex implies $\matchOpt{\theta'}{u'}{p'}$ decided so that $\residu{\pair{t}{b}}{\ppcstep{a}}{\ppcstep{b}}$, contradicting the hypothesis.
\end{itemize}
\vspace{-6mm}
\end{proof}

\begin{lemma}
\label{rsl:ppc-turns-decided-then-outermost}
\hspace*{1pt} \\[-4mm]
\begin{enumerate}
\minitem 
\label{it:ppc-not-mf-to-mf-then-outermost}
Let $t \sstep{\ppcstep{a}} t'$ such that $t \notin \MForms$ and $t' \in \MForms$. Then $\ppcstep{a}$ is outermost.
\minitem 
\label{it:ppc-argument-match-turns-decided-then-outermost}
Let $t \sstep{\ppcstep{a}} t'$ such that $\cmatchOpth{t}{p} = \wait$ and $\cmatchOpth{t'}{p}$ is decided for some $\theta$, $p$. Then $\ppcstep{a}$ is outermost.
\minitem 
\label{it:ppc-pattern-match-turns-decided-then-outermost}
Let $p \sstep{\ppcstep{a}} p'$ such that $\cmatchOpth{t}{p} = \wait$ and $\cmatchOpth{t}{p'}$ is decided for some $\theta$, $t$. Then $\ppcstep{a}$ is outermost.
\end{enumerate}
\end{lemma}

\onlyPaper{
\begin{proof}
By induction on $t'$, $t$ and $p$, respectively.  We use item \ref{it:ppc-not-mf-to-mf-then-outermost} in the proof of the other items. \CompleteInReport .
\end{proof}
}

\onlyReport{
\begin{proof}
We prove item~\ref{it:ppc-not-mf-to-mf-then-outermost} by induction on $t'$.

If $t'$ is a variable or a matchable, then $a = \epsilon$, thus we conclude.

If $t'$ is an abstraction, then $a \neq \epsilon$ implies $t$ is  an abstraction contradicting $t \notin \MForms$. 
Thus $a = \epsilon$ and  we conclude. 

If $t' = t'_1 t'_2$, then $t' \in \MForms$ implies $t'_1 \in \DStructs$. We consider three cases.
(i) If $a = \epsilon$ then we  immediately conclude.
(ii) If $2 \leq a$, then we contradict $t \notin \MForms$.
(iii) If $1 \leq a$, \ie\ $a = 1a'$, then $t = t_1 t'_2$ and
$t_1 \sstep{\ppcstep{a'}} t'_1$. Observe that $t_1 \in \DStructs$ would
contradict $t \notin \MForms$, and $t_1 \in \Abstract$ would
imply $t'_1 \in \Abstract$, contradicting
$t'_1 \in \DStructs$. Therefore, $t_1 \notin \MForms$, and hence the \ih\ 
yields that $\pair{t_1}{a'}$ is outermost. We conclude by
observing that $t_1 \notin \MForms$ implies that $\pair{t}{\epsilon}$
is not a \tredex.

\medskip
We prove item~\ref{it:ppc-argument-match-turns-decided-then-outermost} by induction on $t$. 
Observe that $\cmatchOpth{t}{p} = \wait$ implies $\neg \isbm{p}{\theta}$. In turn, $\cmatchOpth{t'}{p}$ decided implies $p \in \MForms$, and moreover $\neg \isbm{p}{\theta}$ implies $t' \in \MForms$.
If $t \notin \MForms$ then item~\ref{it:ppc-not-mf-to-mf-then-outermost} suffices to conclude.
Therefore, assume $t \in \MForms$. In this case, $\cmatchOpth{t}{p} = \wait$ implies $p = p_1 p_2$, $t = t_1 t_2$, and $\cmatchOpth{t_i}{p_i} \neq \fail$ for $i = 1,2$. 
Furthermore, $t \in \MForms$ implies $a \neq \epsilon$. 
Assume $a = 1 a'$, implying $t' = t'_1 t_2$ and $t_1 \sstep{\ppcstep{a'}} t'_1$. 
Notice that $\cmatchOpth{t_1}{p_1}$ decided would imply $\cmatchOpth{t_1}{p_1}$ positive (since it is not $\fail$), and then $\cmatchOpth{t'_1}{p_1}$ positive by \thelem~\ref{rsl:reduction-mantains-decided-match}; therefore, either possibility for $\cmatchOpth{t_2}{p_2}$ (given that it is not $\fail$) would contradict some hypothesis.
Moreover, $\cmatchOpth{t'_1}{p_1} = \wait$ would contradict $\cmatchOpth{t'}{p}$ decided (again, since $\cmatchOpth{t_2}{p_2} \neq \fail$).
%$\cmatchOpth{t_1}{p_1}$ positive would imply $\cmatchOpth{t'_1}{p'_1}$ positive, \confer\ \thelem~\ref{rsl:reduction-mantains-decided-match}, thus contradicting some hypothesis, while $\cmatchOpth{t'_1}{p'_1} = \wait$ would contradict $\cmatchOpth{t'}{p}$ decided. 
Hence \ih\ can be applied to obtain $\pair{t_1}{a'}$ outermost, which suffices to conclude (given $\pair{t}{\epsilon}$ not being a \tredex). The case $a = 2a'$ admits an analogous argument.

\medskip
We prove item~\ref{it:ppc-pattern-match-turns-decided-then-outermost} by induction on $p$. 
Observe that $\cmatchOpth{t}{p'}$ decided implies $p' \in \MForms$. If $p \notin \MForms$ then item~\ref{it:ppc-not-mf-to-mf-then-outermost} suffices to conclude. Therefore, assume $p \in \MForms$. This implies $p'$ is not a matchable, and consequently $\cmatchOpth{t}{p'}$ decided implies $t \in \MForms$.
In turn, $\cmatchOpth{t}{p} = \wait$ yields $t = t_1 t_2$, $p = p_1 p_2$, $\cmatchOpth{t_i}{p_i} \neq \fail$ for $i = 1,2$, and $a \neq \epsilon$. 
Assume $a = 1 a'$, implying $p' = p'_1 p_2$ and $p_1 \sstep{\ppcstep{a'}} p'_1$. 
In this case, $\cmatchOpth{t_1}{p_1}$ decided, then positive, would imply $p_1$ to be a normal form; while $\cmatchOpth{t_1}{p'_1} = \wait$ would contradict $\cmatchOpth{t}{p'}$ decided (recall that $\cmatchOpth{t_2}{p_2} \neq \fail$). Therefore the \ih\ can be applied to obtain $\pair{p_1}{a'}$ outermost, which suffices to conclude since $\pair{p}{\epsilon}$ is not a \tredex. 
The case $a = 2a'$ admits an analogous argument.
\end{proof}
}
%\end{proof}

\begin{lemma}[\axiomuse{\axEnclaveCreation}]
Let $\ppcstep{a}$, $\ppcstep{b}$ be \tredexes\ such that $\ppcstep{b} < \ppcstep{a}$, $\residu{\ppcstep{b}}{\ppcstep{a}}{\ppcstep{b'}}$, and $\residu{\emptyset}{\ppcstep{a}}{\ppcstep{c'}}$. Then $\ppcstep{b'} < \ppcstep{c'}$.
\end{lemma}

\begin{proof}
Observe that $a \not\leq b$ implying $b' = b$. Say $t \sstep{\ppcstep{a}} t'$, $\subtat{t}{a} = (\lpth p.s) u$, and $\subtat{t'}{c'} = (\lp{\theta'} p'.s') u'$.
We proceed by case analysis \wrt\ \thelem~\ref{rsl:ppc-creation}.
\begin{description}
\item [Case I]
In this case $c' \in \Pos{t}$ and $a = c' 1$, so that $\subtat{t}{c'} = (\lpth p.s) u u'$.
Therefore, it suffices to observe that $b = c'$ would contradict $\ppcstep{b}$ to be a \tredex, then $b < a$ implies $b < c'$.

\minitem [Cases II or III.(\ref{it:create-decided-both})]
In either case $c' = an$, thus $b < a$ implies $b < c'$.

\minitem [Case III.(\ref{it:create-decided-argument})]
In this case $a = c' 2 n$ and $\subtat{t}{c'} = (\lp{\theta'}p'.s') u''$.
Then, $b < a$ implies either $b < c'$, $b = c'$ or $b = c' 2 n'$ where $n' < n$. We conclude by observing that the second and third cases would contradict $\residu{\emptyset}{\ppcstep{a}}{\ppcstep{c'}}$ and \thelem~\ref{rsl:ppc-turns-decided-then-outermost}:(\ref{it:ppc-argument-match-turns-decided-then-outermost}) respectively.

\minitem [Case III.(\ref{it:create-decided-pattern})]
In this case $a = c' 11 n$ and $\subtat{t}{c'} = (\lp{\theta'}p''.s') u'$.
A similar analysis applies, resorting to \thelem~\ref{rsl:ppc-turns-decided-then-outermost}:(\ref{it:ppc-pattern-match-turns-decided-then-outermost}) instead of \thelem~\ref{rsl:ppc-turns-decided-then-outermost}:(\ref{it:ppc-argument-match-turns-decided-then-outermost}).
\end{description}
\vspace{-7mm}
\end{proof}

%%%%%%%%%%%%%%%%%%%%%%%%%%%%%%%%%%%%%%%%%%%%%%%%
%%%%%      Embedding and gripping axioms
\bigskip\noindent\textbf{The other embedding and gripping axioms.} \\
\axiomuse{\axLinearity} is immediate from the definition of residuals.
The remaining embedding axioms, and also \axiomuse{\axFda}, are related with the invariance of embedding \wrt\ residuals. 
The following result characterises those situations in which the embedding relation between two
  steps fails to be preserved \wrt\ the contraction of a third one.

\begin{lemma}
\label{rsl:ppc-embedding-invariance-exceptions}
Suppose $\residu{\ppcstepb}{\ppcstepa}{\ppcstepb'}$ and $\residu{\ppcstepc}{\ppcstepa}{\ppcstepc'}$, such that $\neg (\ppcstepb < \ppcstepc \,\Leftrightarrow\, \ppcstepb' < \ppcstepc')$.
Then:
\begin{itemize}
\minitem  $(\ppcstepa < \ppcstepb) \land (\ppcstepa < \ppcstepc)$, and moreover;
\minitem  either 
$(\ppcstepb < \ppcstepc) \land (\ppcstepb' \disj \ppcstepc')$, or 
$(\ppcstepa \grip \ppcstepb) \land (\ppcstepb \disj \ppcstepc) \land (\ppcstepb' < \ppcstepc') \land (a2 \leq c)$.
\end{itemize}
%
%\begin{itemize}
%\minitem
%either $(\stepa_r \grip \stepb_r) \land (\stepb_r \disj \stepc_r) \land (\stepb'_r < \stepc'_r) \land (a2 \leq c)$,
%\minitem
%or $(\stepb_r < \stepc_r) \land (\stepb'_r \disj \stepc'_r)$.
%% \land (\stepb''_r < \stepc'_r)$ for some $\stepb''_r$ verifying $\residu{\stepb_r}{\stepa_r}{\stepb''_r}$.
%\end{itemize}
\end{lemma}

\onlyPaper{
\begin{proof}
By case analysis of $\stepa$, $\stepb$ and $\stepc$. 
The most challenging case is when $\stepa < \stepb$ and $\stepa < \stepc$. In this case, a detailed analysis of the possible cases \wrt\ the residual relation is required. \CompleteInReport .
\end{proof}
}

\onlyReport{
\begin{proof}
By case analysis of $\stepa$, $\stepb$ and $\stepc$. Say $t \sstep{\stepa} v$ and $\subtat{t}{a} = (\lpth p.s) u$.
\begin{itemize}
\minitem 
$\stepa = \stepb$ or $\stepa = \stepc$: 
either case would contradict the existence of $\ppcstepb'$ and $\ppcstepc'$.
\minitem 
$\stepa \not\leq \stepb$ and $\stepa \not\leq \stepc$:
in this case $\stepb' = \stepb$ and $\stepc' = \stepc$, thus we conclude.
\minitem 
$\stepa \disj \stepb$ and $\stepa < \stepc$: 
implies $\stepb \disj \stepc$ and $\stepb' = \stepb \disj \stepa \leq \stepc'$, thus we conclude.
\minitem 
$\stepa < \stepb$ and $\stepa \disj \stepc$: analogous to the previous case.
\minitem
$\stepb < \stepa < \stepc$: 
implies $\stepb < \stepc$ and $\stepb' = \stepb < \stepa \leq \stepc'$, thus we conclude.
\minitem
$\stepc < \stepa < \stepb$: we obtain analogously $\stepc < \stepb$ and $\stepc' < \stepb'$, which suffices to conclude.
\minitem
$\stepa < \stepb$ and $\stepa < \stepc$: this is the interesting case. We analyse the possible cases \wrt\ the residual relation, recalling that all cases suppose $\matchOpth{u}{p} \neq \fail$, and therefore that $p$ is linear. 
\begin{itemize}
\item 
$b = a12n$ and $c = a12n'$. In this case $b' = an$ and $c' = an'$, thus we conclude immediately.
\item
$b = a12n$ and $c = a2m'n'$. 
In this case $b' = an$ and $c' = ak'n'$, where $\subtat{p}{m'} = \wid{x}$ and $\subtat{s}{k'} = x$ for some $x \in \theta$. 
Observe $b \disj c$. If $b' \not< c'$ then we conclude immediately, so that assume $b' < c'$, implying $n < k'n'$.
In turn, $\subtat{s}{n}$ and $\subtat{s}{k'}$ being a redex and a variable resp. imply $n < k'$. 
%Notice that $\subtat{s}{n}$ is a redex while $\subtat{s}{k'}$ is a variable, then $n < k'n'$ implies $n < k'$. 
Therefore $x \in \theta \cap \fv{\subtat{s}{n}}$, implying $\ppcstepa \grip \ppcstepb$. Thus we conclude.
\item
$b = a2mn$ and $c = a12n'$.
In this case, $b' = akn$ and $c' = an'$, where $\subtat{p}{m} = \wid{x}$ and $\subtat{s}{k} = x$ for some $x \in \theta$.
Observe $b \disj c$. Moreover $\subtat{s}{n'}$ being a redex while $\subtat{s}{k}$ is a variable implies $k \not < n'$, then $kn \not < n'$, hence $b' \not < c'$. Thus we conclude.
\item
$b = a2mn$ and $c = a2m'n'$.
In this case $b' = akn$ and $c' = ak'n'$, where $\subtat{p}{m} = \wid{x}$, $\subtat{s}{k} = x$, $\subtat{p}{m'} = \wid{y}$ and $\subtat{s}{k'} = y$ for some $x, y \in \theta$. Both $\subtat{s}{k}$ and $\subtat{s}{k'}$ being variable occurrences implies $k = k'$ or $k \disj k'$. An analogous argument yields $m = m'$ or $m \disj m'$.

Assume $b \not < c$ and $c \not < b$; \ie, $b \disj c \,$ or $\, b = c$. 
If $k \disj k'$ then we get immediately $b' \disj c'$.
Otherwise we have $k = k'$, implying $x = y$ and therefore $m = m'$ by linearity of $p$. 
In this case, $n = n'$ yields $b' = c'$, and otherwise, it must be $b \disj c$ implying $n \disj n'$, and then $b' \disj c'$.
In any of these cases, we conclude immediately.

If $b < c$, then $m = m'$ implying $x = y$, and $n < n'$. If $k = k'$, then $b' < c'$, otherwise, $b' \disj c'$.
Thus we conclude.
	
Finally, if $b' < c'$, then $k = k'$ and $n < n'$. But $k = k'$ implies $x = y$, and then $m = m'$ by linearity of $p$. Then $b < c$.
\end{itemize}

\end{itemize}
\vspace{-7mm}
\end{proof}
}

It is easy to obtain \axiomuse{\axCtxFreeness}, \axiomuse{\axEnclaveEmbedding} and \axiomuse{\axFda} as corollaries of \thelem~\ref{rsl:ppc-embedding-invariance-exceptions}. 

%\begin{lemma}
%\label{rsl:redex-taken-by-variable-in-positive-match}
%Let $p,t,b$ such that $\matchOpth{t}{p}$ is positive and $\subtat{t}{b}$ is a redex.
%Then there exists some $a \leq b$ verifying $\subtat{p}{a} = \wid{x}$ for some $x \in \theta$.
%\end{lemma}
%
%\begin{proof}
%By induction on $p$, considering the cases in the definition of the compound matching allowing $\matchOpth{t}{p}$ to be positive.
%If $\isbm{p}{\theta}$ then taking $a = \epsilon$ allows to conclude.
%Otherwise, $\matchOpth{t}{p}$ positive implies $p = p_1 p_2$, $t = t_1 t_2 \in \MForms$, and $\matchOpth{t_i}{p_i}$ positive for $i = 1,2$.
%In turn, $t \in \MForms$ and $\subtat{t}{b}$ being a redex imply $b \neq \epsilon$, then $b = i b'$ where $i \in \set{1,2}$, hence $\subtat{t}{b} = \subtat{t_i}{b'}$. \Ih\ yields $\subtat{p_i}{a'} = \wid{x}$ where $x \in \theta$, for some $a' \leq b'$. We conclude by taking $a = i a'$.
%\end{proof}

%\begin{lemma}
%\label{rsl:middle-redex-has-residual}
%Let $\stepa_r < \stepb_r < \stepc_r$ and $\residu{\stepc_r}{\stepa_r}{\stepc'_r}$ for some $\stepc'_r$. Then there exists some $\stepb'_r$ verifying $\residu{\stepb_r}{\stepa_r}{\stepb'_r}$.
%\end{lemma}

\begin{lemma}[\axiomuse{\axMisterious}]
Let $\ppcstepa, \ppcstepb, \ppcstepc, \ppcstepc'$ \tredexes\ verifying $\ppcstepa < \ppcstepc$, $\ppcstepb < \ppcstepc$, $\ppcstepb \not\leq \ppcstepa$, and $\residu{\ppcstepc}{\ppcstepa}{\ppcstepc'}$.
Then there exists a \tredex\ $\ppcstepb'$ such that $\residu{\ppcstepb}{\ppcstepa}{\ppcstepb'}$ and $\ppcstepb' < \ppcstepc'$.
\end{lemma}

\begin{proof}
Observe that $a < c$, $b < c$ and $b \not \leq a$ implies $a < b < c$. 
\onlyPaper{
A case analysis \wrt\ the definition of residuals, considering $\ppcstepa < \ppcstepc$, allows to conclude. \CompleteInReport .
}
\onlyReport{
We proceed by case analysis on the definition of residuals, considering $\ppcstepa < \ppcstepc$. Say $\subtat{t}{a} = (\lpth p.s) u$. Observe that $a < c$ and $\residu{\ppcstepc}{\ppcstepa}{\ppcstepc'}$ imply that $\matchOpth{u}{p}$ is positive.
\begin{itemize}
\item If $c = a 12 n'$, so that $c' = an'$, then $b < c$ implies $b = a12n$ and $n < n'$ (recall $\subtat{t}{a1} \,\in \Abstract$). Hence, taking $b' = an$ suffices to conclude.

\item If $c = a2mn$, then $b = a 2 b''$ and $b'' < mn$. Observe that $\subtat{p}{m} = \wid{x}$ where $x \in \theta$, and $c' = akn$ where $\subtat{s}{k} = x$. 
Noticing that $\matchOpth{u}{p}$ is positive and $\subtat{u}{b''}$ is a redex, a simple induction on $p$ yields $b'' = b_1 b_2$ where $\subtat{p}{b_1} = \wid{y}$.
In turn, $b_1 b_2 < m n$, along with both $\subtat{p}{b_1}$ and $\subtat{p}{m}$ being matchable
occurrences, imply that $b_1 = m$, then $x = y$, and also $b_2 < n$. Hence we conclude by taking $b'
= a k b_2$. 
\end{itemize}
\vspace{-6mm}
}
\end{proof}

%\medskip
%Finally, we verify the two remaining gripping axioms.

\begin{lemma}
\label{rsl:fv-residuals}
Suppose $t = (\lpth p.s) u \sstep{\ppcstepa} t'$, $\residu{\ppcstepc}{\ppcstepa}{\ppcstepc'}$, and $x \in \fv{\subtat{t'}{c'}}$.
Then $x \in \fv{\subtat{t}{c}}$, or $\ppcstepa \grip \ppcstepc$ and $x \in \fv{u}$.
\end{lemma}

\onlyPaper{
\begin{proof}
By case analysis on the definition of residuals. \CompleteInReport.
\end{proof}
}

\onlyReport{
\begin{proof}
If $a \not\leq c$ or $c = a2mn$, then $\subtat{t}{c} = \subtat{t'}{c'}$, implying $x \in \subtat{t}{c}$.
Otherwise, \ie\ if $c = a12n$, $c' = an$, and $\matchOpth{u}{p} \neq \fail$, let us consider $d$ such that $\subtat{t'}{c'd} = \subtat{(\matchOpth{u}{p} s)}{nd} = x$. Given $n \in \Pos{s}$, it is easy to obtain 
$\subtat{(\matchOpth{u}{p} s)}{nd} = \subtat{(\matchOpth{u}{p} (\subtat{s}{n}))}{d} = \subtat{(\matchOpth{u}{p} (\subtat{t}{c}))}{d}$. 
In turn, $x \in \fv{\matchOpth{u}{p} (\subtat{t}{c})}$ yields easily $x \in \fv{\subtat{t}{c}}$ or $\, x \in \fv{u} \, \land \, \subtat{t}{c} \cap \  \theta \neq \emptyset$. We conclude by observing that the latter case implies $\ppcstepa \grip \ppcstepc$.
\end{proof}
}

\begin{lemma}[\axiomuse{\axFdb}]
Consider \tredexes\ $\ppcstepa, \ppcstepb, \ppcstepb', \ppcstepc, \ppcstepc'$ verifying $\residu{\ppcstepb}{\ppcstepa}{\ppcstepb'}$, $\residu{\ppcstepc}{\ppcstepa}{\ppcstepc'}$, and $\ppcstepb' \grip \ppcstepc'$.
Then $\ppcstepb \grip \ppcstepc \,\lor\, \ppcstepb \grip \ppcstepa \grip \ppcstepc$.
\end{lemma}

\begin{proof}
Let $t \sstep{\ppcstepa} t'$, and say $\subtat{t}{a} = (\lpth p.s) u$, $\subtat{t'}{b'} = (\lp{\theta'} p'.s') u'$, and $\subtat{t}{b} = (\lp{\theta'} p''.s'') u''$; notice that the set $\theta'$ is invariant \wrt\ the contraction of $\ppcstepa$. 
Recall that $\ppcstepb' \grip \ppcstepc'$ implies $\matchOpt{\theta'}{u''}{p''}$ positive, $b'12 \leq c'$ and $\theta' \cap \fv{\subtat{t'}{c'}} \neq \emptyset$. 
Observe that $\matchOpt{\theta'}{u''}{p''}$ positive and $\matchOpt{\theta'}{u'}{p'}$ decided imply $\matchOpt{\theta'}{u'}{p'}$ positive; \confer\ \thelem~\ref{rsl:reduction-mantains-decided-match}. 
Let $x \in \theta' \cap \fv{\subtat{t'}{c'}}$. 
Then \thelem~\ref{rsl:fv-residuals} implies $x \in \fv{\subtat{t}{c}} \lor (\ppcstepa \grip \ppcstepc \,\land\, x \in \fv{u})$.

Given $b' < c'$, \thelem~\ref{rsl:ppc-embedding-invariance-exceptions} implies $b < c$ or $(b \disj c \land a2 \leq c)$.
The latter case implies $\ppcstepa \notgrip \ppcstepc$ (since $a2 \leq c$) and  $\theta'  \cap \fv{\subtat{t}{c}} = \emptyset$ (since $b\disj c$ and $\subtat{t}{b} = (\lp{\theta'} p''.s'') u''$), contradicting $x \in \fv{\subtat{t}{c}} \lor \ppcstepa \grip \ppcstepc$.
Hence $b < c$. %
\onlyPaper{%
An analysis of the three possible cases \wrt\ $\stepa$, \ie\ $\stepa < \stepb < \stepc$, $\stepb < \stepa < \stepc$ and $\stepb < \stepc < \stepa$, allows to conclude. \CompleteInReport.
}%
\onlyReport{%
There are three cases to analyse, depending on $a$.

\begin{enumerate}
\minitem $a < b < c$. \\
Assume $b = a12n$, $c = a12n'$  and $n < n'$, so that $b' = an$ and $c' = an'$. 
Then $b'12 \leq c'$ implies $n12 \leq n'$, and therefore $b12 \leq c$. Moreover, $a12 \leq b$ implies $\theta' \cap \fv{u} = \emptyset$, so that $x \in \fv{\subtat{t}{c}}$. Consequently, $\ppcstepb \grip \ppcstepc$.

Assume $b = a2mn$, $c = a2m'n'$, $mn < m'n'$, $\subtat{p}{m} = \wid{y}$, $\subtat{p}{m'} = \wid{z}$, and $y,z \in \theta$. 
In this case, both $\subtat{p}{m}$ and $\subtat{p}{m'}$ being variable occurrences, along with $mn < m'n'$, imply $m = m'$, then $y = z$. 
Therefore $b' = akn$ and $c' = ak'n'$, where $\subtat{s}{k} = \subtat{s}{k'} = y$. 
In turn, the last assertion along $b'12 \leq c'$ imply $k = k'$, then $n12 \leq n'$, therefore $b12 \leq c$. Moreover, in this case $\ppcstepa \notgrip \ppcstepc$ implying $x \in \fv{\subtat{t}{c}}$. Thus $\ppcstepb \grip \ppcstepc$.

\minitem $b < a < c$. \\
We have $b12 = b'12 \leq c'$ and $a \leq c'$, then $b < a$ implies $b12 \leq a < c$. The existence of $\ppcstepc'$ yields $\matchOpth{u}{p} \neq \fail$. 
If $x \in \fv{\subtat{t}{c}}$, then $\ppcstepb \grip \ppcstepc$; otherwise, $\ppcstepa \grip \ppcstepc$ and $x \in \fv{u} \subseteq \fv{\subtat{t}{a}}$ imply $\ppcstepb \grip \ppcstepa$. Thus we conclude.

\minitem $b < c < a$. \\
We have $b12 = b'12 < c' = c$, and $\ppcstepa \notgrip \ppcstepc$ implies $x \in \fv{\subtat{t}{c}}$. Therefore $\ppcstepb \grip \ppcstepc$.
\end{enumerate}
\vspace{-4mm}
}
\end{proof}

\begin{lemma}[\axiomuse{\axFdc}]
Let $\ppcstepa, \ppcstepb, \ppcstepc \in \ROccur{t}$ such that $\ppcstepa \grip \ppcstepb$ and $\ppcstepc < \ppcstepb$. 
Then $\ppcstepa \grip \ppcstepc \,\lor\, \ppcstepc \leq \ppcstepa$.
\end{lemma}

\begin{proof}
Observe that $a < b$ and $c < b$ implies that either $c \leq a$ or $a < c$. In the former case we  immediately conclude. Otherwise, it suffices to notice that $a < c < b$, $a12 \leq b$ and $\subtat{t}{c}$ being a redex imply $a12 \leq c$, and that $c < b$, along with the variable convention, implies $\theta \cap \fv{\subtat{t}{b}} \subseteq \theta \cap \fv{\subtat{t}{c}}$, where $\subtat{t}{a} = (\lpth p.s) u$. Therefore $\emptyset\neq \theta \cap \fv{\subtat{t}{c}}$ so that we conclude $\ppcstepa \grip \ppcstepc$.
\end{proof}

%%%%%%%%%%%%%%%%%%%%%%%%%%%%%%%%%%%%%%%%%%%%%%%%
%%%%%      FD and SO
\bigskip\noindent\textbf{The axioms \axiomuse{\axFD} and \axiomuse{\axSO}.} \\
%We prove the two remaining axioms.
\axiomuse{\axFD} is a consequence of the gripping axioms. Thm. 3.2. in~\cite{thesis-mellies} states that an ARS satisfying the gripping axioms along with \axiomuse{\axSelfReduction}, \axiomuse{\axFiniteResiduals} and \axiomuse{\axLinearity}, and whose embedding and gripping relations are acyclic, also enjoys \axiomuse{\axFD}. 
For the ARS modeling \theppc, we have verified all the required axioms. The embedding relation being an order, and the gripping relation being included in the former, imply immediately that both are acyclic. Hence we obtain \axiomuse{\axFD}.

For the axiom \axiomuse{\axSO}, the interesting case is when the \tredexes\ are nested, \ie\ $\ppcstepa < \ppcstepb$. Let $\subtat{t}{a} = (\lpth p.s)u$, $t \sstep{\ppcstepb} t'$, and $\subtat{t'}{a} = (\lpth p'.s')u'$.
If $\matchOpth{u}{p} = \fail$ is a matching failure, it suffices to observe that $\matchOpth{u'}{p'} = \fail$, \confer\ \thelem~\ref{rsl:reduction-mantains-decided-match}.
If $\matchOpth{u}{p}$ is positive, then 
%a simple, yet tedious, analysis resorting to the fact that substitutions are isomorphisms \wrt\ replacement, residuals and reduction steps, suffices to conclude.
a simple, yet extensive, analysis resorting to various properties related to substitutions (\eg\ that reduction steps, their targets, and residuals, are preserved by substitutions), suffices to conclude.

\subsubsection{A reduction strategy for \theppc}
\label{sec:ppc-strategy}

This section introduces a normalising strategy \strs\ for \theppc.  A \defi{\protoredex} is a term
of the form $(\lpth\ p.t)u$, regardless of whether the match $\matchOpth{u}{p}$ is decided or not.
The rationale behind the definition of \strs\ can be described through two observations.  First, it
focuses on the leftmost-outermost (LO) prestep of $t$, entailing that when \theppc\ is restricted to
the $\l$-calculus it behaves exactly as the LO strategy for the $\l$-calculus.  Second, if the match
corresponding to the LO occurrence of a prestep is not decided, then the strategy selects only the
(outermost) step, or steps, in that subterm which should be contracted to get it ``closer'' to a
decided match. \Eg\ in the term $(\lp{\set{x,y}}\ \ca\ \wid{x}\ (\cc\ \wid{y}).y\ x) \ (\ca\ r_1 \
r_2)$, where all the $r_i$'s are steps,  the match $\matchOpt{\set{x,y}}{\ca\ r_1\ r_2}{\ca\ \wid{x}\ (\cc\ \wid{y})}$ is not decided
and the r\^ole played by $r_1$ is different from that of $r_2$ in obtaining a decided match.
Replacing $r_1$ by an arbitrary term $t_1$ does not yield a decided match, \ie\
$\matchOpt{\set{x,y}}{\ca\ t_1\ r_2}{\ca\ \wid{x}\ (\cc\ \wid{y})}$ is not decided.  However,
replacing $r_2$ by $\cc\ s_2$ (resp. by $\cd\ s_2$) does: $\matchOpt{\set{x,y}}{\ca\ r_1\ (\cc\
  s_2)}{\ca\ \wid{x}\ (\cc\ \wid{y})} = \set{x \rewto r_1, y \rewto s_2}$ (resp.
$\matchOpt{\set{x,y}}{\ca\ r_1\ (\cd\ s_2)}{\ca\ \wid{x}\ (\cc\ \wid{y})} = \fail$). Hence,
contraction of $r_2$ can contribute towards obtaining a decided match, while contraction of $r_1$
does not. 
A different example, in which multiple steps (including one in the pattern) are selected, is $(\lp{\set{x,y}}\ \ca\ (\cb\ \wid{x})\ r_1.r_2) \ (\ca\ r_3\ (\cd\ r_4))$, where all the $r_i$'s are steps. The strategy selects $r_1$ and $r_3$.  Moreover, notice that contraction of $r_4$ is delayed since the match operation is not decided when the pattern is a redex
(if the contractum of $r_1$ were \eg\ either $\cd\ \wid{y}$ or
$\ca$, then the match w.r.t. $\cd\ r_4$ would be decided without the need of reducing $r_4$).
We note that in both examples, the decision made by the strategy (namely, to select $r_2$ in the first case and $\set{r_1,r_3}$ in the second one) coincides for any term having the indicated form.
This decision is based solely on the structure of the term, in order to avoid the need for history or lookahead.

\delia{Formally, we define the \defi{reduction strategy \strs} as a function
  from terms to sets of steps by means of an auxiliary function
  $\strspos$. This auxiliary function gives the \emph{positions}
  of the steps to be selected:  
$\strs(t) \eqdef \set{\pair{t}{p} \ \sthat p \in \strspos(t)}$.

In turn, the definition of $\strspos$ resorts to an additional auxiliary function, called \strsp, that formulates the simultaneous structural analysis of the argument and pattern of a \protoredex. The arguments of \strsp\ are the pattern and the argument of a \protoredex. Its outcome is a pair of sets of positions, corresponding to steps inside the pattern
and argument respectively, which could contribute to turning a non-decided match into a decided one.

The formal definition of $\strspos$ and \strsp\ follows. 
} %
Recall that we write $\isbm{t}{\theta}$ when $t = \wid{x}$ for some $x \in \theta$.

\medskip\noindent
{\small
$\begin{array}{@{}r@{\hspace{1pt}}c@{\hspace{1pt}}l@{\hspace{4mm}}l}
	\strspos(x)   & := & \emptyset \\
	\strspos(\wid{x}) & := & \emptyset \\
	\strspos(\lpth\ p.t) & := & 1\strspos(p) & \textif p \notin \NForms \\
	\strspos(\lpth\ p.t) & := & 2\strspos(t) & \textif p \in \NForms \\
	
	\strspos((\lpth\ p.t) u) & := & \{ \epsilon \} 
			& \textif
			\matchOpth{u}{p} \textnormal{ decided} \\

	\strspos((\lpth\ p.t) u) & := & 11G \cup 2D
			& \textif \matchOpth{u}{p} = \wait,
                         \fnsmth{u}{p} = \pair{G}{D}\neq \emptypair,  \\

	\strspos((\lpth\ p.t) u) & := & 11\strspos(p) 
			& \textif \matchOpth{u}{p} = \wait,
                          \fnsmth{u}{p} = \emptypair, 
                          p \notin \NForms \\

	\strspos((\lpth\ p.t) u) & := & 12\strspos(t) 
			& \textif \matchOpth{u}{p} = \wait,
                         \fnsmth{u}{p} = \emptypair, 
                          p \in \NForms, t \notin \NForms \\

	\strspos((\lpth\ p.t) u) & := & 2\strspos(u)
			& \textif \matchOpth{u}{p} = \wait,
                          \fnsmth{u}{p} = \emptypair, 
                          p \in \NForms, t \in \NForms \\

	\strspos(t u) & := & 1\strspos(t) & \textif
			t \mbox{ is not an abstraction and } t \notin \NForms   \\

	\strspos(t u) & := & 2\strspos(u) & \textif
	        t \mbox{ is not an abstraction and } t \in \NForms   \\ \\
\end{array}$

\noindent
$\begin{array}{@{}r@{\hspace{1pt}}c@{\hspace{1pt}}l@{\hspace{8pt}}l}
	\fnsmth{t}{\wid{x}} & := & \emptypair & \textif x \in \theta \\

	\fnsmth{\wid{x}}{\wid{x}} & := & \emptypair & \textif x \notin \theta \\

	\fnsmth{t_1 t_2}{p_1 p_2} & := & 
			\langle 1G_1 \cup 2G_2,1D_1 \cup 2D_2 \rangle
			& \textif t_1 t_2, p_1 p_2 \in \MForms, 
                      \fnsmth{t_i}{p_i} = \pair{G_i}{D_i} \\
			% & \& \ \fnsmth{t_i}{p_i} = \pair{G_i}{D_i}

	\fnsmth{t}{p} & := & \pair{\strspos(p)}{\emptyset}
			& \textif
			p \notin \MForms \\

      \fnsmth{t}{p} & := & \pair{\emptyset}{\strspos(t)}
			& \textif
			p \in \MForms\ \&\ t \notin \MForms\ \&\  
                        \neg \isbm{p}{\theta}
\end{array}$
} % small 

Notice the similarities between the first three clauses in the definition of \strsp\ and those of
the definition of the matching operation (\confer\ \thesec~\ref{sec:ppc-basic-elements}).
Also notice that whenever a non-decided match can be turned into a decided one, the function \strsp\ chooses at least one (contributing) step.
Formally, it can be proved that, given $p$ and $u$ such that $\matchOpth{u}{p} = \wait$, if $p'$ and $u'$ exist such that $p \mred{} p'$, $u \mred{} u'$ and $\matchOpth{u'}{p'}$ is decided, then $\fnsmth{u}{p} \neq \emptypair$. 

%If the LO \protoredex\ of a term is in fact a step, then the strategy selects exactly that step (fifth clause); if it is not a step, and the function \strsp\ returns some steps which could contribute towards a decided match, then the strategy selects them (sixth clause).
%
%Otherwise, the \protoredex\ will never turn into a step. Indeed, it can be proved that, given $p$ and $u$ such that $\matchOpth{u}{p} = \wait$, if there exist $p'$ and $u'$ such that $p \mred{} p'$, $u \mred{} u'$ and $\matchOpth{u'}{p'}$ is decided, then $\fnsmth{u}{p} \neq \emptypair$. In this case the strategy looks for the LO \protoredex\ inside the components of the term (seventh, eighth and ninth clauses).
%
%The remaining clauses in the definition of \strs\ formalise the focus on the LO \protoredex\ for other forms of the term.

\medskip
Let us analyse briefly the clauses in the definition of $\strspos$.
The focus on the LO \protoredex\ of a term is formalised in the first four and the last two clauses.
If the LO \protoredex\ is in fact a step, then the strategy selects exactly that step; this is the meaning of the fifth clause. 
If the LO \protoredex\ is not a step, then \strsp\ is used.
If it returns some steps which could contribute towards a decided match, then the strategy selects them (sixth clause).
Otherwise, as we already remarked, the \protoredex\ will never turn into a step, so that the strategy looks for the LO \protoredex\ inside the components of the term (seventh, eighth and ninth clauses).

While the strategy focuses on the obtention of a decided match for the LO prestep, it can select more \tredexes\ than needed for that aim. 
\Eg, for the term 
$(\lp{\set{y}}\ \ca\ \cb\ \cc\ \wid{y}.y)\ (\ca\ (\Id\ \cc)\ (\Id\ \cb)\ (\Id\ \ca))$,
the set selected by the strategy $\strs$ is $\set{\Id\ \cc, \Id\ \cb}$, even if  the contraction of just one step of the set suffices to make the head match decided.

Notice that \strs\
  collapses to the LO-strategy
when considering the subset of \theppc\ terms given by the terms of the \lc.

\medskip
The reduction strategy \strs\ is complete, \ie, if $t$ is not a normal form, then $\strs(t) \neq \emptyset$. Moreover, all steps in $\strs(t)$ are \emph{outermost}.
On the other hand, notice that \strs\ is not \textit{outermost
  fair}~\cite{vanRaamsdonk:1997}. Indeed, given $(\l\, \cc\,x.s)\,
\Omega$, where $\Omega$ is a non-terminating term, \strs\ continuously
contracts $\Omega$, even when $s$ contains a step.

\delia{ Additionally, the steps in $\strs(t)$ are not always
  hereditarily outermost, \ie, universally $<$-external in the sense
  of~\cite{DBLP:conf/popl/AccattoliBKL14} (\confer\ Sec.~\ref{sec:free-domin}).  
  Thus for example,  given the term $t = (\l_{x} \ca \cb
  \wid{x}.t_1) ((\Id \cd) (\Id \cb) t_2)$,  the strategy $\strs$ selects
  the set of redexes $\set{\Id \cd, \Id \cb}$. 
  By contracting only $\Id \cd$,
  we get $t \sstep{} t' = (\l_{x} \ca \cb \wid{x}.t_1) (\cd (\Id \cb)
  t_2)$, where $t'$ contains a (created) redex that embeds (the
  residual of the original) $\Id \cb$.  Note that the created,
  embedding redex is a \emph{matching failure}. Such is always the case
  whenever a redex embeds a residual of $\strs(t)$, observation
  which is used to prove that $\strs(t)$ is \ngrip.  }

%%% Local Variables: 
%%% mode: latex
%%% TeX-master: "paper"
%%% End: 

\subsubsection{Properties of the reduction strategy \strs}
\label{sec:ppc-strategy-properties}

In this section we prove that \strs\ computes necessary (Prop.~\ref{rsl:ppc-strategy-necessary}) and non-gripping (Prop.~\ref{rsl:strs-non-gripping}) sets. 
These proofs rely on the notion of \emph{projection} of a \mredseq\ \emph{w.r.t. a position}. We
describe briefly this notion in the following. \onlyPaper{\CompleteInReport}.

Let $a$ be a position. Given $\ppcstep{b} = \pair{t}{b}$, we say that $a \leq \ppcstep{b}$ iff $a \leq b$. 
This definition is extended to multisteps and \redseqs: $a \leq \ppcsetred{B}$ iff $a \leq \ppcstep{b}$ for all $\ppcstep{b} \in \ppcsetred{B}$, $a \leq \reda$ is defined similarly.

If $a \leq \ppcstep{b} = \pair{t}{b}$, implying $b = a b'$, then we define the \defi{projection} of $\ppcstep{b}$ \wrt\ $a$, as follows: 
$\subtat{\ppcstep{b}}{a} = \pair{\subtat{t}{a}}{b'}$. 
If $a \leq \ppcsetred{B}$, then the projection $\subtat{\ppcsetred{B}}{a}$ is defined as expected. We define similarly $\subtat{\reda}{a}$ if $a \leq \reda$.
The targets of steps and \redseqs, the residual relation, and the developments of a multistep, are compatible with these projections. 

\medskip
A multistep $\ppcsetred{B}$ \defi{preserves} $a$ iff all $\ppcstep{b} \in
\ppcsetred{B}$ verify $b \not < a$ (or equivalently $a \leq b$ or $a \disj
b$).  If $\ppcsetred{B}$ preserves $a$, then this set can be
partitioned\footnote{\delia{The relation \emph{preserves}
      is similar to  \emph{free-from} 
      ({\it c.f.}\ Sec.~\ref{sec:free-domin}). 
      Morover, the
      partition given by $\ppcsetred{B} = \thefreewrt{\ppcsetred{B}}{a} \uplus
      \thedominwrt{\ppcsetred{B}}{a}$ bears some similarity to that
      described after the definition of free-from, albeit the former
      is restricted to \tredexsets\ that preserves some position,
      while the latter applies to any \tredexset.
Additionally, free-from is a
relation on abstract steps, \tredexsets\ and \mredseq, while preserves
is a relation between \theppc\ \tredexsets\ and \emph{positions} (not necessarily redexes). 
}} into two parts, say $\thefreewrt{\ppcsetred{B}}{a}$ and
$\thedominwrt{\ppcsetred{B}}{a}$, such that $b \disj a$ if $\ppcstep{b} \in
\thefreewrt{\ppcsetred{B}}{a}$, and $a \leq b$ if $\ppcstep{b} \in
\thedominwrt{\ppcsetred{B}}{a}$. Observe that $\ppcsetred{B} =
\thefreewrt{\ppcsetred{B}}{a} \uplus \thedominwrt{\ppcsetred{B}}{a}$.%

If $t \mstep{\ppcsetred{B}} t'$ preserves $a$, then $\subtat{t'}{a}$ is determined by $\thedominwrt{\ppcsetred{B}}{a}$, \ie\ $\subtat{t'}{a} = \subtat{t''}{a}$ where $t \mred{\thedominwrt{\ppcsetred{B}}{a}} t''$.
Therefore we can extend the definition of the \defi{projection} $\subtat{\ppcsetred{B}}{a}$ to any $\ppcsetred{B}$ preserving $a$: $\thefreewrt{\ppcsetred{B}}{a}$ is simply ignored.

In turn, a \emph{\mredseq} $\mreda$ \defi{preserves} $a$ iff all its elements do. 
Suppose $\mreda$ preserves $a$, and let $t \mstep{\mredel{\mreda}{1}} t_1 \mstep{\mredel{\mreda}{2}} t_2 \mred{\mredsub{\mreda}{3}{}} t'$. 
Observe that $\subtat{t}{a} \mred{\subtatwide{\mredel{\mreda}{1}}{a}} \subtat{t_1}{a} \mred{\subtatwide{\mredel{\mreda}{2}}{a}} \subtat{t_2}{a} \ldots$. 
This observation leads to define $\subtat{\mreda}{a}$, the \defi{projection} of $\mreda$ \wrt $a$, as expected.

Some notions related to \mredseqs\ are compatible with projections:
\begin{lemma}
\label{rsl:mred-projection-good-behavior}
Let $t \mred{\mreda} t'$ and assume $\mreda$ preserves $a$.
Then: 
\begin{enumerate}[(i)]
\minitem \label{it:mred-projection-good-tgt}
$\subtat{t}{a} \mred{\subtatwide{\mreda}{a}} \ \subtat{t'}{a}$.
\minitem \label{it:mred-projection-good-residuals}
If $\ppcstep{ac} \in \ROccur{t}$, then $\residu{\ppcstep{ac}}{\mreda}{\ppcstep{d}}$ iff $d = a d_1$ and $\residu{\ppcstep{c}}{\subtat{\mreda}{a}}{\ppcstep{d_1}}$.
\minitem \label{it:mred-projection-good-uses}
If $\ppcstep{ac} \in \ROccur{t}$, then $\mreda$ uses $\ppcstep{ac}$ iff $\subtat{\mreda}{a}$ uses $\ppcstep{c}$.
\end{enumerate}
\end{lemma}

\onlyPaper{
\begin{proof}
This result admits a simple, though long, proof. \CompleteInReport.
\end{proof}
}

\onlyReport{
\begin{proof}
See the Appendix.
\end{proof}
}

%%% Local Variables: 
%%% mode: latex
%%% TeX-master: "article"
%%% End: 

\medskip
%\carlos{We introduce some results} used to prove that \strs\ selects necesary sets, and also to prove that \strs\ selects \ngrip\ sets.
In the remainder of this section, we show that \strs\ always selects \emph{necessary} and \emph{\ngrip} sets of redexes, along with the needed auxiliary results.

\begin{lemma}
\label{rsl:strsp-empty-if-positive-match}
If $\cmatchOpth{u}{p}$ is positive, then  $\fnsmth{u}{p} = \pair{\emptyset}{\emptyset}$. 
\end{lemma}

\begin{proof}
Observe that $\cmatchOpth{u}{p}$ positive implies $p \in \DStructs$. 
Then a simple induction on $p$ suffices. Particularly, if $p = p_1 p_2$, then $\cmatchOpth{u}{p}$ positive implies $u = u_1 u_2$ and both $\cmatchOpth{u_i}{p_i}$ positive, so that the \ih\ on each $p_i$ allows to conclude.
\end{proof}

\begin{lemma}
\label{rsl:ppc-necessary-strsp} 
Let $t, u$ be terms and $p$ be a pattern. 
\hspace*{1mm} 
\begin{enumerate}[(i)]
\minitem 
\label{it:ppc-necessary-mf}
Let $t \mred{\mreda} t'$ where $t \notin \MForms$%
%ACHICAR
%and
, $t' \in \MForms$. Then $\mreda$ uses $\strs(t)$ and $\residus{\strs(t)}{\mreda} = \emptyset$.
\minitem 
\label{it:ppc-necessary-strsp-cmatch}
Let $p \mred{\mredb} p'$ and $u \mred{\mredc} u'$, where $\cmatchOpth{u}{p} = \wait$ and $\cmatchOpth{u'}{p'}$ is decided. Let $\pair{G}{D} = \fnsmth{u}{p}$.
Then $\mredb$ uses $\ppcsetred{G}$ or $\mredc$ uses $\ppcsetred{D}$. 
Moreover, $\cmatchOpth{u'}{p'}$ positive implies $\residus{\ppcsetred{G}}{\mredb} = \residus{\ppcsetred{D}}{\mredc} = \emptyset$.
\minitem 
\label{it:ppc-necessary-strsp-match}
Let $p \mred{\mredb} p'$ and $u \mred{\mredc} u'$, where $\matchOpth{u}{p} = \wait$ and $\matchOpth{u'}{p'}$ is decided. Let $\pair{G}{D} = \fnsmth{u}{p}$.
Then $\mredb$ uses $\ppcsetred{G}$ or $\mredc$ uses $\ppcsetred{D}$. 
Moreover, $\matchOpth{u'}{p'}$ positive implies $\residus{\ppcsetred{G}}{\mredb} = \residus{\ppcsetred{D}}{\mredc} = \emptyset$.
\end{enumerate}
\end{lemma}

\onlyPaper{
\begin{proof}
For items~(\ref{it:ppc-necessary-mf}) and (\ref{it:ppc-necessary-strsp-cmatch}), we proceed by simultaneous induction on $\tsize{t}+\tsize{u} +\tsize{p}$.
The proof analyses the different cases in the definitions of \strs\ and \strsp; it resorts to \thelem~\ref{rsl:mred-projection-good-behavior} and \thelem~\ref{rsl:strsp-empty-if-positive-match}. Item~(\ref{it:ppc-necessary-strsp-match}) follows easily from item~(\ref{it:ppc-necessary-strsp-cmatch}). \CompleteInReport.
\end{proof}
}

\onlyReport{
\begin{proof}
  Item~(\ref{it:ppc-necessary-strsp-match}) follows from item~(\ref{it:ppc-necessary-strsp-cmatch}) since $\matchOpth{u}{p} = \wait$ implies $\cmatchOpth{u}{p} = \wait$, and $\matchOpth{u'}{p'}$ decided or positive implies $\cmatchOpth{u'}{p'}$ decided and positive respectively. 
	We prove items~(\ref{it:ppc-necessary-mf}) and (\ref{it:ppc-necessary-strsp-cmatch}), by simultaneous induction on $\tsize{t}+\tsize{u} +\tsize{p}$.

\medskip
Item~(\ref{it:ppc-necessary-mf}).
Observe that $t \notin \MForms$ implies that $t$ is either a variable or an application. In the former case $t' = t \notin \MForms$\ contradicting the hypothesis. So we consider the latter one.

Assume $t = (\lpth p.s) u$ where $\matchOpth{u}{p}$ is decided, so that $\strs(t) = \set{\pair{t}{\epsilon}}$. 
If there is some $i \leq \rlength{\mreda}$ such that $\pair{t_{i}}{\epsilon} \in \mredel{\mreda}{i}$, where $t_{i} \mstep{\mredel{\mreda}{i}} t_{i+1}$, taking the minimal such $i$ yields $\residus{\strs(t)}{\mredunt{\mreda}{i-1}} = \set{\pair{t_i}{\epsilon}}$, so that $\mreda$ uses $\strs(t)$, and moreover $\residus{\strs(t)}{\mredunt{\mreda}{i}} = \epsilon$. Otherwise $t' = (\lpth p'.s') u'$, contradicting $t' \in \MForms$. Thus we conclude.

Assume $t = (\lpth p.s) u$ where $\matchOpth{u}{p} = \wait$. Then $t' \in
\MForms$ implies $t \mred{\mreda'} t'' \mred{\mreda''} t'$ where $t'' = (\lpth p''.s'') u''$ and $\matchOpth{u''}{p''}$ is decided.  
Moreover $\mreda'$ preserves 11 and 2, implying $p \mred{\subtatwide{\mreda'}{11}} p''$ and $u \mred{\subtatwide{\mreda'}{2}} u''$ by \thelem~\ref{rsl:mred-projection-good-behavior}:(\ref{it:mred-projection-good-tgt}).  
Let $\fnsmth{u}{p} = \pair{G}{D}$.  
The \ih:(\ref{it:ppc-necessary-strsp-match}) can be applied, yielding that $\subtat{\mreda'}{11}$ uses $\ppcsetred{G}$ or $\subtat{\mreda'}{2}$ uses $\ppcsetred{D}$. Therefore $\pair{G}{D} \neq \pair{\emptyset}{\emptyset}$, implying $\strspos(t) = 11G \cup 2D$. Furthermore, \thelem~\ref{rsl:mred-projection-good-behavior}:(\ref{it:mred-projection-good-uses}) implies that $\mreda'$ uses $\strs(t)$. 
On the other hand, if $\matchOpth{u''}{p''}$ is positive, then \ih:(\ref{it:ppc-necessary-strsp-match}) also implies $\residus{\ppcsetred{G}}{\subtat{\mreda'}{11}} = \residus{\ppcsetred{D}}{\subtat{\mreda'}{2}} = \emptyset$, and $\matchOpth{u''}{p''} = \fail$, along with $t' \in \MForms$, implies $t' = \kid$. In both cases we obtain $\residus{\strs(t)}{\mreda} = \emptyset$.

Assume $t = su$ where $s \notin \MForms$.  Then, $t' \in \MForms\,$ implies $t = su \mred{\mreda'} s'u' \mred{\mreda''} t'$, where $\mreda'$ preserves 1 and 2, and either $s' \in \DStructs$ or $s'$ is an abstraction, \ie\ $s' \in \MForms$.  In turn, \thelem~\ref{rsl:mred-projection-good-behavior}:(\ref{it:mred-projection-good-tgt}) implies $s \mred{\subtatwide{\mreda'}{1}} s'$.  
Therefore, the \ih:(\ref{it:ppc-necessary-mf}) applies, yielding that $\subtat{\mreda'}{1}$ uses $\strs(s)$ and $\residu{\strs(s)}{\subtat{\mreda'}{1}} = \emptyset$. 
Observe that $s \notin \MForms$ and $s' \in \MForms$ imply $s \neq s'$, then $s \notin \NForms$, hence $\strspos(t) = 1 \strspos(s)$.
Hence \theLem~\ref{rsl:mred-projection-good-behavior}:(\ref{it:mred-projection-good-uses}) and \theLem~\ref{rsl:mred-projection-good-behavior}:(\ref{it:mred-projection-good-residuals}) implies that $\mreda'$ uses $\strs(t)$ and $\residus{\strs(t)}{\mreda'} = \emptyset$ respectively. Thus we conclude.
%Therefore, the \ih\ (\ref{it:ppc-necessary-mf}) can be applied, yielding that $\strs(s) \neq \emptyset$ and $\subtat{\mreda'}{1}$ uses $\strs(s)$.  Observe that $s$ is not a normal form, thus implying $\strspos(t) = 1 \strspos(s)$.  
%\theLem~\ref{rsl:mred-projection-good-behavior}:(\ref{it:mred-projection-good-uses}) allows us to conclude.

Finally, the remaining case  $t = s u$ where $s \in \DStructs$ contradicts $t \notin \MForms$.

\medskip
Item~(\ref{it:ppc-necessary-strsp-cmatch}).
Observe that $\cmatchOpth{u'}{p'}$ decided implies $p' \in \MForms$, and also $u' \in \MForms$ unless $\isbm{p'}{\theta}$.  We consider the following cases depending on whether $p$ is in $\MForms$ or not and likewise for $u$.

Assume $p \notin \MForms$, so that $G = \strspos(p)$ and $D = \emptyset$. 
In this case, $p' \in \MForms$ implies that the \ih:(\ref{it:ppc-necessary-mf}) can be applied on $p \mred{\mredb} p'$. We obtain that $\mredb$ uses $\ppcsetred{G}$ and $\residus{\ppcsetred{G}}{\mredb} = \emptyset$, which suffices to conclude. 

Assume $p \in \MForms$ and $u \notin \MForms$, so that $\cmatchOpth{u}{p} = \wait$ implies $\neg \isbm{p}{\theta}$, and therefore $G = \emptyset$ and $D = \strspos(u)$. Observe that $p \in \MForms$, $\cmatchOpth{u}{p} = \wait$ and $p \mred{\mredb} p'$ imply $\neg \isbm{p'}{\theta}$, so that $u' \in \MForms$. Therefore, the \ih:(\ref{it:ppc-necessary-mf}) can be applied on $u \mred{\mredc} u'$. We conclude like in the previous case.

Assume $p, u \in \MForms$, so that $\cmatchOpth{u}{p} = \wait$ implies $p = p_1 p_2$ and $u = u_1 u_2$.
Then $G = 1 G_1 \cup 2 G_2$ and $D = 1 D_1 \cup 2 D_2$, where $\fnsmth{u_i}{p_i} = \pair{G_i}{D_i}$ for $i = 1,2$.
Moreover, it is straightforward to verify that both $\mredb$ and $\mredc$ preserve 1 and 2, so that \thelem~\ref{rsl:mred-projection-good-behavior}:(i) implies $p' = p'_1 p'_2$, $u' = u'_1 u'_2$, and $p_i \mred{\subtatwide{\mredb}{i}} p'_i$ and $u_i \mred{\subtatwide{\mredc}{i}} u'_i$ for $i = 1,2$.
On the other hand, the hypotheses imply the existence of some $k \in \set{1,2}$ verifying $\cmatchOpth{u_k}{p_k} = \wait$ and $\cmatchOpth{u'_k}{p'_k}$ decided. Therefore, the \ih:(\ref{it:ppc-necessary-strsp-cmatch}) can be applied yielding that $\subtat{\mredb}{k}$ uses $\ppcsetred{(G_k)}$ or $\subtat{\mredc}{k}$ uses $\ppcsetred{(D_k)}$. Hence, \thelem~\ref{rsl:mred-projection-good-behavior}:(\ref{it:mred-projection-good-uses}) implies that $\mredb$ uses $\ppcsetred{G}$ or $\mredc$ uses $\ppcsetred{D}$.

Moreover, $\cmatchOpth{u'}{p'}$ positive implies $\cmatchOpth{u'_i}{p'_i}$ positive for $i = 1,2$.
For each $i$, observe that $\cmatchOpth{u}{p} = \wait$ implies $\cmatchOpth{u_i}{p_i} \neq \fail$. 
If $\cmatchOpth{u_i}{p_i} = \wait$, then the \ih:(\ref{it:ppc-necessary-strsp-cmatch}) implies $\residus{(\ppcsetred{G}_i)}{\subtat{\mredb}{i}} = \residus{(\ppcsetred{D}_i)}{\subtat{\mredc}{i}} = \emptyset$; 
if $\cmatchOpth{u_i}{p_i}$ is positive, then \thelem~\ref{rsl:strsp-empty-if-positive-match} implies $G_i = D_i = \emptyset$. 
Hence \thelem~\ref{rsl:mred-projection-good-behavior}:(\ref{it:mred-projection-good-residuals}) yields $\residus{\ppcsetred{G}}{\mredb} = \residus{\ppcsetred{D}}{\mredc} = \emptyset$.
\end{proof}
}  % \onlyReport

%\subsubsection{Normalisation of \strs\ -- main proofs}
%In this section we conclude the proofs showing that \strs\ always selects \emph{necessary} and \emph{\ngrip} sets of redexes.

\begin{proposition}
\label{rsl:ppc-strategy-necessary}
Let $t \mred{\mreda} t'$ where $t \notin \NForms$ and $t' \in \NForms$. Then $\mreda$ uses $\strs(t)$.
\end{proposition}

\onlyPaper{
\begin{proof}
We prove the proposition simultaneously with the following two statements.

Let $p \mred{\mredb} p'$ and $u \mred{\mredc} u'$ where $p', u' \in \NForms$, $\pair{G}{D} = \fnsmth{u}{p} \neq \pair{\emptyset}{\emptyset}$, and $\cmatchOpth{u}{p} = \cmatchOpth{u'}{p'} = \wait$ (resp. $\matchOpth{u}{p} = \matchOpth{u'}{p'} = \wait$). Then $\mredb$ uses $\ppcsetred{G}$ or $\mredc$ uses $\ppcsetred{D}$.

The structure of the proof is similar to that of \thelem~\ref{rsl:ppc-necessary-strsp}, resorting to \thelem~\ref{rsl:mred-projection-good-behavior}, \thelem~\ref{rsl:ppc-necessary-strsp}, and \thelem~\ref{rsl:strsp-empty-if-positive-match}. \CompleteInReport .
\end{proof}
}

\onlyReport{
\begin{proof}
We prove the following three statements simultaneously, where $t, u, p$ are terms.
\begin{enumerate}[(i)]
\minitem 
\label{it:ppc-strategy-necessary}
The statement of the proposition. 
\minitem
\label{it:ppc-strategy-strsp-cmatch-wait}
Let $p \mred{\mredb} p'$ and $u \mred{\mredc} u'$ where $p', u' \in \NForms$, $\pair{G}{D} = \fnsmth{u}{p} \neq \pair{\emptyset}{\emptyset}$, and $\cmatchOpth{u}{p} = \cmatchOpth{u'}{p'} = \wait$. 
Then $\mredb$ uses $\ppcsetred{G}$ or $\mredc$ uses $\ppcsetred{D}$.
\minitem 
\label{it:ppc-strategy-strsp-match-wait}
Let $p \mred{\mredb} p'$ and $u \mred{\mredc} u'$ where $p', u' \in \NForms$, $\pair{G}{D} = \fnsmth{u}{p} \neq \pair{\emptyset}{\emptyset}$, and $\matchOpth{u}{p} = \matchOpth{u'}{p'} = \wait$. Then $\mredb$ uses $\ppcsetred{G}$, or $\mredc$ uses $\ppcsetred{D}$. 
\end{enumerate}
As in Lem.~\ref{rsl:ppc-necessary-strsp}, item~(\ref{it:ppc-strategy-strsp-match-wait}) follows from item~(\ref{it:ppc-strategy-strsp-cmatch-wait}). So we prove the others, by induction on $\tsize{t}+ \tsize{u} +\tsize{p}$.

\medskip
Item~(\ref{it:ppc-strategy-necessary}). If $t$ is either a matchable or a variable, then $t$ is a normal form, contradicting the hypotheses so that let consider that $t$ is an application or an abstraction.

Assume $t = (\lpth p.s) u$ and $\matchOpth{u}{p}$ decided, so that $\strs(t) = \set{\pair{t}{\epsilon}}$.  
Suppose $\mreda$ does not use $\strs(t)$, so that $t' = (\lpth p'.s') u'$, and $\mreda$ preserves 11, 12 and 2. This implies $p \mred{} p'$ and $u \mred{} u'$, \confer\ \thelem~\ref{rsl:mred-projection-good-behavior}:(\ref{it:mred-projection-good-tgt}), so that \thelem~\ref{rsl:reduction-mantains-decided-match} implies $\matchOpth{u'}{p'}$ decided, contradicting $t'$ being a normal form. Thus we conclude.

Assume $t = (\lpth p.s) u$, $\matchOpth{u}{p} = \wait$ and $\pair{G}{D} = \fnsmth{u}{p} \neq \pair{\emptyset}{\emptyset}$.
We define $\mreda'$ as follows. 
If $\mreda$ includes the contraction of, at least, one head step, \ie\ if there exists some $n \leq \rlength{\mreda}$ verifying $\pair{\rtgt{\mredunt{\mreda}{n-1}}}{\epsilon} \in \mredel{\mreda}{n}$, we consider the minimum such $n$ and define $\mreda' \eqdef \mredunt{\mreda}{n-1}$. 
Otherwise, $\mreda' \eqdef \mreda$.
In both cases $t \mred{\mreda'} (\lpth p'.s') u'$ and $\mreda'$ preserves 11 and 2, so that \thelem~\ref{rsl:mred-projection-good-behavior}:(\ref{it:mred-projection-good-tgt}) implies $p \mred{\subtatwide{\mreda'}{11}} p'$ and $u \mred{\subtatwide{\mreda'}{2}} u'$. Notice that in the latter case, $p', u' \in \NForms$.
In both cases we obtain that $\subtat{\mreda'}{11}$ uses $\ppcsetred{G}$ or $\subtat{\mreda'}{2}$ uses $\ppcsetred{D}$, if $\matchOpth{u'}{p'}$ decided by \thelem~\ref{rsl:ppc-necessary-strsp}:(\ref{it:ppc-necessary-strsp-match}), otherwise by the  \ih\ (\ref{it:ppc-strategy-strsp-match-wait}). 
Recalling that in this case, $\strspos(t) = 11G \cup 2D$, we conclude by applying \thelem~\ref{rsl:mred-projection-good-behavior}:(\ref{it:mred-projection-good-uses}).

Assume $t = (\lpth p.s)u$, $\matchOpth{u}{p} = \wait$, and $\fnsmth{u}{p} = \pair{\emptyset}{\emptyset}$.
A simple argument by contradiction based on \thelem~\ref{rsl:ppc-necessary-strsp}:(\ref{it:ppc-necessary-strsp-match}) implies that $t' = (\lpth p'.s') u'$ and $\mreda$ preserves 11, 12 and 2. Therefore, \thelem~\ref{rsl:mred-projection-good-behavior}:(\ref{it:mred-projection-good-tgt}) implies $p \mred{\subtatwide{\mreda}{11}} p'$ and similarly for $s$ and $u$. 
If $p \notin \NForms$, so that $\strspos(t) = 11 \strspos(p)$, then the \ih\ (\ref{it:ppc-strategy-necessary}) can be applied to obtain that $\subtat{\mreda}{11}$ uses $\strs(p)$, so that \thelem~\ref{rsl:mred-projection-good-behavior}:(\ref{it:mred-projection-good-uses}) allows to conclude. 
The remaining cases, \ie\ $p \in \NForms, s \notin \NForms$ and $p, s \in \NForms$ respectively, can be handled similarly.

Assume $t = s u$ and $s \notin \Abstract$. 
If there exists some $n$ such that $\rtgt{\mredel{\mreda}{n}} = s' u'$ and $s' \in \Abstract$, then we consider the minimal such $n$, and let $\mreda' = \mredunt{\mreda}{n}$. It is easy to obtain that $\mreda'$ preserves 1 and 2, so that \thelem~\ref{rsl:mred-projection-good-behavior}:(\ref{it:mred-projection-good-tgt}) implies $s \mred{\subtatwide{\mreda'}{1}} s'$. Observe that $s \notin \NForms$, implying $\strspos(t) = 1 \strspos(s)$. Moreover, $s \in \DStructs$ would imply $s' \in \DStructs$, so that $s \notin \MForms$. Hence, a projection argument similar to that used in previous cases, based on \thelem~\ref{rsl:ppc-necessary-strsp}:(\ref{it:ppc-necessary-mf}), allows to conclude.
%. On the other hand, observe that in this case, $s \notin \NForms$, implying $\strspos(t) = 1 \strspos(s)$. 
%A projection argument similar to that used in previous cases, based on $s \mred{\subtatwide{\mreda'}{1}} s'$ and \ih\ (\ref{it:ppc-strategy-necessary}), allows to conclude.
Otherwise $s$ does not reduce to an abstraction, implying $t = s' u'$, $\mreda$ preserves 1 and 2, and $s', u' \in \NForms$. Again, a projection argument applies, to $s \mred{\subtatwide{\mreda}{1}} s'$ if $s \notin \NForms$, to $u \mred{\subtatwide{\mreda}{2}} u'$ otherwise, based on \ih\ (\ref{it:ppc-strategy-necessary}).

Assume $t = \lpth p.s$.
Then,  $t' = (\lpth p'.s')$, $\mreda$ preserves 1 and 2, and $p', s' \in \NForms$. A projection argument based on \ih\ (\ref{it:ppc-strategy-necessary}) applies to $p \mred{\subtatwide{\mreda}{1}} p'$ or $s \mred{\subtatwide{\mreda}{2}} s'$, depending on whether $p \in \NForms$.

\medskip
Item~(\ref{it:ppc-strategy-strsp-cmatch-wait}).

Assume $p \notin \MForms$, so that $G = \strspos(p)$ and $D = \emptyset$. 
The hypotheses imply $\strspos(p) \neq \emptyset$, and then $p$ is not a normal form.
Therefore, item~(\ref{it:ppc-strategy-necessary}) just proved applies to $p \mred{\mredb} p'$, which suffices to conclude.

Assume $p \in \MForms$, $\neg \isbm{p}{\theta}$, $u \notin \MForms$. In this case, $G = \emptyset$ and $D = \strspos(u)$. Hence, an argument similar to that of the previous case applies on $u \mred{\mredc} u'$.

Assume $p, u \in \MForms$. 
In this case, $\cmatchOpth{u}{p} = \wait$ implies $p = p_1 p_2$ and $u = u_1 u_2$, so that $G = 1 G_1 \cup 2 G_2$ and $D = 1 D_1 \cup 2 D_2$, where $\fnsmth{u_i}{p_i} = \pair{G_i}{D_i}$ for $i = 1,2$.
The assumption $p, u \in \MForms$ also implies $p' = p'_1 p'_2$, $u' = u'_1 u'_2$, and both $\mredb$ and $\mredc$ preserve 1 and 2. 
Then \thelem~\ref{rsl:mred-projection-good-behavior}:(\ref{it:mred-projection-good-tgt}) implies $p_i \mred{\subtatwide{\mredb}{i}} p'_i$ and $u_i \mred{\subtatwide{\mredc}{i}} u'_i$ for $i = 1,2$.
Moreover, $\pair{G}{D} \neq \pair{\emptyset}{\emptyset}$ implies $\pair{G_k}{D_k} \neq \pair{\emptyset}{\emptyset}$ for some $k \in \set{1,2}$.
Notice that $\cmatchOpth{u_k}{p_k}$ being positive (resp. $\fail$) contradicts \thelem~\ref{rsl:strsp-empty-if-positive-match} (resp. $\cmatchOpth{u}{p} = \wait$).
Then $\cmatchOpth{u_k}{p_k} = \wait$, so that either the \ih\ (\ref{it:ppc-strategy-strsp-cmatch-wait}) or \thelem~\ref{rsl:ppc-necessary-strsp}:(\ref{it:ppc-necessary-strsp-cmatch}) applies, depending on whether $\cmatchOpth{u'_k}{p'_k}$ is $\wait$ or positive.
In either case, we obtain that $\subtat{\mredb}{k}$ uses $\ppcsetred{G}_k$, or $\subtat{\mredc}{k}$ uses $\ppcsetred{D}_k$.
Thus \thelem~\ref{rsl:mred-projection-good-behavior}:(\ref{it:mred-projection-good-uses}) allows to conclude.
\end{proof}
}

\begin{lemma}
\label{rsl:above-strs-then-fail}
Let $t \mred{\mreda} t'$, $\ppcstep{b} \in \residus{\strs(t)}{\mreda}$, and $\ppcstep{a}$ verifying $\ppcstep{a} < \ppcstep{b}$. Then $\ppcstep{a}$ is a matching failure.
\end{lemma}

\onlyPaper{
\begin{proof}
The structure of the proof is similar to that of
  \theprop~\ref{rsl:ppc-strategy-necessary}. The two statements which are proved simultaneously follow.

Let $p \mred{\mredb} p'$ and $u \mred{\mredc} u'$ such that $\cmatchOpth{u}{p} = \wait$ (resp. $\matchOpth{u}{p} = \wait$), $\ppcstep{b} \in \residus{\ppcsetred{G}}{\mredb}$ or $\ppcstep{b} \in \residus{\ppcsetred{D}}{\mredc}$ where $\fnsmth{u}{p} = \pair{G}{D}$, and $\ppcstep{a}$ verifying $\ppcstep{a} < \ppcstep{b}$. 
Then $\ppcstep{a}$ is a matching failure.

\CompleteInReport .
\end{proof}
}

\onlyReport{
\begin{proof}
We prove the following, more general statement.
\begin{enumerate}[(i)]
\item 
\label{it:above-strs-then-fail}
The lemma statement. 
\minitem 
\label{it:above-strsp-then-fail-cmatch}
Let $p \mred{\mredb} p'$ and $u \mred{\mredc} u'$ such that $\cmatchOpth{u}{p} = \wait$, $\ppcstep{b} \in \residus{\ppcsetred{G}}{\mredb}$ or $\ppcstep{b} \in \residus{\ppcsetred{D}}{\mredc}$ where $\fnsmth{u}{p} = \pair{G}{D}$, and $\ppcstep{a}$ verifying $\ppcstep{a} < \ppcstep{b}$. 
Then $\ppcstep{a}$ is a matching failure.
\minitem 
\label{it:above-strsp-then-fail-match}
Let $p \mred{\mredb} p'$ and $u \mred{\mredc} u'$ such that $\matchOpth{u}{p} = \wait$, $\ppcstep{b} \in \residus{\ppcsetred{G}}{\mredb}$ or $\ppcstep{b} \in \residus{\ppcsetred{D}}{\mredc}$ where $\fnsmth{u}{p} = \pair{G}{D}$, and $\ppcstep{a}$ verifying $\ppcstep{a} < \ppcstep{b}$. 
Then $\ppcstep{a}$ is a matching failure.
\end{enumerate}

As in Lem.~\ref{rsl:ppc-necessary-strsp}, item~(\ref{it:ppc-strategy-strsp-match-wait}) follows from item~(\ref{it:ppc-strategy-strsp-cmatch-wait}). So we prove the others, by induction on $\tsize{t}+ \tsize{u} +\tsize{p}$.

\medskip
We prove item~(\ref{it:above-strs-then-fail}).
If $t$ is either a variable or a matchable, then $t$ is a normal form, contradicting the existence of $\ppcstep{b}$.

Assume $t = (\lpth p.s) u$ and $\matchOpth{u}{p}$ decided, implying $\strs(t) = \set{\pair{t}{\epsilon}}$. 
Then, a straightforward inductive argument on $\rlength{\mreda}$ yields that $\residus{\strs(t)}{\mreda} = \emptyset$ or $\ppcstep{b} = \pair{t'}{\epsilon}$, contradicting in both cases the existence of $\ppcstep{a}$. Thus we conclude.

Assume $t = (\lpth p.s) u$, $\matchOpth{u}{p} = \wait$, and $\pair{G}{D} = \fnsmth{u}{p} \neq \pair{\emptyset}{\emptyset}$.
Then $\strspos(t) = 11G \cup 2D$.
Consider $\mreda', \mreda''$ such that $\mreda = \mreda' ; \mreda''$, $t \mred{\mreda'} t'' = (\lpth p'.s') u' \mred{\mreda''} t'$, $\mreda'$ preserves 11 and 2, and either $\mreda'' = \emptyred{t'}$ or $\pair{t''}{\epsilon} \in \mredel{\mreda''}{1}$.
\theLem~\ref{rsl:mred-projection-good-behavior}:(\ref{it:mred-projection-good-tgt}) implies $p \mred{\subtatwide{\mreda'}{11}} p'$ and $u \mred{\subtatwide{\mreda'}{2}} u'$.
If $\matchOpth{u'}{p'}$ is positive, then \thelem~\ref{rsl:ppc-necessary-strsp}:(\ref{it:ppc-necessary-strsp-match}) implies $\residus{\ppcsetred{G}}{\subtat{\mreda'}{11}} = \residus{\ppcsetred{D}}{\subtat{\mreda'}{2}} = \emptyset$, and therefore \thelem~\ref{rsl:mred-projection-good-behavior}:(\ref{it:mred-projection-good-residuals}) yields $\residus{\strs(t)}{\mreda'} = \emptyset$.
If $\matchOpth{u'}{p'} = \fail$ and $\mreda'' \neq \emptyred{t''}$, then it is immediate to obtain $t' = \kid$, a normal form, contradicting the existence of $\ppcstep{b}$.
Therefore, assume $\matchOpth{u'}{p'} \in \set{\wait,\fail}$ and $\mreda'' = \emptyred{t''}$, so that $\mreda = \mreda'$ and $t' = (\lpth p'.s') u'$.
An analysis of the ancestor of $\ppcstep{b}$, which is some $\ppcstep{b_0} \in \strs(t)$, along with \thelem~\ref{rsl:mred-projection-good-behavior}:(\ref{it:mred-projection-good-residuals}), yields that $b = 11b'$ where $\ppcstep{b'} \in \residus{\ppcsetred{G}}{\subtat{\mreda}{11}}$ or $b = 2b'$ where  $\ppcstep{b'} \in \residus{\ppcsetred{D}}{\subtat{\mreda}{2}}$, implying respectively that $\ppcstep{b'} \in \ROccur{p'}$ or $\ppcstep{b'} \in \ROccur{u'}$. Let $\ppcstep{a}$ verifying $\ppcstep{a} < \ppcstep{b}$. 
If $a = \epsilon$, then $\matchOpth{u'}{p'} = \fail$, \ie\ $\ppcstep{a}$ is a matching failure.
Otherwise, $a = 11a'$ or $a = 2a'$, so that $\ppcstep{a'} \in \ROccur{p'}$ or $\ppcstep{a'} \in \ROccur{u'}$ respectively, and $\ppcstep{a'} < \ppcstep{b'}$. Therefore \ih\ (\ref{it:above-strsp-then-fail-match}) implies that $\ppcstep{a'}$ is a matching failure, which suffices to conclude.

Assume $t = (\lpth p.s) u$, $\matchOpth{u}{p} = \wait$, and $\fnsmth{u}{p} = \pair{\emptyset}{\emptyset}$.
Observe that $t \mred{\mredb} (\lpth p''.s'') u''$ such that $\matchOpth{u''}{p''}$ is decided would contradict $\fnsmth{u}{p} = \pair{\emptyset}{\emptyset}$; \confer\ \thelem~\ref{rsl:mred-projection-good-behavior}:(\ref{it:mred-projection-good-tgt}) and \thelem~\ref{rsl:ppc-necessary-strsp}:(\ref{it:ppc-necessary-strsp-match}) considering a minimal such $\mredb$. 
Therefore, $t' = (\lpth p'.s') u'$, $\mreda$ preserves 11, 12 and 2, and $\matchOpth{u'}{p'} = \wait$.
If $p' \notin \NForms$, so that $\strspos(t) = 11 \strspos(p)$, then \thelem~\ref{rsl:mred-projection-good-behavior}:(\ref{it:mred-projection-good-tgt}) and \thelem~\ref{rsl:mred-projection-good-behavior}:(\ref{it:mred-projection-good-residuals}) imply $p \mred{\subtatwide{\mreda}{11}} p'$ and $b = 11b'$ where $\ppcstep{b'} \in \residus{\strs(p)}{\subtat{\mreda}{11}}$ respectively.
Observe $\matchOpth{u'}{p'} = \wait$ implies that $\pair{t'}{\epsilon} \notin \ROccur{t'}$. Then $\ppcstep{a} < \ppcstep{b}$ implies $a = 11a'$, so that $\ppcstep{a'} \in \ROccur{p'}$, and $\ppcstep{a'} < \ppcstep{b'}$. Hence the  \ih\ (\ref{it:above-strs-then-fail}) applies, which suffices to conclude.
The other cases ($p' \in \NForms$ and $s' \notin \NForms$, and $p', s' \in \NForms$) admit analogous arguments. 

Assume $t = su$ where $s \notin \Abstract$.
Let $\mreda', \mreda''$ such that $\mreda = \mreda' ; \mreda''$, $t \mred{\mreda'} s'u' \mred{\mreda''} t'$, $\mreda'$ preserves 1 and 2, and either $\mreda'' = \emptyred{t'}$ or $\pair{s' u'}{\epsilon} \in \mredel{\mreda''}{1}$.
\theLem~\ref{rsl:mred-projection-good-behavior}:(\ref{it:mred-projection-good-tgt}) implies $s \mred{\subtatwide{\mreda'}{1}} s'$ and $u \mred{\subtatwide{\mreda'}{2}} u'$.
\begin{itemize}
\minitem 
If $s' \in \Abstract$, then $s \neq s'$ implying that $s$ is not a normal form, and therefore $\strspos(t) = 1 \strspos(s)$.
Moreover, $s \notin \MForms$; notice that $s \in \DStructs$ would imply $s' \in \DStructs$. 
Therefore, \thelem~\ref{rsl:ppc-necessary-strsp}:(\ref{it:ppc-necessary-strsp-match}) implies $\residus{\strs(s)}{\subtat{\mreda'}{1}} = \emptyset$, so that \thelem~\ref{rsl:mred-projection-good-behavior}:(\ref{it:mred-projection-good-residuals}) contradicts the existence of $\ppcstep{b}$. Thus we conclude.
\minitem
If $s' \notin \Abstract$, then $\mreda'' = \emptyred{s'u'}$, so that $\mreda = \mreda'$ and $t' = s'u'$. Moreover, 
$\pair{t'}{\epsilon} \notin \ROccur{t'}$.
If $s$ is not a normal form, so that $\strspos(t) = 1 \strspos(s)$, then \thelem~\ref{rsl:mred-projection-good-behavior}:(\ref{it:mred-projection-good-residuals}) implies $b = 1 b'$ where $\ppcstep{b'} \in \residus{\strs(s)}{\subtat{\mreda}{1}}$. On the other hand, $\ppcstep{a} < \ppcstep{b}$ implies $a = 1 a'$ where $\ppcstep{a'} \in \ROccur{s'}$. 
Then the  \ih\ (\ref{it:above-strs-then-fail}) applies, which suffices to conclude.
If $s$ is a normal form, so that $\strspos(t) = 2 \strspos(u)$, then a similar argument applies.
\end{itemize}

Assume $t = \lpth p.s$. Then $\mreda$ preserves 1 and 2, so that $t' = \lpth p'.s'$ and \thelem~\ref{rsl:mred-projection-good-behavior}:(\ref{it:mred-projection-good-tgt}) implies $p \mred{\subtatwide{\mreda}{1}} p'$ and $s \mred{\subtatwide{\mreda}{2}} s'$. 
A projection argument based on \ih\ (\ref{it:above-strs-then-fail}) analogous to those used in previous cases, on $p \mred{\subtatwide{\mreda}{1}} p'$ or $s \mred{\subtatwide{\mreda}{2}} s'$ depending whether $p \in \NForms$, allows to conclude.

\medskip
We prove item~(\ref{it:above-strsp-then-fail-cmatch}). There are three cases to analyse, given $\cmatchOpth{u}{p} = \wait$.

If $p \notin \MForms$, then $\ppcsetred{G} = \strs(p)$ and $\ppcsetred{D} = \emptyset$, so that $\ppcstep{b} \in \residus{\ppcsetred{G}}{\mredb} = \residus{\strs(p)}{\mredb}$. Ih\ (\ref{it:above-strs-then-fail}) on $p \mred{\mredb} p'$ suffices to conclude.

If $p \in \MForms$ and $u \notin \MForms$, so that $\ppcsetred{G} = \emptyset$ and $\ppcsetred{D} = \strs(u)$, then an analogous argument applies.

If $p = p_1 p_2$, $u = u_1 u_2$, and $p, u \in \MForms$, then $G = 1 G_1 \cup 2 G_2$ and $D = 1 D_1 \cup 2 D_2$, where $\pair{G_i}{D_i} = \fnsmth{u_i}{p_i}$ for $i = 1,2$. 
Moreover, $p, u \in \MForms$ implies that $\mredb$ and $\mredc$ preserve 1 and 2, $p' = p'_1 p'_2$ and $u' = u'_1 u'_2$.
\theLem~\ref{rsl:mred-projection-good-behavior}:(\ref{it:mred-projection-good-tgt}) yields $p_i \mred{\subtatwide{\mredb}{i}} p'_i$ and $u_i \mred{\subtatwide{\mredc}{i}} u'_i$ for $i = 1,2$.
In turn, \theLem~\ref{rsl:mred-projection-good-behavior}:(\ref{it:mred-projection-good-residuals}) implies $b = i b'$ where $\ppcstep{b'} \in \residus{\ppcsetred{G}_i}{\subtat{\mredb}{i}}$ or $\residus{\ppcsetred{D}_i}{\subtat{\mredc}{i}}$, for some $i \in \set{1,2}$.
Observe that $\cmatchOpth{u_i}{p_i} = \fail$ would contradict $\cmatchOpth{u}{p} = \wait$, and $\cmatchOpth{u_i}{p_i}$ positive would imply $G_i = D_i = \emptyset$ by \thelem~\ref{rsl:strsp-empty-if-positive-match}. Therefore $\cmatchOpth{u_i}{p_i} = \wait$.
Observe that neither $\pair{p'}{\epsilon}$ nor $\pair{u'}{\epsilon}$ are \tredexes, so that $\ppcstep{a} < \ppcstep{b}$ implies $a = i a'$. 
Hence the \ih\ (\ref{it:above-strsp-then-fail-cmatch}), applied on $p_i \mred{\subtatwide{\mredb}{i}} p'_i$ and $u_i \mred{\subtatwide{\mredc}{i}} u'_i$, allows to conclude.
\end{proof}
}

\begin{proposition}
\label{rsl:strs-non-gripping}
Let $t$ be a term not in normal form. Then $\strs(t)$ is a non-gripping set.
\end{proposition}

\begin{proof}
Let $t \mred{\mredd} u$, $\ppcstep{a} \in \ROccur{u}$, $\ppcstep{b} \in
\residus{\strs(t)}{\mredd}$; it suffices to deduce that $\ppcstep{b}$
does not grip $\ppcstep{a}$. If $\ppcstep{a} \not < \ppcstep{b}$, then 
we immediately conclude. 
 If $\ppcstep{a} < \ppcstep{b}$, then \thelem~\ref{rsl:above-strs-then-fail} entails that
$\ppcstep{a}$ is a matching failure so $\ppcstep{b}$
cannot  grip $\ppcstep{a}$.
\end{proof}

%%% Local Variables: 
%%% mode: latex
%%% TeX-master: "article"
%%% End: 

\subsection{\l-Calculus with Parallel-Or}

The lambda calculus extended with parallel-or also falls within the scope of our abstract proof.
Its terms are given by the grammar:
\begin{center}
$\begin{array}{rcl}
	t & ::= & x \ | \ \l x.t \ |\  t \, t \ |\  \kor(t,t) \ |\  \kt 
\end{array}$
\end{center}
The reduction rules are
\begin{center}
$\begin{array}{rcl}
	(\l x.s) u & \to & s \{x \leftarrow u \}  \\
	\kor(t,\kt) & \to & \kt \\
	\kor(\kt,t) & \to & \kt
\end{array}$
\end{center}
It may be seen as an ARS under the standard reading of each of its elements. Two comments on
this. First the notion of gripping. A step $\pair{s}{p}$ grips a
step $\pair{s}{q}$ where $\subtat{s}{q}=(\l y.u')\, v'$, if $q1\leq p$ and $\subtat{s}{p}$ has a free ocurrence of
$y$ (we assume the standard variable convention). Second, the fact that although this is an
almost-orthogonal higher-order rewrite system, from the point of view of the underlying ARS it enjoys
semantic orthogonality since the critical pair is trivial.  

The reduction strategy $\strs$ is defined by means of an auxiliary function $\strspos$ that gives the positions of the steps to be selected, as described for \theppc\ in Sec.~\ref{sec:ppc-strategy}. In turn, $\strspos$ is defined as follows
\begin{center}
$
\begin{array}{rcll}
\strspos((\l x.s)\,u) & \eqdef & \set{\epsilon} \\
\strspos(\kor(\kt,u)) & \eqdef & \set{\epsilon} \\
\strspos(\kor(u,\kt)) & \eqdef & \set{\epsilon} \\
\strspos(s\,u) & \eqdef & 1 \strspos(s) & \textif s \neq \l x.s' \textand s \notin \NForms \\
\strspos(s\,u) & \eqdef & 2 \strspos(u) & \textif s \neq \l x.s' \textand s \in \NForms \\
\strspos(\l x.t) & \eqdef & 1 S(t) \\
%\strspos(\kor(\kf,\kf)) & \eqdef & \set{\epsilon} \\
\strspos(\kor(s,u)) & \eqdef & 1 \strspos(s) \cup 2 \strspos(u) & 
\textif s \neq \kt \textand u \neq \kt \\
% \textand \neg(t = u = \kf) \\
%\multicolumn{4}{@{}l}{\strspos(\kt) = \strspos(\kf) = \strspos(\ka) = \strspos(\kb) = \strspos(\kc) \eqdef \emptyset}
\strspos(\kt) & \eqdef & \emptyset
\end{array}
$
\end{center}

This strategy may be proved to produce necessary and \ngrip\ sets of redexes following the
lines of the (more complicated) proofs developed for \theppc.  As a consequence,
Thm.~\ref{rsl:necessary-and-ngrip-then-normalises} is applicable and allows us to infer that $\strs$ is normalising.

%%% Local Variables: 
%%% mode: latex
%%% TeX-master: "article"
%%% End: 

\section{Conclusions}
\label{sec:conclusions}
Relying on an axiomatic presentation of rewriting~\cite{thesis-mellies}, we 
study normalisation for a
wide class of rewriting systems.  The main result of this paper states that
multistep strategies that contract {\it sets} of \emph{necessary} and
\emph{\ngrip} steps are normalising, \ie\ they  reach a normal form, if it
exists. 

This is particularly appealing for non-sequential rewrite systems,
in which terms that are not in normal form may not have any needed
redex, where strategies that contract only a {\it single step} rather than a {\it set of steps}
and rely only on the term itself to decide which redex to reduce,
 cannot be normalising.

We give a concrete example of such a phenomenon by means of the
pattern calculus \theppc, that fails to be sequential, and hence
includes reducible terms without any needed redex.  More precisely,
this behavior is manifested by the failure mechanism of \theppc.
Consider for example the term $ t = (\lpset{x} \ka \,\kb \, \kc. \kb
\, \kd) (\ka \, r_1 r_2) $ where $r_1$ and $r_2$ are redexes.  If
$r_1$ rewrites to $\kd$, then $t$ can be reduced to $ t' = (\lpset{x}
\ka \,\kb \, \kc. \kb \, \kd) (\ka \, \kd \, r_2) $ which rewrites to
the normal-form $\Id$ in one step, because the match of the pattern
$\ka \, \kb \, \kc$ against the argument $\ka \, \kd \, r_2$ yields
$\fail$.  A similar situation holds if $r_2$ rewrites to $\kd$.
Consequently, either $r_1$ or $r_2$ could be selected to yield a
normal form from $t$.  But choosing always $r_1$ would be a bad
decision for another terms, as for example $u = (\lpset{x} \ka \,\kb
\, \kc. \kb \, \kd) (\ka \, r'_1 r'_2)$, where $r'_1$ leads to an
infinite reduction, whilst $r'_2$ rewrites to $\kd$.  An analogous
reasoning invalidates the selection of $r_2$.

Since the reduction strategy \strs\ (\confer\ Sec.~\ref{sec:ppc-strategy}) for \theppc\ chooses a set of redexes
(both \tredexes\ $r_1$ and $r_2$ are selected in our example $t$), it
is then a \emph{multistep} reduction strategy. We prove that
\strs\ computes necessary and \ngrip\ sets of steps. Following
the above mentioned abstract normalisation result, this implies that
the multistep strategy is normalising for \theppc. This result can then be seen as 
an extension of needed normalising strategies to non-sequential rewrite systems.
Moreover, our strategy \strs\ coincides with the leftmost-outermost strategy when restricting
\theppc\ to the $\lambda$-calculus. 

\delia{Another interesting remark concerns the recent
  embedding~\cite{vanRaamsdonkvanOostrom:2014} of \theppc\ into higher-order pattern rewriting
  systems, which was motivated by the fact that one can understand
  some properties of \theppc\ by just looking at the corresponding
  properties of the image of the embedding. However, as explained in
  Section~\ref{sec:ppc-strategy}, the strategy \strs\ is not outermost-fair, so that
  no available normalisation result for higher-order rewriting can be
  applied in our case. More importantly, the results developed in this
  paper can be applied to other higher-order rewriting systems for
  which outermost-fair strategies are, in general, difficult to
  compute or to express inductively. }

\delia{
\medskip
This work also shows that the notion of \emph{gripping} can be a
useful tool to study fine properties of reduction in $\l$-calculi.
We already noted, in Sec.~\ref{sec:gripping}, that gripping is used in an abstract proof of the finiteness of developments.
We cite other links between gripping and $\l$-calculi.
\begin{enumerate}
\item 
Gripping explains the size-exploding phenomenon described in \cite{DBLP:conf/csl/AccattoliL14}. Let $t_0 = y x x$ and $t_{n+1} = (\l x.t_n) (y x x)$. The term $t_n$ reduces in $n$ steps to a term whose size is \emph{exponential} in $n$, while the size of $t_n$ is lineal in $n$.
We observe that the $n$ redexes present in $t_n$ are all linked by gripping. \Eg, in 
$t_3 = (\l x_3.(\l x_2.(\l x_1.y x_1 x_1) (y x_2 x_2)) (y x_3 x_3)) (y x x)$
we have $a_3 \grip a_2 \grip a_1$, where $a_i$ is the redex corresponding to the bound $x_i$. The successive gripping between redexes produces the multiplication of variable occurrences, and thus the explosion in the size of the normal form.
\item
A link also exists between gripping and the box order on redexes in the linear substitution calculus \cite{DBLP:conf/popl/AccattoliBKL14}. In this calculus, the term $x [x / y] [y / z]$ has two substitution redexes, corresponding to the bound occurrences of the variables $x$ and $y$. In the box order, the $x$-redex precedes the $y$-redex. Beta-expansion of this term yields $(\l y. (\l x. x) y) z$, where the $x$-redex grips the $y$-redex.
\item
Finally, we observe that a
variant of gripping is used  in~\cite{DBLP:journals/tcs/EndrullisGKO11} to characterise the cases in which
$\alpha$-conversion is unavoidable in 
\cdelia{a calculus having as}{calculi containing the} rewrite
rule $\mu x. M \to M [x := \mu x.M]$.  \Eg, in the term $t = \mu
x.F(y,\mu y.x)$, the inner redex $\mu y.x$ grips the outer one. Observe that the step $t \sstep{} F(y,\mu y' . (\mu x.F(y,\mu y.x)))$ forces the renaming of the bound variable associated to the (residual of the) gripping redex.
\end{enumerate}

}

\medskip
The scope of our work could be expanded in
several ways.  First, we believe that the main ideas underlying
the definition of \strs\ for \theppc\ can lead to reduction strategies
for other abstract rewriting formats, such
as HRS, CRS or ERS.  These strategies could be proved to be
normalising by resorting to the abstract normalisation proof given in
this paper.  This would give a powerful extension of the results
in~\cite{sekar-rama} to higher-order rewriting.

A second research direction is to broaden the scope of the
normalisation proof presented in
Section~\ref{sec:reduction-strategies}.  More precisely, the abstract
proof has been instantiated for \theppc\ with the strategy $\strs$, which always
selects a subset of the \emph{outermost} \tredexes\ in a term.  On the
other hand, the proof does not apply to the
\emph{parallel-outermost} reduction strategy, which simultaneously
selects  \emph{all} the outermost \tredexes\ in a term.  This is due
to the fact that ${\cal A} \subseteq {\cal B}$ and ${\cal A}$
\ngrip\ does not imply ${\cal B}$ \ngrip.  For example, consider: 
$$ t = (\lpset{x} \ca \, x.\underbrace{D x}_{\ppcstep{b}}) (\underbrace{\Id (\ca \, \cb)}_{\ppcstep{a}}) $$
whose only \tredexes\ are $\ppcstep{a}$ and $\ppcstep{b}$. 
Let $\strs(t) = \set{\ppcstep{a}} \subseteq \set{\ppcstep{a}, \ppcstep{b}} = 
\mathcal{O}(t)$. Remark that
$\strs(t)$ is the set of redexes selected by the strategy
 $\strs$ and $\mathcal{O}(t)$ 
is the  set of all outermost \tredexes\ of $t$. 
The set $\strs(t)$ is indeed \ngrip. However, 
the set $\mathcal{O}(t)$ does not satisfies the \ngrip\ property in
the general case.  Indeed, contracting $\ppcstep{a}$ results in 
$$      \overbrace{(\lpset{x} \ca \, x.\underbrace{D x}_{\ppcstep{b'}}) (\ca \, \cb)}^{\ppcstep{c'}} $$
where $\residu{\ppcstep{b}}{\ppcstep{a}}{\ppcstep{b'}}$ and $\ppcstep{c'} \grip \ppcstep{b'}$.
Hence $\mathcal{O}(t)$ does not enjoy the \ngrip\ property.

We conjecture that some variation of the given proof could apply to the parallel-outermost
strategy in some cases, for example for \theppc.
In this perspective, it could be possible that the property of always selecting \emph{necessary}
sets of \tredexes\ could suffice to guarantee that a reduction strategy is normalising. A proof of
this conjecture, or a counterexample falsifying it, would be an interesting result in this
direction.

\subsection*{Acknowledgements: } 
To Vincent van Oostrom for having pointed out a mistake in a previous
version of this work.  To Yann R\'egis-Gianas for discussions
  on coinduction. To Beniamino Accattoli who provided valuable
  comments.  This work was partially supported by LIA INFINIS, the
  ECOS-Sud cooperation program between France and Argentina,
  and by the grants PUNQ of the Universidad Nacional de
  Quilmes and UBACyT of the Universidad de Buenos
  Aires, Argentina.

% \medskip
% Besides the proposed initiatives, which are of a theoretic nature, 
% we would like to test the practical feasibility of the strategy $\strs$ by developing an interpreter of
% \theppc\ based on it.  \edu{Esto no me convence: c'omo es que implementando la estrategia pod'es
%   testear la factibilidad pr'actica? Eso no puede saberse mirando la definici'on de la estratega? O
%   estudiando la estrategia? } \delia{ a menos que se diga otra cosa en vez de ``test'' que es lo feo.
% Uno puede decir simplemente que quiere una implementacion basada en esta estrategia.}

%%% Local Variables: 
%%% mode: latex
%%% TeX-master: "paper"
%%% End: 

\onlyReport{
\section{Appendix -- Projection of a step--\tredexset--\mredseq}
\label{sec:projection}

In this section, we give precise definions for the notions of \emph{projection} and \emph{preservations}. We also prove \thelem~\ref{rsl:mred-projection-good-behavior}, along with the needed auxiliary results.

\begin{notation}
\label{not:leq-mstep-redseq}
Let $\setredb \subseteq \ROccur{t}$ and $a \in \Pos{t}$. 
We write $\stepa \leq \setredb$ iff $a \leq b$ for all $\ppcstep{b} \in \setredb$.
Analogously, for every \redseq\ $\reda$ and $a \in \Pos{\rsrc{\reda}}$, we write $a \leq \reda$ iff for any $i \leq \rlength{\reda}$, $a \leq b_i$ where $\mredel{\reda}{i} = \pair{t_i}{b_i}$.
\end{notation}

\begin{definition}
\label{dfn:preserves-mstep-mred}
Let $\setredb$ be a \tredexset, and $a \in \Pos{\rsrc{\setredb}}$.
We say that $\setredb$ \textbf{preserves} $a$ iff all $\ppcstep{b} \in \setredb$ verify $b \not < a$, or equivalently, $a \leq b$ or $a \disj b$.
In turn, a \emph{\mredseq} $\mreda$ \textbf{preserves} $a$ iff all its elements do.
\end{definition}

\begin{definition}
\label{dfn:free-domin-parts}
If $\setredb$ preserves $a$, then we define the \textbf{free part} and the \textbf{\domintext\ part} of $\setredb$ \wrt\ $a$, written $\thefreewrt{\setredb}{a}$ and $\thedominwrt{\setredb}{a}$ respectively, as follows:
$\thefreewrt{\setredb}{a} \eqdef \set{\ppcstep{b} \in \setredb \sthat a \disj b}$ and 
$\thedominwrt{\setredb}{a} \eqdef \set{\ppcstep{b} \in \setredb \sthat a \leq b}$.% 
\footnote{A remark about the names ``free'' and ``\domintext\'' given to $\thefreewrt{\setredb}{a}$ and $\thedominwrt{\setredb}{a}$ follows. 
We recall that $\stepb$ is free from $\stepa$ (that is, $\stepb \freefrom \stepa$) iff $\stepa \not\leq \stepb$, \ie\ $\stepb < \stepa$ or $\stepb \disj \stepa$. The former possibility cannot occur since $\setredb$ preserves $\stepa$, hence the name given to $\thefreewrt{\setredb}{a}$.
In turn, it is not true in general that $\stepb \in \thedominwrt{\setredb}{a}$ implies that $\stepb$ is \dominbytext\ $\set{\stepa}$, the exception being the case $\stepb = \stepa$; hence, the name ``\domintext'' is in fact approximate.} %
Observe $\setredb = \thefreewrt{\setredb}{a} \uplus \thedominwrt{\setredb}{a}$, and 
$\ppcstep{b}_1 \in \thefreewrt{\setredb}{a}$ and $\ppcstep{b}_2 \in \thedominwrt{\setredb}{a}$ imply $b_1 \disj b_2$.
\end{definition}

\begin{definition}
\label{dfn:projection-redseq}
Let $\reda$ be a \redseq, and $a \in \Pos{t}$ where $t = \rsrc{\reda}$, such that $a \leq \reda$.
We define the \textbf{projection} of $\reda$ \wrt\ $a$, notation $\proj{\reda}{a}$, as follows:
if \emph{$\reda = \emptyred{t}$}, then \emph{$\proj{\reda}{a} = \emptyred{\subtatwide{t}{a}}$}, otherwise $\rlength{\subtatnarrow{\reda}{a}} = \rlength{\reda}$ and 
%$\redel{\subtatnarrow{\reda}{a}}{i} = \subtat{\redel{\reda}{i}}{a}$ 
$\redel{\subtatnarrow{\reda}{a}}{i} = \pair{\subtat{t_i}{a}}{b}$ where $\redel{\reda}{i} = \pair{t_i}{ab}$, 
for all $i \leq \rlength{\reda}$.
\end{definition}

\begin{definition}
\label{dfn:projection-mstep}
If $\setredb \subseteq \ROccur{t}$ preserves $a \in \Pos{t}$, then we define the \textbf{projection of $\setredb$ \wrt\ $a$}, notation $\proj{\setredb}{a}$, as $\set{\pair{\subtat{t}{a}}{b'} \sthat \ppcstep{ab'} \in \setredb}$; if this set is empty, then $\proj{\setredb}{a} = \emptyset_{\subtatwide{t}{a}}$. Notice that $\proj{\setredb}{a} = \proj{\thedominwrt{\setredb}{a}}{a}$.
\end{definition}

\begin{definition}
\label{dfn:projection-mred}
If a \mredseq\ $\mreda$ preserves $a \in \Pos{\rsrc{\mreda}}$, then we define the \textbf{projection of $\mreda$ \wrt\ $a$}, notation $\proj{\mreda}{a}$, as follows:
\emph{$\proj{\emptyred{t}}{a} = \emptyred{\subtatwide{t}{a}}$}, and in any other case, 
$\proj{\mreda}{a} = \langle \proj{\mredel{\mreda}{1}}{a}; \ldots; \proj{\mredel{\mreda}{n}}{a}; \ldots \rangle$.
\end{definition}

\medskip
We prove that $\subtatnarrow{\reda}{a}$ is a well-defined \redseq\ (\theLem~\ref{rsl:proj-step-compatible-target}, along with a straightforward induction on $\rlength{\reda}$, suffices), and that targets (\thelem~\ref{rsl:proj-redseq-compatible-target}) and residuals (\thelem~\ref{rsl:proj-redseq-compatible-residuals}) are compatible with the projection of \redseqs.

\begin{lemma}
\label{rsl:proj-step-compatible-target}
Let $t \sstepx{\ppcstep{ab}} t'$. 
Then $\subtat{t}{a} \sstepx{\ppcstep{b}} \subtat{t'}{a}$.
\end{lemma}

\begin{proof}
Let $\subtat{t}{ab} = (\lpth p.s)u$ and $s' = \matchOpth{u}{p} s$. Then $t' = \repl{t}{s'}{ab}$. 
Observe 
$\subtat{(\subtat{t}{a})}{b} = \subtat{t}{ab}$ and $t' = \repl{t}{\repl{(\subtat{t}{a})}{s'}{b}}{a}$ implying $\subtat{t'}{a} = \repl{(\subtat{t}{a})}{s'}{b}$.
Thus we conclude.
\end{proof}

\begin{lemma}
\label{rsl:proj-redseq-compatible-target}
Let $a$ be a position and $t \sredx{\reda} t'$, such that $a \leq \reda$. 
Then $\subtat{t}{a} \sredx{\subtatwide{\reda}{a}} \subtat{t'}{a}$.
\end{lemma}

\begin{proof}
We proceed by induction on $\rlength{\reda}$. 
If $\reda = \nil_{t}$, then $t' = t$ and $\subtat{\reda}{a} = \nil_{\subtatwide{t}{a}}$, so we conclude.
Otherwise, $a \leq \reda$ implies $\reda = \ppcstep{ab} ; \reda'$, say $t \sstepx{\ppcstep{ab}} t'' \sredx{\reda'} t'$. 
Then \thelem~\ref{rsl:proj-step-compatible-target} and \ih\ imply $\subtat{t}{a} \sstepx{\ppcstep{b}} \subtat{t''}{a} \sredx{\subtatwide{\reda'}{a}} \subtat{t'}{a}$. Thus we conclude.
\end{proof}

\begin{lemma}
\label{rsl:proj-step-compatible-residuals}
Let $\ppcstep{ab}, \ppcstep{ac} \in \ROccur{t}$, so that $\ppcstep{b}, \ppcstep{c} \in \ROccur{\subtat{t}{a}}$. 
Then $\residu{\ppcstep{ac}}{\ppcstep{ab}}{\ppcstep{d}}$ iff $d = ad'$ and $\residu{\ppcstep{c}}{\ppcstep{b}}{\ppcstep{d'}}$.
\end{lemma}

\begin{proof}
Let $\subtat{t}{ab} = \subtat{(\subtat{t}{a})}{b} = (\lpth p.s) u$. 
In the analysis of $\residu{\ppcstep{ac}}{\ppcstep{ab}}{\ppcstep{d}}$ and $\residu{\ppcstep{c}}{\ppcstep{b}}{\ppcstep{d'}}$, \confer\ the definition of residuals for \theppc\ in page~\pageref{dfn:ppc-residual}, always the case applying is the same, and moreover with the same arguments. 
\Eg\ if $ab = ac2mn$, then $b = c2mn$, the values for $m$ and $n$ coincide. In this case, the subterms $p$ and $s$ also coincide. 
These observations suffice to conclude.
\end{proof}

\begin{lemma}
\label{rsl:proj-redseq-compatible-residuals}
Let $a$ be a position, $\ppcstep{ab} \in \ROccur{t}$, so that $\ppcstep{b} \in \ROccur{\subtat{t}{a}}$, and $\reda$ a \redseq\ verifying $\rsrc{\reda} = t$ and $a \leq \reda$. 
Then $\residu{\ppcstep{ab}}{\reda}{\ppcstep{d}}$ iff $d = ad'$ and $\residu{\ppcstep{b}}{\subtat{\reda}{a}}{\ppcstep{d'}}$.
\end{lemma}

\begin{proof}
We proceed by induction on $\rlength{\reda}$. 
If $\reda = \nil_{t}$, so that $\subtat{\reda}{a} = \nil_{\subtatwide{t}{a}}$, then $\residu{\ppcstep{ab}}{\reda}{\ppcstep{d}}$ implies $d = ab$, and $\residu{\ppcstep{b}}{\subtat{\reda}{a}}{\ppcstep{d'}}$ implies $d' = b$, thus we conclude.
Otherwise, $a \leq \reda$ implies $\reda = \ppcstep{ac} ; \reda'$, $a \leq \reda'$, and $\subtat{\reda}{a} = \ppcstep{c} ; \subtat{\reda'}{a}$. We proceed by double implication. Let us define $t' = \rsrc{\reda'}$.

\noindent
$\mathbf{\Longrightarrow )} \ \ $ 
$\residu{\ppcstep{ab}}{\reda}{\ppcstep{d}}$ implies $\residu{\ppcstep{ab}}{\ppcstep{ac}}{\ppcstep{e}}$ and $\residu{\ppcstep{e}}{\reda'}{\ppcstep{d}}$ for some $\ppcstep{e}$. 
\theLem~\ref{rsl:proj-step-compatible-residuals} implies $e = ae'$ and $\residu{\ppcstep{b}}{\ppcstep{c}}{\ppcstep{e'}}$. Observe that $\ppcstep{e} = \ppcstep{ae'} \in \ROccur{t'}$. Therefore \ih\ yields $d = ad'$ and $\residu{\ppcstep{e'}}{\subtat{\reda'}{a}}{\ppcstep{d'}}$, hence $\residu{\ppcstep{b}}{\subtat{\reda}{a}}{\ppcstep{d'}}$.

\noindent
$\mathbf{\Longleftarrow )} \ \ $ 
$\residu{\ppcstep{b}}{\subtat{\reda}{a}}{\ppcstep{d'}}$ implies $\residu{\ppcstep{b}}{\ppcstep{c}}{\ppcstep{e'}}$ and $\residu{\ppcstep{e'}}{\subtat{\reda'}{a}}{\ppcstep{d'}}$ for some $\ppcstep{e'}$.  
Let us call $e = ae'$ and $d = ad'$. Observe $\ppcstep{e'} \in \ROccur{\subtat{t'}{a}}$, \confer\ \thelem~\ref{rsl:proj-step-compatible-target}, then $\ppcstep{e} \in \ROccur{t'}$.
\theLem~\ref{rsl:proj-step-compatible-residuals} implies $\residu{\ppcstep{ab}}{\ppcstep{ac}}{\ppcstep{e}}$. In turn, \ih\ implies $\residu{\ppcstep{e}}{\reda'}{\ppcstep{d}}$. Thus we conclude.
\end{proof}

\medskip
We verify that if $a \leq \setredb$, then residuals (\thelem~\ref{rsl:proj-redexset-compatible-residuals}) and complete developments (\thelem~\ref{rsl:proj-redexset-compatible-development}) are compatible with the projection $\proj{\setredb}{a}$.

\begin{lemma}
\label{rsl:setred-residual-preserves-leq}
Let $a \leq \setredb$ and $\ppcstep{b} \in \setredb$. Then $a \leq \residus{\setredb}{\ppcstep{b}}$.
\end{lemma}

\begin{proof}
Hypotheses imply $b = ab'$.
For all $\ppcstep{c} \in \residus{\setredb}{\ppcstep{ab'}}$, \thelem~\ref{rsl:proj-step-compatible-residuals} implies $c = ac'$. Thus we conclude.
\end{proof}

\begin{lemma}
\label{rsl:development-preserves-leq}
Let $a \leq \setredb$ and $\reda \develops \setredb$. Then $a \leq \reda$.
\end{lemma}

\begin{proof}
We proceed by induction on $\depth{\setredb}$. Let $t \mstep{\setredb} t'$.
If $\setredb = \emptyset_{t}$ then $\reda = \nil_{t}$ and we conclude immediately.
Otherwise $\setredb = \ppcstep{b} ; \reda'$ where $\ppcstep{b} \in \setredb$, implying $a \leq b$, and $\reda' \develops \residus{\setredb}{\ppcstep{b}}$.  \theLem~\ref{rsl:setred-residual-preserves-leq} implies $a \leq \residus{\setredb}{\ppcstep{b}}$. Hence \ih\ yields $a \leq \reda'$, which suffices to conclude.
\end{proof}

\begin{lemma}
\label{rsl:proj-redexset-compatible-residuals}
Let $a \leq \setredb$ and $\ppcstep{ab} \in \setredb$. Then $\subtat{(\residus{\setredb}{\ppcstep{ab}})}{a} = \residus{\subtatnarrow{\setredb}{a}}{\ppcstep{b}}$.
\end{lemma}

\begin{proof}
By double inclusion. 

\noindent
$\mathbf{\supseteq )} \ \ $ 
Let $\ppcstep{c} \in \subtat{(\residus{\setredb}{\ppcstep{ab}})}{a}$, so that $\ppcstep{ac} \in \residus{\setredb}{\ppcstep{ab}}$. 
Let $\ppcstep{ad} \in \setredb$ such that $\residus{\ppcstep{ad}}{\ppcstep{ab}}{\ppcstep{ac}}$, observe $\ppcstep{d} \in \subtat{\setredb}{a}$. \theLem~\ref{rsl:proj-step-compatible-residuals} implies $\residus{\ppcstep{d}}{\ppcstep{b}}{\ppcstep{c}}$. Hence $\ppcstep{c} \in \residus{\subtatnarrow{\setredb}{a}}{\ppcstep{b}}$.

\noindent
$\mathbf{\subseteq )} \ \ $ 
Let $\ppcstep{c} \in \residus{\subtatnarrow{\setredb}{a}}{\ppcstep{b}}$, let $\ppcstep{d} \in \subtat{\setredb}{a}$ such that $\residu{\ppcstep{d}}{\ppcstep{b}}{\ppcstep{c}}$, observe that $\ppcstep{ad} \in \setredb$. 
\theLem~\ref{rsl:proj-step-compatible-residuals} implies $\residu{\ppcstep{ad}}{\ppcstep{ab}}{\ppcstep{ac}}$. Then $\ppcstep{ac} \in \residus{\setredb}{\ppcstep{ab}}$, implying $\ppcstep{c} \in \subtat{(\residus{\setredb}{\ppcstep{ab}})}{a}$.
\end{proof}

\begin{lemma}
\label{rsl:proj-redexset-compatible-development}
Let $a \leq \setredb$ and $\reda \develops \setredb$. 
Then $\subtat{\reda}{a} \develops \subtat{\setredb}{a}$.
\end{lemma}

\begin{proof}
By induction on $\depth{\setredb}$. Let $t = \rsrc{\setredb}$. 
If $\setredb = \emptyset_{t}$ then observing $\reda = \nil_t$ suffices to conclude.
Otherwise $\reda = \ppcstep{ab} ; \reda'$ where $\reda' \develops \residus{\setredb}{\ppcstep{ab}}$.
In this case, $\subtat{\reda}{a} = \ppcstep{b} ; \subtat{\reda'}{a}$. \ih\ yields $\subtat{\reda'}{a} \,\develops \subtat{(\residus{\setredb}{\ppcstep{ab}})}{a}$. In turn, \thelem~\ref{rsl:proj-redexset-compatible-residuals} implies $\subtat{(\residus{\setredb}{\ppcstep{ab}})}{a} = \residus{\subtat{\setredb}{a}}{\ppcstep{b}}$. Hence $\subtat{\reda}{a} \develops \subtat{\setredb}{a}$.
\end{proof}

\medskip
We verify that given a \tredexset\ $t \mstep{\setredb} t'$ s.t. $\setredb$ preserves $a$, it is only the \domintext\ part of $\setredb$ that actually modifies $\subtat{t}{a}$; \confer\ \thelem~\ref{rsl:subtat-depends-only-on-dominated-part}. 

\begin{lemma}
\label{rsl:free-domin-invariant-residuals}
Let $a, \setredb$ such that $\setredb$ preserves $a$, and $\ppcstep{b} \in \setredb$. Then $\residus{\setredb}{\ppcstep{b}}$ preserves $a$. Moreover $\thefreewrt{\residus{\setredb}{\ppcstep{b}}}{a} = \residus{\thefreewrt{\setredb}{a}}{\ppcstep{b}}$ and $\thedominwrt{\residus{\setredb}{\ppcstep{b}}}{a} = \residus{\thedominwrt{\setredb}{a}}{\ppcstep{b}}$.
\end{lemma}

\begin{proof}
Take $\ppcstep{b'_1} \in \residus{\setredb}{\ppcstep{b}}$ and let $\ppcstep{b_1} \in \setredb$ such that $\residu{\ppcstep{b_1}}{\ppcstep{b}}{\ppcstep{b'_1}}$. 
Observe that either $b \leq b'_1$ (if $b < b_1$), or $b'_1 = b_1$ (if $b \not\leq b_1$).
We verify that $b'_1 \not < a$. $\setredb$ preserves $a$ implies $a \leq b$ or $a \disj b$, and analogously for $b_1$.
\begin{itemize}
\minitem 
Assume $a \leq b$. 
If $a \disj b_1$ then $b'_1 = b_1$ implying $a \disj b'_1$.
If $a \leq b_1$, then either $b'_1 = b_1$ or $b \leq b'_1$ imply $a \leq b'_1$.
\minitem
Assume $a \disj b$.
If $a \disj b_1$ then either $b'_1 = b_1$ or $b \leq b'_1$ imply $a \disj b'_1$.
If $a \leq b_1$, so that $b \disj b_1$, then $b'_1 = b_1$, implying $a \leq b'_1$.
\end{itemize}
Consequently, $\residus{\setredb}{\ppcstep{b}}$ preserves $a$. 
Furthermore, $a \disj b_1$ implies $a \disj b'_1$ and $a \leq b_1$ implies $a \leq b'_1$.
The former assertion implies $\residus{\thefreewrt{\setredb}{a}}{\ppcstep{b}} \subseteq \thefreewrt{\residus{\setredb}{\ppcstep{b}}}{a}$. 
Moreover, let $\ppcstep{b'_2} \in \thefreewrt{\residus{\setredb}{\ppcstep{b}}}{a}$ and $\ppcstep{b_2} \in \setredb$ such that $\residu{\ppcstep{b_2}}{\ppcstep{b}}{\ppcstep{b'_2}}$. Observe that $a \leq b_2$ would imply $a \leq b'_2$, therefore $\setredb$ preserves $a$ implies $a \disj b_2$, \ie\ $\ppcstep{b_2} \in \thefreewrt{\setredb}{a}$.
Therefore $\thefreewrt{\residus{\setredb}{\ppcstep{b}}}{a} \subseteq \residus{\thefreewrt{\setredb}{a}}{\ppcstep{b}}$, so that we obtain $\thefreewrt{\residus{\setredb}{\ppcstep{b}}}{a} = \residus{\thefreewrt{\setredb}{a}}{\ppcstep{b}}$. An analogous argument on the \domintext\ parts allows to conclude.
\end{proof}

\begin{lemma}
\label{rsl:subtat-depends-only-on-dominated-part}
Let $\setredb \in \ROccur{t}$ and assume  $\setredb$ preserves $a$ and $t \mstep{\thedominwrt{\setredb}{a}} t'' \mstep{\residus{\thefreewrt{\setredb}{a}}{\thedominwrt{\setredb}{a}}} t'$.
Then $\subtat{t'}{a} = \subtat{t''}{a}$.
\end{lemma}

\begin{proof}
A simple induction based on \thelem~\ref{rsl:free-domin-invariant-residuals} yields that $b \disj a$ if $\ppcstep{b} \in \residus{\thefreewrt{\setredb}{a}}{\thedominwrt{\setredb}{a}}$. Therefore, a straightforward analysis allows to conclude.
\end{proof}

\medskip
\theLem~\ref{rsl:subtat-depends-only-on-dominated-part} allows to verify that targets and residuals are compatible with the projection $\proj{\setredb}{a}$.

\begin{lemma}
\label{rsl:mstep-projection-good-behavior}
Let $t \mstep{\setredb} t'$ and assume $\setredb$ preserves $a$.
Then: 
\begin{enumerate}[(i)]
\minitem 
\label{it:mstep-projection-good-tgt}
$\subtat{t}{a} \mstep{\subtatwide{\setredb}{a}} \ \subtat{t'}{a}$.
\minitem 
\label{it:mstep-projection-good-residuals}
If $\ppcstep{ac} \in \ROccur{t}$, so that $\ppcstep{c} \in \ROccur{\subtat{t}{a}}$, then $\residu{\ppcstep{ac}}{\setredb}{\ppcstep{d}}$ iff $d = a d'$ and $\residu{\ppcstep{c}}{\subtat{\setredb}{a}}{\ppcstep{d'}}$.
\end{enumerate}
\end{lemma}

\begin{proof}
Let $t \mstep{\thedominwrt{\setredb}{a}} t'' \mstep{\residus{\thefreewrt{\setredb}{a}}{\thedominwrt{\setredb}{a}}} t'$. 
Let $\reda$ such that $\reda \develops \thedominwrt{\setredb}{a}$, and $\redb \develops \residus{\thefreewrt{\setredb}{a}}{\thedominwrt{\setredb}{a}}$. Observe $t \sredx{\reda} t'' \sredx{\redb} t'$.
Moreover, $a \leq \reda$ and $\subtat{\reda}{a} \develops \subtat{\thedominwrt{\setredb}{a}}{a} = \subtat{\setredb}{a}$, by \thelem~\ref{rsl:development-preserves-leq} and \theLem~\ref{rsl:proj-redexset-compatible-development} respectively.  
On the other hand, $b \disj a$ for all $\ppcstep{b} \in \residus{\thefreewrt{\setredb}{a}}{\thedominwrt{\setredb}{a}}$ implies $a \disj \redel{\redb}{i}$ for all $i$. Notice that $a \disj b \land a \disj c$ implies $a \disj d$ whenever $\residu{\ppcstep{b}}{\ppcstep{c}}{\ppcstep{d}}$.

To prove item~(\ref{it:mstep-projection-good-tgt}), it suffices to observe that \thelem~\ref{rsl:proj-redseq-compatible-target} implies $\subtat{t}{a} \sredx{\subtatwide{\reda}{a}} \subtat{t''}{a} = \subtat{t'}{a}$; \confer\ \thelem~\ref{rsl:subtat-depends-only-on-dominated-part}.  

We prove item~(\ref{it:mstep-projection-good-residuals}), by double implication.

\noindent
$\mathbf{\Longrightarrow )} \ \ $ 
Let $\residu{\ppcstep{ac}}{\setredb}{\ppcstep{d}}$. Then $\residu{\ppcstep{ac}}{\reda}{\ppcstep{e}}$ and $\residu{\ppcstep{e}}{\redb}{\ppcstep{d}}$ for some $\ppcstep{e}$.
\theLem~\ref{rsl:proj-redseq-compatible-residuals} implies $e = ae'$ and $\residu{\ppcstep{c}}{\subtat{\reda}{a}}{\ppcstep{e'}}$. In turn, $a \disj \redel{\redb}{i}$ for all $i$ and $a \leq e$ imply $d = e$, \ie\ $d = ad'$ where $d' = e'$, and $\residu{\ppcstep{c}}{\subtat{\reda}{a}}{\ppcstep{d'}}$.
We conclude by recalling that $\subtat{\reda}{a} \develops \subtat{\setredb}{a}$.

\noindent
$\mathbf{\Longleftarrow )} \ \ $ 
Let $\residu{\ppcstep{c}}{\subtat{\setredb}{a}}{\ppcstep{d'}}$, and $d = ad'$. Then $\residu{\ppcstep{c}}{\subtat{\reda}{a}}{\ppcstep{d'}}$. 
\theLem~\ref{rsl:proj-redseq-compatible-residuals} implies 
$\residu{\ppcstep{ac}}{\reda}{\pair{t''}{d}}$. 
In turn, $a \disj \redel{\redb}{i}$ for all $i$ and $a \leq d$ imply $\residu{\pair{t''}{d}}{\redb}{\pair{t'}{d}}$.
Hence $\residu{\ppcstep{ac}}{\setredb}{\ppcstep{d}}$
\end{proof}

\medskip
Now consider a \emph{\mredseq} $\mreda$ which preserves some position $a$. For any $n < \rlength{\mreda}$, \thelem~\ref{rsl:mstep-projection-good-behavior}:(\ref{it:mstep-projection-good-tgt}) implies that 
$\rsrc{\subtat{\mredel{\mreda}{n+1}}{a}} = \subtat{\rsrc{\mredel{\mreda}{n+1}}}{a} = \rtgt{\subtat{\mredel{\mreda}{n}}{a}}$.
This implies that the definition of the projection of $\mreda$ over $a$ is well-defined.

\medskip
We finish this section by giving a proof of \thelem~\ref{rsl:mred-projection-good-behavior}. We recall the statement:

\medskip\noindent
Let $t \mred{\mreda} t'$ and assume $\mreda$ preserves $a$.
Then: 
\begin{enumerate}[(i)]
\minitem 
$\subtat{t}{a} \mred{\subtatwide{\mreda}{a}} \ \subtat{t'}{a}$.
\minitem 
If $\ppcstep{ac} \in \ROccur{t}$, then $\residu{\ppcstep{ac}}{\mreda}{\ppcstep{d}}$ iff $d = a d_1$ and $\residu{\ppcstep{c}}{\subtat{\mreda}{a}}{\ppcstep{d_1}}$.
\minitem 
If $\ppcstep{ac} \in \ROccur{t}$, then $\mreda$ uses $\ppcstep{ac}$ iff $\subtat{\mreda}{a}$ uses $\ppcstep{c}$.
\end{enumerate}

\begin{proof}
To prove item~(\ref{it:mred-projection-good-tgt}) a simple induction on $\rlength{\mreda}$, resorting on \thelem~\ref{rsl:mstep-projection-good-behavior}:(\ref{it:mstep-projection-good-tgt}), suffices.

Item~(\ref{it:mred-projection-good-residuals}) admits an argument similar to the one used to prove \thelem~\ref{rsl:proj-redseq-compatible-residuals}, resorting on \thelem~\ref{rsl:mstep-projection-good-behavior}:(\ref{it:mstep-projection-good-tgt}) instead of \thelem~\ref{rsl:proj-step-compatible-residuals}.

We prove item~(\ref{it:mred-projection-good-uses}).
Assume $\subtat{\mreda}{a}$ uses $\ppcstep{c}$, \ie\ $\mreda = \mreda_1 ; \setredd ; \mreda_2$ and there exists some $\ppcstep{d} \in \subtat{\setredd}{a} \cap \ \residus{\ppcstep{c}}{\subtat{\mreda_1}{a}}$.
Item~(\ref{it:mred-projection-good-residuals}) implies $\residu{\ppcstep{ac}}{\mreda_1}{\ppcstep{ad}}$, and moreover $\ppcstep{d} \in \subtat{\setredd}{a}$ implies $\ppcstep{ad} \in \setredd$. Hence $\mreda$ uses $\ppcstep{ac}$.

Assume $\mreda$ uses $\ppcstep{ac}$, \ie\ $\mreda = \mreda_1 ; \setredd ; \mreda_2$ and there exists some $\ppcstep{d} \in \setredd \cap \ \residus{\ppcstep{ac}}{\mreda_1}$.
Item~(\ref{it:mred-projection-good-residuals}) implies $d = ad'$, so that $\ppcstep{d'} \in \subtat{\setredd}{a}$, and $\residu{\ppcstep{c}}{\subtat{\mreda_1}{a}}{\ppcstep{d'}}$. On the other hand, $\subtat{\mreda}{a} = \subtat{\mreda_1}{a} ; \subtat{\setredd}{a} ; \subtat{\mreda_2}{a}$. Hence $\subtat{\mreda}{a}$ uses $\ppcstep{c}$.
\end{proof}

}

\ignore{
\section{Appendix -- extended strategies}
\label{sec:extended-strategies}
In this section we will state a result about normalisation of a given reduction strategy \strt, given the existence for the same ARS of a normalising strategy \strs\ such that \strt\ \emph{extends} \strs, \ie, the set of \tredexes\ selected by \strt\ is always a superset of that selected by \strs\ for the same term. Formally:

\begin{definition}
\label{dfn:extends-strategy}
Let \strs, \strt\ be two reduction strategies for an ARS $\arsa$. We will say that \strt\ \emph{extends} \strs, notation $\strs \subseteq \strt$, iff $\strs(t) \subseteq \strt(t)$ for any term $t$ not in normal form.
\end{definition}

Some technicalities about \redseqs\ which \emph{follow} a reduction strategy.

\begin{definition}
\label{dfn:redseq-follows-strategy}
Let $t \mred{\mreda} u$ be a \mredseq\ and \strs\ a reduction strategy. We will say that $\mreda$ \emph{follows} \strs\ iff either $\mreda = \nil_{t}$ or $\mreda = \strs(t); \mreda'$ where $\mreda'$ follows \strs.
\end{definition}

\begin{lemma}
\label{rsl:follows-basic-facts}
Let \strs\ be a reduction strategy. Then: 
\begin{enumerate}[a)]
\item \label{rsl:follows-unique}
for all $t$ term and $n \in \Nat$, there is at most one \mredseq\ $\mreda$ such that $\mreda$ follows \strs, $\rsrc{\mreda} = t$ and $\rlength{\mreda} = n$.
\item \label{rsl:follows-normal-form}
for all $t$ term and $n \in \Nat$, if there is no \mredseq\ $\mreda$ such that $\mreda$ follows \strs, $\rsrc{\mreda} = t$ and $\rlength{\mreda} = n$, then there exists a \mredseq\ $\mreda'$ such that $\rlength{\mreda'} < n$, $t \mred{\mreda'} u$ and $u$ is a normal form.
\item \label{rsl:follows-parts}
If $\mreda_1 ; \mreda_2$ follows \strs, then both $\mreda_1$ and $\mreda_2$ follow \strs.
\item \label{rsl:follows-completion}
If $t \mred{\mreda} s$ and $t \mred{\mredb} u$ follow \strs\ and $\rlength{\mreda} \leq \rlength{\mredb}$, then $s \mred{\mredc} u$ where $\mredb = \mreda ; \mredc$.
\end{enumerate}
\end{lemma}

\begin{proof}
Items (\ref{rsl:follows-unique}) and (\ref{rsl:follows-normal-form}) 
admit simple proofs by induction on $n$. 
Item (\ref{rsl:follows-parts}) can be easily proved by induction on $\rlength{\mreda_1}$.
We prove item (\ref{rsl:follows-completion}). Consider $\mreda', \mredc'$ such that $\mredb = \mreda' ; \mredc'$ and $\rlength{\mreda'} = \rlength{\mreda}$, say $t \mred{\mreda'} s' \mred{\mredc'} u$. Item (\ref{rsl:follows-parts}) implies that both $\mreda'$ and $\mredc'$ follow \strs, in turn item (\ref{rsl:follows-unique}) implies $\mreda' = \mreda$, and therefore $s' = s$. We conclude by taking $\mredc \eqdef \mredc'$.
\end{proof}

The aim of the following material is to prove that any strategy $\strt$ extending a normalising strategy $\strs$ is also normalising.
The forthcoming proof requires the strategy $\strs$ to enjoy an additional property, namely:

\begin{definition}
A reduction strategy $\strs$ is said \emph{\tRecatch} iff for any term $t$ and $\setreda \subseteq \ROccur{t}$, $t \mstep{\strs(t)} s$ and $t \mstep{\setreda} u \mstep{\strs(u)} u'$ imply $s \mred{\mredc} u'$ for some \mredseq\ $\mredc$. 
Graphically: \\[2mm]
\minicenter{$
\xymatrix@R=20pt@C=60pt{
  t \arMulti{d}_{\strs(t)} \arMulti{r}^{\setreda} & 
	u \arMulti{d}^{\strs(u)}  \\ 
  s \arMultiOp{r}{-->>}_{\mredc} & u'
}
$}
\end{definition}

Further ``recatching'' properties of \tRecatchStrs\ admit simple, abstract proofs.

\begin{lemma}
\label{rsl:recatch-mredseq}
Let \strs\ be a \tRecatchStr, $t \mred{\mreda} u$, $t \mstep{\strs(t)} s$ and $u \mstep{\strs(u)} u'$. Then $s \mred{\mredc} u'$ for some $\mredc$. 
Graphically: \\[2mm]
\minicenter{$
\xymatrix@R=20pt@C=60pt{
  t \arMulti{d}_{\strs(t)} \arMultiOp{r}{>>}^{\mreda} & 
	u \arMulti{d}^{\strs(u)}  \\ 
  s \arMultiOp{r}{-->>}_{\mredc} & u'
}
$}
\end{lemma}

\begin{proof}
We proceed by induction on $\rlength{\mreda}$.
If $\mreda = \nil_t$, then $u = t$ and therefore $u' = s$, thus it suffices to take $\mredc \eqdef \nil_s$.

Assume $\mreda = \setreda ; \mreda'$, say $t \mstep{\setreda} u_1 \mred{\mreda'} u$ and $u_1 \mstep{\strs(u_1)} u'_1$.
We can build the following diagram: \\[2mm]
\minicenter{$
\xymatrix@R=20pt@C=60pt{
  t \arMulti{d}_{\strs(t)} \arMulti{r}^{\setreda} & 
  u_1 \arMulti{d}_{\strs(u_1)} \arMultiOp{r}{>>}^{\mreda'} & 
	u \arMulti{d}^{\strs(u)}  \\ 
  s \arMultiOp{r}{>>}_{\mredc_1} &
  u'_1 \arMultiOp{r}{>>}_{\mredc_2} & u'
}
$} \\
where $\strs$ being \tRecatch\ and \ih\ entail the existence of $\mredc_1$ and $\mredc_2$ respectively. Thus it suffices to take $\mredc \eqdef \mredc_1 ; \mredc_2$.
\end{proof}

\begin{lemma}
\label{rsl:recatch-mredseq-mredseq}
Let \strs\ be \tRecatchStr, $t \mred{\mreda}s$ where $\mreda$ follows \strs, $t \mred{\mredc} u$, and $u \mred{\mredb} u'$ where $\mredb$ follows \strs, and either $\rlength{\mredb} = \rlength{\mreda}$ or $\rlength{\mredb} < \rlength{\mreda}$ and $u'$ is a normal form.
Then $s \mred{\mredc'} u'$ for some $\mredc'$. 
Graphically: \\[2mm]
\minicenter{$
\xymatrix@R=30pt@C=60pt{
  t \arMultiOp{d}{>>}_{\mreda} \arMultiOp{r}{>>}^{\mredc} & 
	u \arMultiOp{d}{>>}^{\ \mredb}  \\ 
  s \arMultiOp{r}{-->>}_{\mredc'} & u'
}
$}
\end{lemma}

\begin{proof}
We proceed by induction on $\rlength{\mreda}$.

Assume $\mreda = \nil_t$, implying $\mredb = \nil_u$ and therefore $s = t$ and $u' = u$. In this case it suffices to take $\mredc' \eqdef \mredc$.

Assume that $u'$ is a normal form. In this case, the property CR entails the existence of $\mredc'$, thus we conclude.

Assume that $\mreda = \strs(t) ; \mreda'$ and that $u'$ is not a normal form, implying that $\rlength{\mredb} = \rlength{\mreda}$. Observe that $\mredb = \strs(u) ; \mredb'$ where $\rlength{\mredb'} = \rlength{\mreda'}$. 
In this case, we can build the following diagram:\\[2mm]
\minicenter{$
\xymatrix@R=30pt@C=60pt{
  t \arMulti{d}_{\strs(t)} \arMultiOp{r}{>>}^{\mredc} & 
	u \arMulti{d}^{\ \strs(u)}  \\ 
  t_1 \arMultiOp{d}{>>}_{\mreda'} \arMultiOp{r}{>>}^{\mredc_1} & 
	u_1 \arMultiOp{d}{>>}^{\ \mredb'}  \\ 
  s \arMultiOp{r}{>>}_{\mredc'} & u'
}
$} \\
for some terms $t_1$ and $u_1$. \theLem~\ref{rsl:recatch-mredseq} and \ih\ entail the existence of $\mredc_1$ and $\mredc'$ respectively. Thus we conclude.
\end{proof}

Given the just stated properties, we can prove that any strategy extending a normalising \emph{and \tRecatch} strategy, is normalising as well.

\begin{lemma}
\label{rsl:extends-then-catch-nf}
Let $\strs$ and $\strt$ reduction strategies, such that $\strs$ is \tRecatch\ and $\strs \subseteq \strt$. Let $t \mred{\mreda} u$ where $\mreda$ follows $\strs$ and $u$ is a normal form. 
Then there exists a \mredseq\ $\mredc$ such that $\mredc$ follows $\strt$, $t \mred{\mredc} u$, and $\rlength{\mredc} \leq \rlength{\mreda}$.
\end{lemma}

\begin{proof}
We proceed by induction in $\rlength{\mreda}$.

Assume that $\mreda = \nil_t$, implying that $t = u$. In this case it suffices to take $\mredc \eqdef \nil_t$.

Assume that $\mreda = \strs(t) ; \mreda'$. 
In this case, $\strs(t) \subseteq \strt(t)$ implies that $\residus{\strs(t)}{\strt(t)} = \emptyset$, and therefore \theprop~\ref{rsl:SOredexset} implies that the targets of $\strs(t) ; \residus{\strt(t)}{\strs(t)}$ and $\strt(t)$ coincide. 
Hence we can build the following diagram: \\[4pt]
\minicenter{$
\xymatrix@R=20pt@C=60pt{
  & 
	u_1 \arMultiOp{rr}{>>}^{\mreda'} \arMulti{dd}^{\ \residus{\strt(t)}{\strs(t)}} & &
	u \\
	t \arMulti{ru}^{\strs(t)} \arMulti{rd}_{\strt(t)} \\
	&
	u'_1 \arMultiOp{rr}{>>}^{\mreda''} & &
	u'
}
$} \\[4pt]
where $\mreda''$ is the only (\confer\ \thelem~\ref{rsl:follows-basic-facts}) \mredseq\ verifying $\rsrc{\mreda''} = u'_1$, $\mreda''$ follows \strs, and either $\rlength{\mreda''} = \rlength{\mreda'}$ or $\rlength{\mreda''} < \rlength{\mreda'}$ and $u'$ is a normal form.
Then \thelem~\ref{rsl:recatch-mredseq-mredseq} implies that $u \mred{} u'$, implying that $u' = u$ since $u$ is a normal form.
In turn, \ih\ can be applied on $u'_1 \mred{\mreda''} u$, yielding the existence of a \mredseq\ $\mredc'$ such that $\mredc'$ follows $\strt$, $\rlength{\mredc'} \leq \rlength{\mreda''} \leq \rlength{\mreda'}$, and $u'_1 \mred{\mredc'} u$.
Therefore we can complete the preceding diagram as follows:
\\[4pt]
\minicenter{$
\xymatrix@R=20pt@C=60pt{
  & 
	u_1 \arMultiOp{rr}{>>}^{\mreda'} \arMulti{dd}^{\ \residus{\strt(t)}{\strs(t)}} & &
	u \ar@{=}[dd] \\
	t \arMulti{ru}^{\strs(t)} \arMulti{rd}_{\strt(t)} \\
	&
	u'_1 \arMultiOp{rr}{>>}^{\mreda''} \arMultiOp{rrd}{>>}_{\mredc'} & &
	u \ar@{=}[d] \\
	& & & u
}
$} \\[4pt]
We conclude by taking $\mredc \eqdef \strt(t) ; \mredc'$.
\end{proof}

\begin{proposition}
\label{rsl:extends-catching-normalising-then-normalising}
Let \strs, \strt\ be reduction strategies such that \strs\ is \tRecatch\ and normalising, and $\strs \subseteq \strt$. Then \strt\ is normalising.
\end{proposition}

\begin{proof}
Let $t$ be a term not being a normal form. Then \strs\ normalising implies that $t \mred{\mreda} u$ where $u$ is a normal form and $\mreda$ follows \strs.
In turn, \thelem~\ref{rsl:extends-then-catch-nf} implies $t \mred{\mredc} u$ where $\mredc$ follows \strt. Thus we conclude.
\end{proof}

\bigskip\noindent
\textbf{Why \theprop~\ref{rsl:extends-catching-normalising-then-normalising} could be interesting} \\
It gives an additional criterion for a reduction strategy to be normalising. 
Particularly, for the strategy \strs\ we proved normalising for the PPC, we know that $\strs \subseteq PO$ where $PO$ stands for parallel-outermost; \confer\ \thelem~2 in the RTA 2012 paper.
Therefore, if we can prove that \strs\ is \tRecatch, then we obtain using \theprop~\ref{rsl:extends-catching-normalising-then-normalising} that $PO$ is normalising for PPC, based on the fact that \strs\ is normalising as well.
In this way, we would provide a proof of normalisation of $PO$ for a non-sequential calculus using an approach different from the one described in~\cite{vanOostrom:1999}.

\bigskip
We notice that a generalisation of \thelem~\ref{rsl:extends-then-catch-nf} can be stated.

\begin{lemma}
\label{rsl:extends-then-catch}
Let $\strs$ and $\strt$ reduction strategies, such that $\strs$ is \tRecatch\ and $\strs \subseteq \strt$. Let $t \mred{\mreda} s$ where $\mreda$ follows $\strs$, and $t \mred{\mredc} u$ where $\mredc$ follows \strt\ and either $\rlength{\mredc} = \rlength{\mreda}$, or $\rlength{\mredc} < \rlength{\mreda}$ and $u$ is a normal form. 
Then there exists a \mredseq\ $\mredb$ such that $s \mred{\mredb} u$.
\end{lemma}

\begin{proof}
We proceed by induction on $\rlength{\mreda}$.

Assume $\mreda = \nil_t$, so that $s = t$. In this case $\mredc = \nil_t$, so that $u = t = s$. Therefore it is enoughto take $\mredb \eqdef \nil_t$.

Assume $\mreda = \strs(t) ; \mreda'$. In this case $t$ is not a normal form, implying $\mredc = \strt(t) ; \mredc'$. Observe that $\strs(t) \subseteq \strt(t)$ implies that $\residus{\strs(t)}{\strt(t)} = \emptyset$, and therefore \theprop~\ref{rsl:SOredexset} implies that the targets of $\strs(t) ; \residus{\strt(t)}{\strs(t)}$ and $\strt(t)$ coincide. 
Hence we can build the following diagram: \\[4pt]
\minicenter{$
\xymatrix@R=20pt@C=60pt{
  & 
	s_1 \arMultiOp{rr}{>>}^{\mreda'} \arMulti{dd}^{\ \residus{\strt(t)}{\strs(t)}} & &
	s \arMultiOp{dd}{>>}^{\ \mredb_1} \\
	t \arMulti{ru}^{\strs(t)} \arMulti{rd}_{\strt(t)} \\
	&
	u_1 \arMultiOp{rr}{>>}^{\mreda''} \arMultiOp{rrd}{>>}_{\mredc'} & &
	u' \\
	& & & u
}
$} \\[4pt]
where $\mreda''$ is the only (\confer\ \thelem~\ref{rsl:follows-basic-facts}) \mredseq\ verifying $\rsrc{\mreda''} = u_1$, $\mreda''$ follows \strs, and either $\rlength{\mreda''} = \rlength{\mreda'}$ or $\rlength{\mreda''} < \rlength{\mreda'}$ and $u'$ is a normal form. \theLem~\ref{rsl:recatch-mredseq-mredseq} entails the existence of $\mredb_1$.

If $u'$ is a normal form, then \thelem~\ref{rsl:extends-then-catch-nf} implies that $u_1 \mred{\mredc''} u'$ where $\mredc''$ follows \strt\ and $\rlength{\mredc''} \leq \rlength{\mreda''} \leq \rlength{\mreda'}$.
Observe that $\rlength{\mredc'} < \rlength{\mredc''} \leq \rlength{\mreda'}$ would imply $u$ to be a normal form by lemma hypotheses, and also that $u \mred{\mredd} u'$ where $\rlength{\mredd} > 0$ by \thelem~\ref{rsl:follows-basic-facts}, \ie\ a contradiction.
On the other hand, $\rlength{\mredc''} < \rlength{\mredc'}$ would imply $u' \mred{\mredd} u$ where $\rlength{\mredd} > 0$, again  by \thelem~\ref{rsl:follows-basic-facts}, contradicting the assumption that $u'$ is a normal form.
Hence $\rlength{\mredc''} = \rlength{\mredc'}$, so that \thelem~\ref{rsl:follows-basic-facts} implies $\mredc''= \mredc'$ and consequently $u' = u$. Thus we conclude by taking $\mredb \eqdef \mredb_1$.

If $u'$ is not a normal form, then $\rlength{\mreda''} = \rlength{\mreda'}$, implying that \ih\ can be applied on $u_1 \mred{\mreda''} u'$ and $u_1 \mred{\mredc'} u$, entailing $u' \mred{\mredb_2} u$ for some \mredseq\ $\mredb_2$. In this case, we conclude by taking $\mredb \eqdef \mredb_1 ; \mredb_2$.

The two cases considered, $u'$ being and not being a normal form respectively, can be depicted as follows: \\[4pt]
$\begin{array}{@{}l@{\quad}l}
\xymatrix@R=20pt@C=30pt{
  & 
	s_1 \arMultiOp{rr}{>>}^{\mreda'} \arMulti{dd}^{\ \residus{\strt(t)}{\strs(t)}} & &
	s \arMultiOp{dd}{>>}^
	{\ 
	\begin{array}{@{}l} 
	  \textnormal{\scriptsize{$\mredb_1$}} \\ 
	  \textnormal{\scriptsize{\thelem~\ref{rsl:recatch-mredseq-mredseq}}} 
	\end{array}} 
	\\
	t \arMulti{ru}^{\strs(t)} \arMulti{rd}_{\strt(t)} \\
	&
	u_1 \arMultiOp{rr}{>>}^{\mreda''} \arMultiOp{rrd}{>>}_{\mredc'} & &
	u' \ar@{=}[d]^ 
	{\ 
	\begin{array}{@{}l} 
	  \textnormal{\scriptsize{\thelem~\ref{rsl:extends-then-catch-nf}}} \\
	  \textnormal{\scriptsize{\thelem~\ref{rsl:follows-basic-facts}}} 
	\end{array}} 
	\\
	& & & u
}
& 
\xymatrix@R=20pt@C=30pt{
  & 
	s_1 \arMultiOp{rr}{>>}^{\mreda'} \arMulti{dd}^{\ \residus{\strt(t)}{\strs(t)}} & &
	s \arMultiOp{dd}{>>}^
	{\ 
	\begin{array}{@{}l} 
	  \textnormal{\scriptsize{$\mredb_1$}} \\ 
	  \textnormal{\scriptsize{\thelem~\ref{rsl:recatch-mredseq-mredseq}}} 
	\end{array}} 
	\\
	t \arMulti{ru}^{\strs(t)} \arMulti{rd}_{\strt(t)} \\
	&
	u_1 \arMultiOp{rr}{>>}^{\mreda''} \arMultiOp{rrd}{>>}_{\mredc'} & &
	u' \arMultiOp{d}{>>}^
	{\ 
	\begin{array}{@{}l} 
	  \textnormal{\scriptsize{$\mredb_2$}} \\ 
	  \textnormal{\scriptsize{\ih}} 
	\end{array}} 
	\\
	& & & u
}
\end{array}
$ \\
\end{proof}

} % ignore

%\end{document}

%%% Local Variables:
%%% mode: latex
%%% TeX-master: "paper"
%%% End:

%\newpage
%%   before submitting, adapt the bibliographic style to the settings adequate for 
%%   the TCS journal
\bibliographystyle{alpha}
\bibliography{ppc-normalisation}

\end{document}